\newcommand*{\QUANTUM}{}
\ifdefined\QUANTUM
\documentclass[a4paper,onecolumn,11pt,accepted=2022-10-03]{quantumarticle}
\pdfoutput=1
\else
\documentclass[11pt]{article}
\fi

\usepackage{color, xcolor, colortbl}
\usepackage{graphicx,epstopdf}
\graphicspath{{./fig/}}
\usepackage{amsmath,amssymb,amsthm}
\usepackage{bm}
\usepackage{geometry}
\usepackage{appendix}
\usepackage{multirow}
\usepackage{braket}
\usepackage[english]{babel}
\usepackage{hyperref}
\usepackage[capitalize]{cleveref}
\usepackage{caption}
\usepackage{subcaption}
\usepackage{enumerate}
\usepackage{enumitem}
\crefrangeformat{enumi}{#3#1#4--#5#2#6}

\usepackage{algorithm}
\usepackage{algorithmic}

\usepackage[normalem]{ulem}

\newtheorem{theorem}{Theorem}
\newtheorem{definition}[theorem]{Definition}

\newtheorem{remark}[theorem]{Remark}
\newtheorem{lemma}[theorem]{Lemma}
\newtheorem{corollary}[theorem]{Corollary}

\newcommand{\bvec}[1]{\mathbf{#1}}

\newcommand{\REVN}[1]{{#1}}
\newcommand{\REV}[1]{{#1}}

\newcommand{\ve}{\bvec{e}}

\newcommand{\fp}{\mathbf{p}}
\newcommand{\fq}{\mathbf{q}}
\newcommand{\ff}{\mathbf{f}}
\newcommand{\fs}{\mathbf{s}}

\renewcommand{\Re}{\mathrm{Re}}
\renewcommand{\Im}{\mathrm{Im}}
\newcommand{\I}{\mathrm{i}}

\newcommand{\mc}[1]{\mathcal{#1}}
\newcommand{\mf}[1]{\mathfrak{#1}}
\newcommand{\wt}[1]{\widetilde{#1}}

\newcommand{\abs}[1]{\left\lvert#1\right\rvert}
\newcommand{\norm}[1]{\left\lVert#1\right\rVert}
\newcommand{\argmin}{\mathop{\mathrm{argmin}}}

\newcommand{\rd}{\mathrm{d}}

\newcommand{\Or}{\mathcal{O}}

\newcommand{\NN}{\mathbb{N}}
\newcommand{\RR}{\mathbb{R}}
\newcommand{\CC}{\mathbb{C}}

\newcommand{\wrt}{\textit{w.r.t.} }
\newcommand{\ie}{\textit{i.e.}}

\begin{document}

\title{On the energy landscape of symmetric quantum signal processing}

\author{Jiasu Wang} 
        \affiliation{Department of Mathematics, University of California, Berkeley,  CA 94720, USA. }
        \orcid{0000-0002-1321-2649}
\author{Yulong Dong} 
        \affiliation{Department of Mathematics, University of California, Berkeley,  CA 94720, USA. }
        \orcid{0000-0003-0577-2475}
\author{Lin Lin}
        \affiliation{Department of Mathematics, University of California, Berkeley,  CA 94720, USA. }
        \affiliation{Challenge Institute for Quantum Computation, University of California, Berkeley,  CA 94720, USA}
        \affiliation{Applied Mathematics and Computational Research Division, Lawrence Berkeley National Laboratory, Berkeley, CA 94720, USA}
        \orcid{0000-0001-6860-9566}
             
\maketitle

\begin{abstract}
Symmetric quantum signal processing provides a parameterized representation of a real polynomial, which can be translated into an efficient quantum circuit for performing a wide range of computational tasks on quantum computers. For a given polynomial \REV{$f$}, the parameters (called phase factors) can be obtained by solving an optimization problem. However, the cost function is non-convex, and has a very complex energy landscape with numerous global and local minima. It is therefore surprising that the solution can be robustly obtained in practice, starting from a fixed initial guess $\Phi^0$ that contains no information of the input polynomial. To investigate this phenomenon, we first explicitly characterize all the global minima of the cost function. We then prove that one particular global minimum (called the maximal solution) belongs to a neighborhood of $\Phi^0$, on which the cost function is strongly convex \REV{under the condition $\norm{f}_{\infty}=\mathcal{O}(d^{-1})$ with $d=\deg(f)$.}  \REV{Our result provides a partial explanation of the aforementioned success of optimization algorithms}.

\end{abstract}

\newpage
\tableofcontents

\section{Introduction}

Given a target polynomial $f\in\RR[x]$ satisfying (1) $\deg(f)=d$, (2) the  parity of $f$ is $d \bmod 2$, (3) $\norm{f}_{\infty}:=\max_{x\in[-1,1]} \abs{f(x)}< 1$, the problem of  quantum signal processing (QSP)~\cite{LowChuang2017} is to find a set of parameters (called phase factors) $\Phi:=(\phi_0,\cdots,\phi_d)\in [-\pi,\pi)^{d+1}$
so that
\begin{equation}\label{eqn:match_target}
f(x)=g(x,\Phi):=\Re[\braket{0|U(x,\Phi)|0}], \quad x\in[-1,1],
\end{equation}
with
\begin{equation}\label{eqn:unitary-qsvt}
    U(x, \Phi) := e^{\I \phi_0 Z} e^{\I \arccos(x) X} e^{\I \phi_1 Z} e^{\I \arccos(x) X} \cdots e^{\I \phi_{d-1} Z} e^{\I \arccos(x) X} e^{\I \phi_d Z}.
\end{equation}
Here \begin{equation}
X := \left(\begin{array}{cc} 0 & 1\\1 & 0\end{array}\right), \quad Z := \left(\begin{array}{cc} 1 & 0\\0&-1\end{array}\right)
\label{eqn:pauli_xz}
\end{equation}
are Pauli $X$ and $Z$ matrices, respectively. $\braket{0|U(x,\Phi)|0}$ stands for the upper-left entry of the matrix $U(x,\Phi)\in\mathrm{SU}(2)$, and $g(x,\Phi)$
defines a mapping from $[-1,1]$ to $\RR$ for fixed $\Phi$. When the phase factors are restricted to be symmetric, \ie,
\begin{equation}
\Phi=(\phi_0,\phi_1,\phi_2,\ldots,\phi_2,\phi_1,\phi_0)\in [-\pi,\pi)^{d+1},
\label{eqn:symmetry_phase}
\end{equation}
this is referred to as the symmetric quantum signal processing. The simplest example is $\Phi=(0,\ldots,0)$. This gives $U(x,\Phi)=e^{\I d\arccos(x) X}$ and $g(x,\Phi)=\cos(d \arccos(x))=T_d(x)$, where $T_d$ is the Chebyshev polynomial of the first kind of degree $d$. 
 
Due to the parity constraint, the number of degrees of freedom in the target polynomial $f(x)$ is $\wt{d} := \lceil \frac{d+1}{2} \rceil$. Hence $f(x)$ is entirely determined by its values on $\wt{d}$ distinct points. 
Throughout the paper, we choose these points to be $x_k=\cos\left(\frac{2k-1}{4\wt{d}}\pi\right)$, $k=1,...,\wt{d}$, \ie, positive nodes of the Chebyshev polynomial $T_{2 \wt{d}}(x)$. 
For any target polynomial $f(x)$ defined above, the solution to \cref{eqn:match_target} exists~\cite{LowChuang2017,GilyenSuLowEtAl2019}, and Ref.~~\cite{DongMengWhaleyEtAl2021} suggests that the solution can be restricted to be symmetric.
The existence of the solution implies that the problem of symmetric quantum signal processing can be equivalently solved via the following optimization problem
\begin{equation}\label{eqn:optprob-intro}
    \Phi^* = \argmin_{\substack{\Phi \in [-\pi,\pi)^{d+1},\\
    \text{symmetric.}}} F(\Phi),\ F(\Phi) := \frac{1}{\wt{d}} \sum_{k=1}^{\wt{d}} \abs{g(x_k, \Phi) - f(x_k)}^2,
\end{equation}
\ie, any solution $\Phi^*$ to \cref{eqn:match_target} achieves the global minimum of the cost function with $F(\Phi^*)=0$, and vice versa.

However, the energy landscape of the cost function $F(\Phi)$ is very complex, and has numerous global as well as local minima (see \cref{sec:numer}). 
This is already the case with the symmetry constraint. 
(Without the symmetry constraint, the number of variables is larger than the number of equations and there should be an infinite number of global minima.)
Starting from a random initial guess, an optimization algorithm can easily be trapped at a local minima already when $d$ is small\footnote{An early attempt (not using optimization-based approach) showed that even finding $\Phi$ with $d>30$ can be very costly~\cite[Apendix H.3]{ChildsMaslovNamEtAl2018}.}.
It is therefore surprising that starting from a special symmetric initial  guess
\begin{equation}\label{eqn:phi0}
\Phi^0=(\pi/4,0,0,\ldots, 0,0,\pi/4),
\end{equation}
\REV{which is independent of the target function,} at least one global minimum can be robustly identified using standard unconstrained optimization algorithms even when $d$ is as large as $10,000$~\cite{DongMengWhaleyEtAl2021}, and the optimization method is free from being trapped by any local minima. 
Direct calculation shows that $g(x,\Phi^0)=0$, and therefore $\Phi^0$ does not contain any \emph{a priori} information of the target polynomial $f(x)$! 

\subsection{Main results}

In this work, we provide a theory to characterize the energy landscape of the cost function $F(\Phi)$, and to explain this unusual behavior of numerical optimization. Under the assumptions on the maxnorm $\norm{f}_{\infty}$ as stated in \cref{cor:conv_pg}, this leads to the first provable algorithm for solving the  QSP problem without referring to extended precision arithmetic operations~\cite{Haah2019}.

Let the domain of the symmetric phase factors be 
\begin{equation}
    D_d=\begin{cases}
    [-\frac{\pi}{2},\frac{\pi}{2})^{\frac{d}{2}} \times [-\pi,\pi) \times [-\frac{\pi}{2},\frac{\pi}{2})^{\frac{d}{2}}, & d \mbox{ is even,}\\
    [-\frac{\pi}{2},\frac{\pi}{2})^{d+1}, &d \mbox{ is odd.}\\
    \end{cases}
\end{equation}
Our first main result is the existence and uniqueness of symmetric phase factors for a class of polynomial matrices in $\mathrm{SU}(2)$. 
\begin{theorem}[Existence and uniqueness of symmetric phase factors]
\label{thm:existandunique}
Consider any $P\in \mathbb{C}[x]$ and $Q\in \mathbb{R}[x]$ satisfying the following conditions
\begin{enumerate}[label=(\arabic*)]
    \item $\deg(P)= d$ and $\deg(Q)= d-1$.
    \item $P$ has parity $(d \bmod 2)$ and $Q$ has parity $(d-1 \bmod 2)$.
    \item (Normalization condition) $\forall x\in[-1,1]: |P(x)|^2+(1-x^2)|Q(x)|^2=1$.
    \item \label{itm:4} If $d$ is odd, then the leading coefficient of $Q$ is positive.
\end{enumerate}
There exists a unique set of symmetric phase factors $\Phi:=(\phi_0,\phi_1,\cdots,\phi_1,\phi_0)\in D_d$ such that 
\begin{equation}\label{eq:UPQ}
U(x,\Phi)=\begin{pmatrix}
P(x) & \I Q(x)\sqrt{1-x^2}\\
\I Q(x) \sqrt{1-x^2} & P^* (x)
\end{pmatrix}.
\end{equation}
\end{theorem}

This result is a refinement of \cite[Theorem 4]{GilyenSuLowEtAl2019}, which only shows the existence of the phase factors without symmetry constraints. 
\cref{thm:existandunique} states that in the presence of the symmetric constraint, the phase factors are also unique. Besides, under some constraint, the converse of \cref{thm:existandunique} also holds (see \cref{re:converse_thm1}). It follows that there is a bijection between the global minimizers of \cref{eqn:optprob-intro} and all possible pairs of $(P(x), Q(x))$ satisfying the assumption in \cref{thm:existandunique} and $\Re[P](x)=f(x)$.
The conditions (1), (2) for the target polynomial $f$ are compatible with the first two requirements in \cref{thm:existandunique}.
The condition (3) on the maxnorm, \ie, $\norm{f}_{\infty}< 1$ is compatible with the normalization condition, which is itself a natural condition due to the unitarity of $U(x,\Phi)$.

To simplify the discussion, we introduce the following definition of admissible pair of polynomials associated with a target polynomial. Here $P_{\Re}:=\Re[P], P_{\Im}:=\Im[P]$.
\begin{definition}[Admissible pair of polynomials]\label{con:PQ}
Let $f\in\RR[x]$ be a target polynomial satisfying (1) $\deg(f)=d$, (2) the  parity of $f$ is $(d \bmod 2)$, (3) $\norm{f}_{\infty}<1$. Then $(P,Q)$ is an admissible pair of polynomials associated with $f$ if the conditions (1)-(3) in \cref{thm:existandunique} are satisfied, together with
\begin{enumerate}[label=(\arabic*)]
\setcounter{enumi}{3}
    \item The leading coefficient of $Q$ is positive,
    \item $P_{\Re}(x)=f(x)$.
\end{enumerate}
The polynomials $P_{\Im},Q\in\RR[x]$ are also called complementary polynomials to $f$.
\end{definition}

Comparing \cref{thm:existandunique} and \cref{con:PQ}, the main modification is that the leading coefficient of the complementary polynomial $Q$ is always restricted to be positive. 
This allows us to unify the discussion of odd and even values of $d$.
The bijection relation can then be concisely formulated as follows.
\begin{corollary}[Bijection between global minima and admissible pairs]\label{cor:bijection}
If $d$ is odd, there is a bijection between the global minima of \cref{eqn:optprob-intro} and all admissible pairs $(P, Q)$.

If $d$ is even, there is a bijection between the global minima  of \cref{eqn:optprob-intro} and all pairs of polynomials $(P, \pm Q)$, where $(P, Q)$ is an admissible pair.
\end{corollary}

The proof of \cref{thm:existandunique} is constructive. 
So given an admissible pair, we have an algorithm to evaluate the symmetric phase factor. \cref{thm:admissible_pair} generalizes the result of \cite[Lemma 4]{Haah2019} by explicitly constructing \textit{all} admissible pairs.
Together with \cref{thm:existandunique}, we have a complete description of all global minima of the cost function in \cref{eqn:optprob-intro}.
\begin{theorem}[Construction of admissible pairs]
\label{thm:admissible_pair}
   Given a target polynomial $f(x)$, all admissible pairs $(P, Q)$ must take the following form,
   \begin{equation}
 \begin{split}
    P_\Im\left(x\right)&=\sqrt{\alpha} \frac{e\left(x+\I \sqrt{1-x^2}\right)+e\left(x-\I\sqrt{1-x^2}\right)}{2},\\ \quad Q\left(x\right)& =\sqrt{\alpha} \frac{e\left(x+\I \sqrt{1-x^2}\right)-e\left(x-\I\sqrt{1-x^2}\right)}{2\I \sqrt{1-x^2}}.
\end{split}
\end{equation}
\REV{Here, $e(z):=z^{-d} \prod_{i=1}^{2d} (z-r_i)$, where each $r_i$ $(i=1,\ldots,2d)$ is a root of the function
\begin{equation}
    \mf{F}(z):=1-\left[f\left(\frac{z+z^{-1}}{2}\right)\right]^2,
\end{equation} 
and $\alpha\in \CC$ satisfies $\mf{F}(z)=\alpha  e(z)e(z^{-1})$.}
\end{theorem}
\REV{Note that the set $\{r_i\}_{i=1}^{2d}$ does not include all roots of the \emph{Laurent polynomial} $\mf{F}(z)$, which has $4d$ roots in total.}
By properly choosing the roots (see a more detailed description in \cref{re:construct_method}), we can construct all the admissible pairs. 
The number of global minima is finite, which is a consequence of the compactness of the domain $D_d$. \cref{thm:admissible_pair} shows that the number of global minima can grow combinatorially with respect to $d$. 

Unfortunately, the procedure described above for finding phase factors is numerically unstable. It requires extended precision arithmetic operations and is therefore very expensive when $d$ is large (see detailed discussions in \cref{sec:related}).
This is noticeably different from solving the optimization problem in \cref{eqn:optprob-intro}, which is numerically stable and can be readily performed using standard double precision arithmetic operations. 

Among the myriad of global minima, \cref{thm:admissible_pair} allows us to identify a special global minimum, which is obtained by choosing $\{r_i\}_{i=1}^{2d}$ to be the roots of $\mf{F}(z)$ within the unit disc. 
The unique symmetric phase factor associated with this admissible pair is referred to as the \emph{maximal solution} (the reason for the naming is technical and is explained in \cref{sec:maximal_sol}). 

The maximal solution enjoys many desirable properties. For any target polynomial $f$ with $\norm{f}_{\infty}\leq \frac{1}{2}$, the maximal solution lies in the neighborhood of $\Phi^0$, \REV{ which is defined in \cref{eqn:phi0}. To characterize this property, we first introduce the definition of reduced phase factors. Given any symmetric phase factors $\Phi$ of length $d+1$, reduced phase factors is referred to the left half part of $\Phi$, namely $\left(\phi_0,\phi_1,\cdots,\phi_{\wt{d}-1}\right)$, and is denoted as $\wt{\Phi}$. Recall that $\wt{d}:=\left\lceil\frac{d+1}{2}\right\rceil$. Due to symmetry, $\Phi$ is fully determined by $\wt{\Phi}$. Hence, we may adopt Euclidean distance between $\wt{\Phi}^*$ and $\wt{\Phi}^0$ instead of that between $\Phi^*$ and $\Phi^0$. Furthermore, throughout this paper, we choose $\wt{\Phi}$ as the free variables. With some abuse of notation, we identify the cost function $F(\Phi)$ with $F(\wt{\Phi})$.}

\begin{theorem}[Distance between the maximal solution and $\Phi^0$]
\label{thm:estimate_maximal} 
Let $\Phi^*$ be the maximal solution for the target function $f(x)$. Denote $\wt{\Phi}^*$ and $\wt{\Phi}^0$ as the corresponding reduced phase factors of $\Phi^*$ and $\Phi^0$ respectively. If $\norm{f(x)}_{\infty}\leq \frac{1}{2}$, then 
\begin{equation}
    \norm{\wt{\Phi}^*-\wt{\Phi}^0}_2\leq \frac{\pi}{\sqrt{3}} \norm{f(x)}_{\infty}.
\end{equation}
\end{theorem}

In particular, if the target polynomial is $f=0$, then the maximal solution is the initial guess $\Phi^0$.  We then prove that when $\norm{f}_{\infty}$ is sufficiently small, the maximal solution belongs to a neighborhood of $\Phi^0$, on which the cost function $F(\wt{\Phi})$ is strongly convex, \ie, the Hessian matrix denoted by $\mathrm{Hess}(\wt{\Phi})$ is positive definite. 

\begin{theorem}[Local strong convexity]
\label{thm:Hess_PD}
        If the target polynomial satisfies $\norm{f}_\infty \leq \frac{\sqrt{3}}{20\pi\wt{d}}$, for any symmetric phase factors $\Phi$ of length $d+1$ satisfying $\norm{\wt{\Phi} -\wt{\Phi}^0}_2 \leq \frac{1}{20\wt{d}}$, the following estimate holds:
        \begin{equation}
            \frac{1}{4}\leq \lambda_{\min} \left(\mathrm{Hess}(\wt{\Phi})\right)\leq \lambda_{\max} \left(\mathrm{Hess}(\wt{\Phi})\right)\leq \frac{25}{4}.
        \end{equation}
\end{theorem}

The local strong convexity result in \cref{thm:Hess_PD} immediately implies that when $\norm{f}_{\infty}$ is sufficiently small, standard optimization algorithms, such as the projected gradient method~\cite{Bertsekas1976,Kelley1999}, can converge in the neighborhood of $\Phi^0$, without being trapped by any local minima.

\begin{corollary}[Convergence of projected gradient method]
\label{cor:conv_pg}
If the target polynomial satisfies $\norm{f}_\infty \leq \frac{\sqrt{3}}{20\pi\wt{d}}$, starting from $\Phi^0$, the projected gradient method with step size $t=\frac{1}{L}$ converges exponentially to the maximal solution $\Phi^*$, \ie, at the $\ell$-th iteration 
\begin{equation}
    \norm{\wt{\Phi}^\ell-\wt{\Phi}^*}^2_2\leq e^{-\frac{\sigma}{L}\ell} \norm{\wt{\Phi}^0-\wt{\Phi}^*}^2_2.
\end{equation}
Here $\sigma=\frac{1}{4}$ and $L=\frac{25}{4}$.
\end{corollary}

Therefore when $\norm{f}_\infty \leq \frac{\sqrt{3}}{20\pi\wt{d}}$, in order to reach precision $\epsilon$, the projected gradient method can be terminated after $\Or(\log(1/\epsilon))$ steps independent of the details of $f$. Since each iteration is numerically stable, the maximal solution can be readily obtained using standard double precision arithmetic operations  in practice.

\subsection{Background and related works}\label{sec:related}

The mapping in \cref{eqn:match_target,eqn:unitary-qsvt} may seem a very peculiar way for encoding a polynomial (and it is). 
However, such an SU(2) representation can be directly translated into a quantum circuit, which is so far the most concise way for performing eigenvalue transformation $f(A)$ (when $A$ is an Hermitian matrix)~\cite{LowChuang2017}, and singular value transformations $f^{\mathrm{SV}}(A)$ (when $A$ is a general matrix) on quantum computers~\cite{GilyenSuLowEtAl2019}.
Many tasks in quantum computation can be formulated using such transformations.
When $f$ is not a polynomial, it can be approximated by a polynomial with bounded precision. 
This strategy can be used to unify a large class of quantum algorithms including Grover's search and quantum phase estimation~\cite{NielsenChuang2000,MartynRossiTanEtAl2021}, and to perform a wide range of computational tasks, such as solving linear system of equations~\cite{NielsenChuang2000,LinTong2020}, eigenvalue problems~\cite{LinTong2020a}, Hamiltonian simulation~\cite{LowChuang2017,GilyenSuLowEtAl2019} etc.

After the initial efforts~\cite{LowChuang2017,ChildsMaslovNamEtAl2018}, significant progress has been made by using direct methods to obtain phase factors~\cite{GilyenSuLowEtAl2019,Haah2019}. These methods are based on finding roots of high degree polynomials to high precision and are not numerically stable. Specifically, these algorithms require $\Or(d\log(d/\epsilon))$ bits of precision, where $d$ is the degree of $f(x)$ and $\epsilon$ is the target accuracy~\cite{Haah2019}. 
It is worth mentioning that the extended precision needed in these algorithms is not an artifact of the proof technique. For instance, for $d\approx 500$, the number of bits needed to represent each floating point number can be as large as  $1000\sim 2000$~\cite{DongMengWhaleyEtAl2021}, which is much larger than the $64$ bits provided by standard double precision floating point format.

\REV{There have been two recent improvements of the factorization based method, based on the capitalization method~\cite{ChaoDingGilyenEtAl2020}, and the Prony method~\cite{Ying2022}, respectively. 
Although the two methods differ significantly, empirical results indicate that both methods are numerically stable, and are applicable to polynomials of large degrees. The algorithm in~\cite{ChaoDingGilyenEtAl2020} is a variation of the algorithm in~\cite{GilyenSuLowEtAl2019}, which introduces a small perturbation to the high order Chebyshev coefficients to enhance the numerical stability. The Prony method in~\cite{Ying2022} avoids the root finding of high degree polynomials by directly constructing the admissible pairs for a given target polynomial, and uses randomness to enhance the numerical stability.
Ref.~\cite{DongMengWhaleyEtAl2021} proposes a very different optimization based algorithm and suggests the use of the symmetric phase factors.  This work provides a detailed discussion of the analytic structure of quantum signal processing with symmetric phase factors, and demonstrates the effectiveness of the optimization based approach under the condition that $\norm{f}_{\infty}$ is sufficiently small.} 
This provides a partial solution to the open problem of finding phase factors using numerically stable algorithms, and justifies the usage of standard double precision arithmetic operations in practice. 

\subsection{Discussion and open questions}

Without the symmetry constraint, the optimization problem in \cref{eqn:optprob-intro} has $d$ degrees of freedom. 
This is larger than the number of degrees of freedom of target function $f$, which is only $\wt{d}$ due to the parity constraint. 
As a result, the Hessian matrix is always singular, preventing us from proving the exponential convergence of optimization methods \cite{NocedalWright1999}. 

\REV{Our analytical result is the first result trying to explain the unexpected success of optimization algorithms in the problem of finding phase factors. In \cref{thm:Hess_PD}, we prove the local strong convexity of the cost function when $\norm{f}_\infty=\Or(d^{-1})$.
Therefore for a given target polynomial, in order to apply \cref{cor:conv_pg}, we need to first multiply the target polynomial by a scaling factor $c=\Or(d^{-1})$ and find the phase factors for $c f(x)$ instead. Such a scaling factor can be undesirable and may lead to suboptimal query complexities when implementing certain quantum algorithms. On the other hand, numerical results suggest that optimization algorithms can robustly converge to a global minimum when $\norm{f}_\infty$ is smaller than a constant (less than $1$) that is independent of $d$. Note that the estimate in \cref{thm:estimate_maximal} is already independent of $d$. Hence the gap is mainly due to \cref{thm:Hess_PD}, where the eigenvalue estimate of the Hessian matrix is obtained via the estimate of each matrix entry, which introduces the $d$-dependence. It is a natural open question to prove the strong convexity of the cost function near $\Phi^0$ where $\norm{f}_{\infty}$ can be independent of $d$. }

The proof of \cref{cor:conv_pg} requires the projection of $\wt{\Phi}^{\ell}$ to the domain $\{\wt{\Phi}:\norm{\wt{\Phi} -\wt{\Phi}^0}_2 \leq \frac{1}{20\wt{d}}\}$ at each iteration. This is because the steepest descent method may potentially overshoot and deviate away from the locally strong convex region. 
Numerical observation shows that unconstrained optimizers, such as quasi-Newton methods can robustly achieve the global minimum, which is beyond our current theoretical analysis. 

Finally, since $\Phi^0$ is only one of the solutions when $f=0$, we may ask what if we run the projected gradient method starting from other solutions when $f=0$, and to obtain optimal phase factors for more general target polynomials. These solutions correspond to non-maximal solutions. Numerical observations show that the convergence rate towards such non-maximal solutions can be significantly slower than that towards the maximal solution (see \cref{fig:convergence_rate,fig:convergence_qsppack}). 
The analysis of such behavior requires the generalization of estimates of in \cref{thm:estimate_maximal,thm:Hess_PD} to non-maximal solutions. These will be our future works.

\subsection*{Acknowledgments:} 

This work was partially supported by the NSF Quantum Leap Challenge Institute (QLCI) program through grant number OMA-2016245 (J.W.,Y.D.), by Department of Energy under Grant No. DE-SC0017867 and No. DE-AC02-05CH11231 (L.L.). L.L. is a Simons Investigator.

\section{Preliminaries}
\subsection{Chebyshev polynomials}\label{sec:chebyshev}
We first collect some basic facts of Chebyshev polynomials, which is used to formulate the optimization problem in \cref{eqn:optprob-intro} \REV{and will be extensively used in \cref{sec:local_convergence,sec:local_strong_convexity}.}
Given $x\in[-1,1]$, the Chebyshev polynomial of the first kind is $T_n(x)=\cos\left(n\arccos{(x)}\right)$, and that of the second kind is given by $U_n(x)\sin(\arccos{(x)})=\sin\left((n+1)\arccos{(x)}\right)$.

A useful fact is that both $T_n$ and $U_n$ form a sequence of orthogonal polynomials. The Chebyshev polynomial of the first kind $T_n$ are orthogonal with respect to the weight $\frac{1}{\sqrt{1-x^2}}$ over $[-1,1]$, and that of the second kind $U_n$ are orthogonal with respect to the weight $\sqrt{1-x^2}$ over $[-1,1]$,
\begin{equation*}
\begin{split}
    &\frac{1}{\pi}\int_{-1}^1 T_n(x) T_m(x) \frac{1}{\sqrt{1-x^2}}\rd x=\frac{1}{2}\delta_{nm}(1+\delta_{n0}),\\
    &\frac{1}{\pi}\int_{-1}^1 U_n(x) U_m(x) \sqrt{1-x^2}\rd x=\frac{1}{2}\delta_{nm}.
\end{split}
\end{equation*}
The induced norms
$$\norm{f}_T^2:=\frac{1}{\pi}\int_{-1}^1 f^2(x)\frac{1}{\sqrt{1-x^2}}\rd x,\quad \norm{f}_U^2:=\frac{1}{\pi}\int_{-1}^1 f^2(x)\sqrt{1-x^2}\rd x$$
are referred to as the Chebyshev norm of the first kind and the second kind respectively. Furthermore, any function $h$ with $\norm{h}_T<\infty$ can be uniquely expressed as a series of Chebyshev polynomials of the first kind, 
\begin{equation}
    h(x)=\sum_{n=0}^{\infty} c_n T_n(x), \qquad c_n =\frac{2}{\pi(1+\delta_{n0})}\int_{-1}^1 h(x) T_n(x) \frac{1}{\sqrt{1-x^2}}\rd x.
\end{equation}
Similarly, any function $h$ with $\norm{h}_U<\infty$ can be uniquely expressed as a series of Chebyshev polynomials of the second kind, 
\begin{equation}
    h(x)=\sum_{n=0}^{\infty} c_n \REV{U_n(x)}, \qquad c_n =\frac{2}{\pi}\int_{-1}^1 h(x) \REV{U_n(x) }\sqrt{1-x^2}\rd x.
\end{equation}

Apart from the weighted orthogonality, Chebyshev polynomials of the first kind $T_n$ satisfies the discrete orthogonality condition:
\begin{equation}
    \sum_{j=1}^{2\wt{d}} T_n(x_j)T_m(x_j)=\begin{cases}
    0 & n\ne m\\
    2\wt{d} & n=m=0\\
    \wt{d} & n=m\ne 0
    \end{cases},
\end{equation}
where $x_j =\cos\left(\pi\frac{2j-1}{4\wt{d}}\right)$ are the nodes of $T_{2\wt{d}}$ and $\max(n,m)\leq 2\wt{d}$. This set of points is also referred to as Chebyshev nodes. Furthermore, if $n$ and $m$ have the same parity which implies $T_n(-x)T_m(-x) = T_n(x) T_m(x)$, we have
\begin{equation}
    \sum_{j=1}^{\wt{d}} T_n(x_j)T_m(x_j) = \frac{1}{2} \sum_{j=1}^{2\wt{d}} T_n(x_j) T_m(x_j) = \begin{cases}
    0 & n\ne m\\
    \wt{d} & n=m=0\\
    \frac{\wt{d}}{2} & n=m\ne 0
    \end{cases}.
\end{equation}

\subsection{Notation}\label{sec:notation}
Given a positive integer $n$, $\mathbb{K}_n[x]$ denotes the set of all polynomials of degree at most $n$ with variable $x$, and coefficients taken in the field $\mathbb{K}$. The Laurent polynomial is a linear combination of positive as well as negative powers of the variable $x$ with coefficients in the field $\mathbb{K}$, and $\mathbb{K}[x,x^{-1}]$ denotes the set of all Laurent polynomials.  Specifically, we are only interested in the cases where $\mathbb{K}$ is either $\RR$ or $\CC$. Unless otherwise noted, the term ``polynomial'' refers to real polynomial or complex polynomial, which only have nonnegative power of the variable $x$. 

We assume that the length of the set of phase factors is $d+1$ unless otherwise noted. We also assume that a pair of polynomials $(P,Q)\in \CC_d[x]\times \CC_{d-1}[x]$ satisfies the conditions in \cref{thm:existandunique}. Given a polynomial $P \in \CC[x]$, we define $P_\Re(x):=\Re[P(x)]$ and $P_\Im(x) := \Im\left[P(x)\right]$. The real polynomial $P_\Re(x)$ is usually denoted by $f(x)$ when it refers to the target of the optimization.

The QSP problem involves some delicate relations of polynomial coefficients.
Throughout the paper, we will interchangeably expand a polynomial using the monomial basis and the Chebyshev basis, \ie,
\begin{equation}
f(x)=\sum_{j=0}^d \ff_j x^j=\sum_{j=0}^d f_j T_j(x),\quad \ff_j,f_j\in \RR.
\end{equation}
Similarly, 
\begin{equation}\label{eq:cheby_expansion}
    \begin{split}
        P(x)&=\sum_{j=0}^d \fp_j x^j=\sum_{j=0}^d p_j T_j(x),\quad \fp_j,p_j\in \CC, \\
        P_\Im(x) &=\sum_{j=0}^d \fs_j x^j=\sum_{j=0}^d s_j T_j(x),\quad \fs_j,s_j\in \RR, \\
        Q(x)&=\sum_{j=0}^{d-1} \fq_j x^j=\sum_{j=0}^{d-1} q_j U_j(x)\quad \fq_j,q_j\in \CC .
    \end{split}
\end{equation}
For convenience, we set the values of $\ff_j,f_j,\fp_j,p_j,\fs_j,s_j,\fq_j,q_j$ with a negative integer index to zero. 

The coefficients in the monomial and Chebyshev basis satisfy the relation
\begin{equation}\label{eq:relation_monomial_chebyshev}
    \quad \ff_j =2^{j-1} f_j,\quad \fp_j =2^{j-1} p_j,\quad \fs_d =2^{j-1} s_j,\quad\fq_{j} =2^{j} q_{j}, \quad \forall j\geq0.
\end{equation}

The set  $[n]:=\{0,1,\cdots, n-1\}$ is referred to as the index set generated by a positive integer $n$. The row and column indices of a $n$-by-$n$ matrix run from $0$ to $n-1$, namely in the index set $[n]$. For a matrix $A\in\CC^{m\times n}$, the transpose, Hermitian conjugate and complex conjugate are denoted by $A^{\top}$, $A^{\dag}$, $A^*$, respectively. The same notations are also used for the operations on a vector. We use the convention in quantum computing to write the basis vectors in $\CC^2$. It is equivalent to the following identification
\begin{equation*}
    \ket{0} = \left(\begin{array}{c}
         1 \\
         0 
    \end{array}\right),\ \ket{1} = \left(\begin{array}{c}
         0 \\
         1 
    \end{array}\right),\ \bra{0} = \ket{0}^\dagger = \left(1, 0\right), \text{ and } \bra{1} = \ket{1}^\dagger = \left(0, 1\right).
\end{equation*}
Given a $2$-by-$2$ matrix $A = (a_{ij})$, we have $\braket{i | A | j} = a_{ij}$ for any $i, j \in \{0, 1\}$.

For a matrix $A\in\CC^{m\times n}$, we denote the operator norm $\norm{\cdot}_2$ induced by the vector $\ell_2$ norm as
\begin{equation}
\norm{A}_2: =\sup_{x \neq 0}  \frac{\norm{A x}_2 }{\norm{x}_2}, 
\end{equation}
and denote the Frobenius norm $\norm{\cdot}_F$ as
\begin{equation}
    \norm{A}_F:=\sqrt{\sum_{i\in [m],j\in[n]}\abs{\REV{a_{ij}}}^2}.
\end{equation}
We use $\sigma_k(A)$ to denote its $(k+1)$-th largest singular value, \ie,
\begin{equation*}
    \REV{0\leq \sigma_{\min(m,n)-1}(A)\leq \cdots\leq \sigma_{1}(A)\leq \sigma_0(A)}.
\end{equation*}
We define $\sigma_{\min}(A):=\sigma_{\min(m,n)-1}(A)$ and  $\sigma_{\max}(A):=\sigma_0(A)$.
Here, the index $k$ starts from $0$ which agrees with the convention of matrix index.

For a Hermitian matrix $A\in \CC^{n\times n}$, we use the notation $\lambda_k(A)$ to designate the $(k+1)$-th largest eigenvalue, \ie,
\begin{equation*}
    \REV{\lambda_{n-1}(A)\leq \cdots\leq \lambda_{1}(A)\leq \lambda_0(A)}.
\end{equation*}
We define $\lambda_{\min}(A):=\lambda_{n-1}(A)$ and  $\lambda_{\max}(A):=\lambda_0(A)$.

\subsection{Quantum signal processing}
In the absence of the symmetry constraint, the existence of phase factors for a unitary representation of a pair of polynomials $(P,Q)$ is proved in~\cite{GilyenSuLowEtAl2019}. 
\begin{theorem}[\textbf{Quantum signal processing} {\cite[Theorem 4]{GilyenSuLowEtAl2019}}]\label{thm:qsp}
    Let $d\in \NN$. There exists a set of phase factors $\Phi := (\phi_0, \cdots, \phi_d) \in \RR^{d+1}$ such that
\begin{equation}
\label{eq:qsp-gslw}
        U(x, \Phi) = e^{\I \phi_0 Z} \prod_{j=1}^{d} \left( W(x) e^{\I \phi_j Z} \right) = \left( \begin{array}{cc}
        P(x) & \I Q(x) \sqrt{1 - x^2}\\
        \I Q^*(x) \sqrt{1 - x^2} & P^*(x)
        \end{array} \right),
\end{equation}
where 
\begin{displaymath}
W(x) = e^{\I \arccos(x) X}=\left(\begin{array}{cc}{x} & {\I \sqrt{1-x^{2}}} \\ {\I \sqrt{1-x^{2}}} & {x}\end{array}\right),
\end{displaymath} 
if and only if $P,Q\in \CC[x]$ satisfy that
    \begin{enumerate}[label=(\arabic*)]
        \item \label{itm:1} $\deg(P) \leq d, \deg(Q) \leq d-1$,
        \item \label{itm:2} $P(x)$ has parity $(d\mod2)$ and $Q(x)$ has parity $(d-1 \mod 2)$,
        \item  \label{itm:3} $|P(x)|^2 + (1-x^2) |Q(x)|^2 = 1, \forall x \in [-1, 1]$.
    \end{enumerate}
\end{theorem}

When the real component of $P$ is of interest only, the conditions for the existence of phase factors can be further simplified. 

\begin{corollary}[\textbf{Real quantum signal processing} {\cite[Corollary 5]{GilyenSuLowEtAl2019}}]\label{cor:complementary}
Let $f\in \RR[x]$ be a degree-d polynomial for some $d\geq 1$ such that 
\begin{itemize}
    \item $f(x)$ has parity $(d \mod 2)$,
    \item $\abs{f(x)}<1, \forall x\in[-1,1]$.
\end{itemize}
Then there exists some $P,Q\in \CC[x]$ satisfying properties \cref{{itm:1},{itm:2},{itm:3}} of \cref{thm:qsp} such that $f(x)=\Re[P(x)]$.
\end{corollary}

\subsection{Quantum signal processing as an optimization problem}
The existence result in \cref{cor:complementary} enables the following optimization based strategy to find phase factors.
Given the target polynomial $f(x)\in \RR[x]$, we may minimize the $\mathcal{L}_2$ distance between $g(x,\Phi)$ and target function $f(x)$, \ie,
\begin{equation}
    L(\Phi)=\left(\int_{-1}^1\left[g(x,\Phi)-f(x)\right]^2 \rd x\right)^{\frac{1}{2}}.
\end{equation}
The desired phase factors $\Phi^*$ satisfy $L(\Phi^*)=0$. 

Since both $g(x,\Phi)$ and $f(\Phi)$ are polynomials of degree $d$ with definite parity, the number of degrees of freedom is $\wt{d} := \lceil \frac{d+1}{2} \rceil$. These polynomials are determined by their values at $x_k=\cos\left(\frac{2k-1}{4\wt{d}}\pi\right)$, $k=1,...,\wt{d}$, which are positive Chebyshev nodes of $T_{2 \wt{d}}(x)$. This leads to the following optimization problem,
\begin{equation}
    \Phi^* = \argmin_{\Phi \in [-\pi,\pi)^{d+1}} F(\Phi),\ F(\Phi) := \frac{1}{\wt{d}} \sum_{k=1}^{\wt{d}} \abs{g(x_k, \Phi) - f(x_k)}^2.
\end{equation}

\section{Symmetric quantum signal processing}

In this section, we present some basic properties of symmetric quantum signal processing. We first explain a constraint of the sign of the leading coefficient of $Q(x)$ when $d$ is odd. Then we introduce a constructive method for finding phase factors given suitable polynomials $\left(P,Q\right)$. In the end, we prove \cref{thm:existandunique}, which establishes a bijection between the global minimizers of \cref{eqn:optprob-intro} and all suitable complementary polynomials $\left(P_\Im, Q\right)$.

\subsection{Sign of the leading coefficient of $Q(x)$}
We first derive an explicit expression of the coefficients of $T_d$ and $U_{d-1}$ in the Chebyshev expansion of $P$ and $Q$ respectively.
\begin{lemma}\label{lma:leading_coef}
Let $\Phi\in\RR^{d+1}$ be a set  of phase factors, $P(x)$ and $Q(x)$ be defined by \cref{eq:qsp-gslw}. Considering the Chebyshev expansion of $P$ and $Q$ in \cref{eq:cheby_expansion}, one has
    \begin{equation}\label{eq:chebycoefPQ}
        p_d = e^{\I\left(\phi_0+\phi_d\right)}\prod_{j=1}^{d-1} \cos\left(\phi_j\right),\quad q_{d-1} = e^{\I\left(\phi_0-\phi_d\right)} \prod_{j=1}^{d-1} \cos\left(\phi_j\right).
    \end{equation}
\end{lemma}
\begin{proof}
See \cref{sec:proof_leading_coef}.
\end{proof}

After imposing the symmetry constraint on phase factors,  we find that there is an additional constraint on the leading coefficient of $Q$ when $d$ is odd. It originates from the observation that when $\Phi$ is symmetric, one has $\phi_{0}=\phi_{d}$ and then the second equality in \cref{eq:chebycoefPQ} becomes 
\begin{equation}
q_{d-1} = \prod_{j=1}^{d-1} \cos\left(\phi_j\right)=\prod_{j=1}^{\wt{d}-1} \cos^2\left(\phi_j\right)\ge 0.
\end{equation}
Note that $\deg(Q)$ may be less than $d-1$. The more precise statement is given in \cref{lma:leading_coef_sym}.

\begin{lemma}\label{lma:leading_coef_sym}
Let $\Phi\in\RR^{d+1}$ be a set  of phase factors, $P(x)$ and $Q(x)$ be defined by \cref{eq:qsp-gslw}. If $\Phi$ is symmetric and $d$ is odd, then the leading coefficient of $Q$ in the Chebyshev basis has the same sign as $\left(-1\right)^{\frac{d-1-\deg (Q)}{2}}$.
\end{lemma}
\begin{proof}
See \cref{sec:proof_leading_coef_sym}.
\end{proof}

As a remark, the statement in \cref{lma:leading_coef_sym} also holds with the  Chebyshev basis replaced by the monomial basis, which is an consequence of directly applying the relation in \cref{eq:relation_monomial_chebyshev}.
\begin{remark}\label{re:converse_thm1}
By directly applying \cref{lma:leading_coef_sym} and \cref{thm:qsp}, one gets that the converse of \cref{thm:existandunique} also holds, given that $\deg(P)= d$ and $\deg(Q)= d-1$.
\end{remark}

\subsection{Existence of symmetric phase factors}\label{sec:existence}
We present a useful lemma, which evaluates the outmost symmetric phase factor, and reduces the polynomial degree by $2$. It also connects the leading coefficients in the monomial basis before and after the reduction. The result of \cref{lma:newPQ_coef} will be used throughout the paper. 
To clarify, we only consider the expansion of polynomial in the monomial basis in rest of this section.
\begin{lemma}\label{lma:newPQ_coef}
Suppose that $P(x)\in \CC[x],Q(x)\in\RR[x]$ satisfy the conditions in \cref{thm:existandunique} and $d\geq 2$. Choose 
\begin{equation}
e^{2\I \phi_0}=\frac{\fq_{d-1}}{\fp_{d}^*}
\end{equation}
and define
\begin{equation}\label{eq:constructP1Q1}
\begin{split}
&\begin{pmatrix}
P^{(1)}(x) & \I Q^{(1)}(x)\sqrt{1-x^2}\\
\I Q^{(1)}(x)\sqrt{1-x^2} & \left(P^{(1)}\right)^* (x)
\end{pmatrix}\\
&:= W^{-1} e^{-\I\phi_0 Z} \begin{pmatrix}
P(x) & \I Q(x)\sqrt{1-x^2}\\
\I Q(x) \sqrt{1-x^2} & P^* (x)
\end{pmatrix} e^{-\I\phi_0 Z} W^{-1}.
\end{split}
\end{equation}
Consider the expansion of $P$, $Q$, $P^{(1)}$ and $Q^{(1)}$ in the monomial basis, 
\begin{equation}
\begin{split}
    P(x)&=\sum_{i=0}^{d} \fp_i x^i, \quad\quad Q(x)=\sum_{j=0}^{d-1} \fq_j x^j.\\
    P^{(1)}(x)&=\sum_{i=0}^{d+2} \fp_i^{(1)} x^i, \quad\quad Q^{(1)}(x)=\sum_{j=0}^{d+1} \fq_j^{(1)} x^j.
\end{split}
\end{equation}
Then $Q^{(1)}(x)\in \RR[x]$ and $P^{(1)},Q^{(1)}$ satisfy the conditions (1)-(3) in \cref{thm:existandunique} with $d$ replaced by $d-2$. 

Furthermore, if $d=2$, then 
\begin{equation}
P^{(1)}=-\fp_0^* e^{2\I \phi_0}=-\left(\frac{\fp_0}{\fp_2}\right)^*\fq_1, \quad Q^{(1)}=0.
\end{equation}

If $d\geq 3$, then 
\begin{equation}\label{eq:monomial_PQ_leading}
\begin{split}
    &\fp_{d-2}^{(1)}=\frac{\fq_{d-1}}{4}-\frac{1}{\fq_{d-1}}\abs{\Im[\fp_{d-2}e^{-2\I\phi_0}]}^2+ \I \Im[\fp_{d-2}e^{-2\I\phi_0}],\\
    &\fq_{d-3}^{(1)}=\frac{\fq_{d-1}}{4}+\frac{1}{\fq_{d-1}}\abs{\Im[\fp_{d-2}e^{-2\I\phi_0}]}^2.
\end{split}
\end{equation}
And $P^{(1)},Q^{(1)}$ also satisfy the condition (4) in \cref{thm:existandunique}.
\end{lemma}
\begin{proof}
See \cref{sec:proof_newPQ_coef}.
\end{proof}

For convenience we denote $\left(P^{(0)},Q^{(0)}\right)=(P,Q)$. Inspired by \cref{lma:newPQ_coef}, we choose $\phi_\ell$ to satisfy 
\begin{equation}
e^{2\I \phi_\ell}=\frac{\fp^{(\ell)}_{d-2\ell}}{\fq^{(\ell)}_{d-1-2\ell}}
\end{equation}
and repeat the construction 
\begin{equation}\label{eq:construct_newPQ}
\begin{split}
&\begin{pmatrix}
P^{(\ell+1)}(x) & \I Q^{(\ell+1)}(x)\sqrt{1-x^2}\\
\I Q^{(\ell+1)}(x)\sqrt{1-x^2} & \left(P^{(\ell+1)}\right)^* (x)
\end{pmatrix}\\
&:= W^{-1} e^{-\I\phi_\ell Z} \begin{pmatrix}
P^{(\ell)}(x) & \I Q^{(\ell)}(x)\sqrt{1-x^2}\\
\I Q^{(\ell)}(x) \sqrt{1-x^2} & \left(P^{(\ell)}\right)^* (x)
\end{pmatrix} e^{-\I\phi_\ell Z} W^{-1}.
\end{split}
\end{equation}
as long as $\mathrm{deg}\left(P^{(\ell)}\right)\geq 2$. Notice that 
\begin{equation*}
    \deg\left(P^{(\ell)}\right)=\deg\left(P^{(\ell-1)}\right)-2=\deg\left(P^{(0)}\right)-2\ell.
\end{equation*}
Thus, the iteration stops when $\ell=\wt{d}-1$. Moreover, $\deg(P^{(\wt{d}-1)})=0$ if $d$ is even, and
$\deg(P^{(\wt{d}-1)})=1$ if $d$ is odd. For convenience, we introduce the notation of $\hat{d}$,
\begin{equation}\label{eq:definition_hat_d}
    \hat{d}:=\begin{cases}
    \wt{d}-1 & \text{ if } d \text{ is odd},\\
    \wt{d}-2 & \text{ if } d \text{ is even}.
    \end{cases}
\end{equation}

\begin{lemma}\label{lma:phi_expression}
Suppose that $P(x), Q(x)$ satisfy the conditions in \cref{thm:existandunique}. Then there exists a set of symmetric phase factor $\Phi:=(\phi_0,\phi_1,\cdots,\phi_1,\phi_0) \in D_d$ such that \cref{eq:qsp-gslw} holds, where $\phi_\ell$ satisfies
\begin{equation}\label{eq:value_phi}
    e^{2\I \phi_{\ell}} =\frac{\fp^{(\ell)}_{d-2\ell}}{\fq^{(\ell)}_{d-1-2\ell}}, \quad \forall 0\leq \ell\leq \hat{d}.
\end{equation}
In particular, if $d$ is even, $\phi_{\wt{d}-1}$ satisfies 
\begin{equation}\label{eq:value_phi_end}
    e^{\I \phi_{\wt{d}-1}} =\fp^{(\wt{d}-1)}_0.
\end{equation}
Here  $\fp^{(\ell)}_{d-2\ell}$ and $\fq^{(\ell)}_{d-1-2\ell}$ are the leading coefficient of the expansion of $P^{(\ell)}$ and $Q^{(\ell)}$ in the monomial basis respectively. We use the notation $\left(P^{(0)},Q^{(0)}\right)=\left(P,Q\right)$ and define $(P^{(\ell)},Q^{(\ell)})$ by \cref{eq:construct_newPQ} for $1\leq \ell\leq \wt{d}-1$. 
\end{lemma}
\begin{proof}
By repeating the construction in \cref{eq:construct_newPQ}, we obtain a sequence of polynomials $\left(P^{(\ell)}, Q^{(\ell)}\right)$ where $0\leq \ell\leq \wt{d}-1$. \cref{eq:value_phi} determines the value of $\phi_{\ell}$ over $D_d$ for $0\leq \ell\leq \wt{d}-2$.

When $d$ is odd, $\deg \left(P^{(\wt{d}-1)}\right)=1$ and $\deg \left(Q^{(\wt{d}-1)}\right)=0$. \cref{lma:newPQ_coef} implies that the leading coefficient of $Q^{(\wt{d}-1)}$ is positive as well. By the normalization condition, $P^{(\wt{d}-1)}$ must be of the form $c x$ where $\abs{c}=1$ and $Q^{(\wt{d}-1)}$ must be 1. Choose $\phi_{\wt{d}-1}\in [-\frac{\pi}{2},\frac{\pi}{2})$ according to \cref{eq:value_phi} and one can check
\begin{equation*}
\begin{pmatrix}
P^{(\wt{d}-1)}(x) & \I Q^{(\wt{d}-1)}(x)\sqrt{1-x^2}\\
\I Q^{(\wt{d}-1)}(x)\sqrt{1-x^2} & \left(P^{(\wt{d}-1)}\right)^* (x)
\end{pmatrix} = e^{\I \phi_{\wt{d}-1}Z} W e^{\I \phi_{\wt{d}-1}Z}.
\end{equation*}
In this way, $\Phi=(\phi_0,\cdots,\phi_{\wt{d}-1}, \phi_{\wt{d}-1},\cdots,\phi_0)$ satisfies the requirements.

When $d$ is even, we know that $P^{(\wt{d}-1)}$ is a constant of absolute value $1$ and $Q^{(\wt{d}-1)}=0$, \ie, 
\begin{equation*}
\begin{pmatrix}
P^{(\wt{d}-1)}(x) & \I Q^{(\wt{d}-1)}(x)\sqrt{1-x^2}\\
\I Q^{(\wt{d}-1)}(x)\sqrt{1-x^2} & \left(P^{(\wt{d}-1)}\right)^* (x)
\end{pmatrix}  =e^{\I\phi_{\wt{d}-1}Z},
\end{equation*}
where $\phi_{\wt{d}-1}$ satisfies \cref{eq:value_phi_end}. Furthermore, we may choose $\phi_{\wt{d}-1}\in[-\pi,\pi)$ and $\Phi=(\phi_0,\cdots,\phi_{\wt{d}-2},\phi_{\wt{d}-1}, \phi_{\wt{d}-2},\cdots,\phi_0)$ satisfies the requirements.
\end{proof}

\subsection{Proof of \cref{thm:existandunique}}
Now we are ready to prove \cref{thm:existandunique}.
\begin{proof}[Proof of \cref{thm:existandunique}]
The existence is given by \cref{lma:phi_expression}. We only need to show the uniqueness. If there exists another set of symmetric phase factors $\Psi:=(\psi_0,\psi_1,\cdots,\psi_1,\psi_0)\in D_d$ satisfying the requirement, we may consider
\begin{equation*}
\begin{split}
&\begin{pmatrix}
P^{(1)}_\Psi(x) & \I Q^{(1)}_\Psi(x)\sqrt{1-x^2}\\
\I Q^{(1)}_\Psi(x)\sqrt{1-x^2} & \left(P^{(1)}_\Psi\right)^* (x)
\end{pmatrix}\\
&:= W^{-1} e^{-\I\psi_0 Z} \begin{pmatrix}
P(x) & \I Q(x)\sqrt{1-x^2}\\
\I Q(x) \sqrt{1-x^2} & P^* (x)
\end{pmatrix} e^{-\I\psi_0 Z} W^{-1}.
\end{split}
\end{equation*}
Since $P^{(1)}_\Psi(x)$ can be obtained by a set of symmetric phase factors of length $d-2$, the coefficient of $x^{d+2}$ of $P^{(1)}_\Psi$ is zero, \ie,
\begin{equation*}
\fp^{(1)}_{d+2}= -2 \fq_{d-1} + \fp_d^* e^{2 \I \psi_0} + \fp_d e^{-2\I  \psi_0} =0,
\end{equation*}
which is obtained by applying \cref{eq:tiledP}. Together with $\abs{\fp_d}=\abs{\fq_{d-1}}$, it leads to $e^{2\I \psi_0}= \frac{\fq_{d-1}}{\fp_d^*}$ and this equation has only one solution over $[-\frac{\pi}{2},\frac{\pi}{2})$, hence $\psi_0=\phi_0$. Then $\Phi=\Psi$ follows inductively.

\end{proof}

\subsection{Optimization based method for finding symmetric phase factors}

Ref.~\cite{DongMengWhaleyEtAl2021} suggests that the optimization problem \cref{eqn:optprob-intro} should be solved by imposing the symmetry constraint on the phase factors. This allows us to use the reduced phase factors $\wt{\Phi}=(\phi_0,\phi_1,\cdots,\phi_{\wt{d}-1})$, as optimization variables.  For any scalar valued function of the set of reduced phase factors, we use subscripts to represent the partial derivative with respect to the corresponding phase factor, for example, $F_i(\wt{\Phi}) := \partial_{\phi_i} F(\wt{\Phi})$. We also use subscripts to represent the corresponding entry of any matrix. These notations and conventions can simplify the discussion in \cref{{sec:local_convergence},{sec:local_strong_convexity}}.

The Hessian matrix of the problem in \cref{eqn:optprob-intro} can be written \REV{element-wise} as
\begin{equation}\label{eq:hess-proof}
\begin{split}
    F_{ij}(\wt{\Phi})&=\frac{2}{\wt{d}}\sum_{k=1}^{\wt{d}} g_i(x_k,\wt{\Phi})g_j(x_k,\wt{\Phi})+\frac{2}{\wt{d}}\sum_{k=1}^{\wt{d}} [g(x_k,\wt{\Phi})-f(x_k)]g_{ij}(x_k,\wt{\Phi})\\
    &=\frac{2}{\wt{d}}(A^\top A)_{ij}+R_{ij},
\end{split}
\end{equation}
where $A(\wt{\Phi}) := [g_j(x_i,\wt{\Phi})]\in \RR^{\wt{d}\times \wt{d}}$ is the Jacobian matrix of $\wt{\Phi} \mapsto \left(g(x_1,\wt{\Phi}),\dots,g(x_{\wt{d}},\wt{\Phi})\right)^\top$, and $R$ is a real $\wt{d}\times \wt{d}$ matrix closely related to the residual with entry (note that $R_{ij}$ does not mean the partial derivative here) 
\begin{displaymath}
R_{ij}(\wt{\Phi}):=\frac{2}{\wt{d}}\sum_{k=1}^{\wt{d}} [g(x_k,\wt{\Phi})-f(x_k)]g_{ij}(x_k,\wt{\Phi}).
\end{displaymath}
Due to the existence of phase factors stated in \cref{thm:existandunique}, at the optimal point $\wt{\Phi}^*$, the second term vanishes, namely $R(\wt{\Phi}^*) = 0$, and $\text{Hess}(\wt{\Phi}^*)=\frac{2}{\wt{d}}(A^*)^\top A^*$, where $A^*:=A(\wt{\Phi}^*)$ and $\wt{\Phi}^*$ are the reduced phase factors corresponding to symmetric phase factors $\Phi^*$.

\section{Global minimizer}
According to the last section, the problem of finding all global minima of the optimization problem is converted to that of finding all admissible pairs of $(P,Q)$ associated with the target polynomial $f$. In this section, we provide a characterization of \emph{all} possible complementary polynomials $(P_\Im,Q)$.

\subsection{Characterizing complementary polynomials}
Before constructing all possible pairs of $(P(x),Q(x))$ satisfying \cref{con:PQ}, we first introduce the following lemma, which views the problem from the perspective of Laurent polynomials.
\begin{lemma}\label{lma:unique_decom}
Given any real polynomial $f(x)$ of degree $d$ and with $\norm{f}_{\infty}<1$, $\mf{F}(z)=1-\left(f\left(\frac{z+z^{-1}}{2}\right)\right)^2$ has a decomposition \begin{equation}
    \mf{F}(z)= \alpha  \prod_{i=1}^{2d} (z-r_i)(z^{-1}-r_i),
\end{equation}
such that $\alpha, r_i \in \CC\backslash\{0\}$. 

The decomposition is unique in the following sense: if $\wt{\alpha }\in \CC$ and $\wt{r}_i\in \CC$, $i=1,\cdots,n$ such that
\begin{equation}
    \mf{F}(z)= \wt{\alpha} \prod_{i=1}^n (z-\wt{r}_i) (z^{-1}-\wt{r}_i),
\end{equation}
then 
\begin{equation}\label{eq:root_permutation}
    n=2d, \qquad\alpha \prod_{i=1}^{2d} r_i = \wt{\alpha} \prod_{i=1}^{2d} \wt{r}_i, \qquad     \frac{r_i+r_i^{-1}}{2}=\frac{\wt{r}_i+\wt{r}_i^{-1}}{2},
\end{equation}
up to permutation of the roots.
\end{lemma}
\begin{proof}
According to the fundamental theorem of algebra, $1-f(x)^2$ has a unique decomposition 
\begin{equation*}
    1-f(x)^2= c\prod_{i=1}^{2d} (s_i -x), \quad s_i \in \CC.
\end{equation*}
Since $\norm{f}_{\infty}<1$, $s_i \ne 0$. The equation $s=\frac{z+z^{-1}}{2}$ always has 2 nonzero solutions, say $z_1, z_2$. In particular, $z_1= z_2^{-1}$. For any $s_i$, pick one solution $r_i$,  use equality
\begin{equation*}
    r_i(s_i-x)=\frac{1}{2}(z-r_i)(z^{-1}-r_i),
\end{equation*}
and one gets
\begin{equation*}
    \mf{F}(z)=c \prod_{i=1}^{2d}\frac{1}{2r_i}(z-r_i)(z^{-1}-r_i)=\alpha \prod_{i=1}^{2d} (z-r_i)(z^{-1}-r_i),
\end{equation*}
where $\alpha=c\prod_{i=1}^{2d} \frac{1}{2r_i}$. It provides a decomposition of $\mf{F}(z)$ satisfying the requirement.

On the other hand, given any decomposition
\begin{equation*}
    \mf{F}(z)= \wt{\alpha} \prod_{i=1}^n (z-\wt{r}_i) (z^{-1}-\wt{r}_i),
\end{equation*}
then one has
\begin{equation*}
    \mf{F}(z)= \wt{\alpha} \prod_{i=1}^n (2\wt{r}_i) \left(\frac{\wt{r}_i + \wt{r}_i^{-1}}{2}-\frac{z+z^{-1}}{2}\right).
\end{equation*}
It provides a decomposition of $1-f(x)^2$, \ie,
\begin{equation*}
    1-f(x)^2= \wt{\alpha} \left(\prod_{i=1}^n 2\wt{r}_i\right) \prod_{i=1}^n\left(\frac{\wt{r}_i + \wt{r}_i^{-1}}{2}-x\right).
\end{equation*}

By the uniqueness decomposition of $1-f(x)^2$ in $\CC[x]$, we obtain \cref{eq:root_permutation} up to permutation of the roots. 
\end{proof}

In particular, we are interested in the following decomposition of $\mf{F}(z)$, which is referred to as an \emph{admissible decomposition},
\begin{itemize}
    \item $\mf{F}(z)= \alpha  \prod_{i=1}^{2d} (z-r_i)(z^{-1}-r_i)$, $\alpha>0, r_i\in \CC\backslash\{0\}$.
    \item Multiset $\mathcal{D}:=\{r_i: 1\leq i\leq 2d\}$ is closed under complex conjugation and additive inverse.
\end{itemize}
Firstly, $\mathcal{D}$ contains half of the roots of $\mf{F}(z)$ with multiplicity (hence $\mathcal{D}$ is a multiset rather than set). Secondly, note that $\alpha>0$ naturally comes from that $\alpha=\frac{1-f(1)^2}{\prod_{r\in \mathcal{D}} (1-r)^2}>0$ for any $\mathcal{D}$ closed under complex conjugation. Another observation is that the roots of $\mf{F}(z)$ cannot be any number of absolute value 1. Otherwise, $1-f\left(\frac{r_i +r_i^{-1}}{2}\right)^2=0$, \ie, $\norm{f}_{\infty}=1$, which contradicts the constraint on the maxnorm of $f$. 

To be more precise, we rephrase \cref{thm:admissible_pair} as follows and present its proof.
\begin{theorem}\label{thm:characterize_PQ}
The pair of polynomials $\left(P(x),Q(x)\right)$ satisfy \cref{con:PQ} if only if there exists an admissible decomposition of $\mf{F}(z)$,
\begin{equation}
    \mf{F}(z)= \alpha  \prod_{i=1}^{2d} (z-r_i)(z^{-1}-r_i),
\end{equation}
such that 
\begin{equation}\label{eq:PQformula}
\begin{split}
    P_\Im\left(x\right)&=\sqrt{\alpha} \frac{e\left(x+\I \sqrt{1-x^2}\right)+e\left(x-\I\sqrt{1-x^2}\right)}{2},\\ \quad Q\left(x\right)& =\sqrt{\alpha} \frac{e\left(x+\I \sqrt{1-x^2}\right)-e\left(x-\I\sqrt{1-x^2}\right)}{2\I \sqrt{1-x^2}},
\end{split}
\end{equation}
where $e(z):=z^{-d} \prod_{i=1}^{2d} (z-r_i)$.
\end{theorem}

\begin{proof}
$\Longrightarrow$: Denote $p_\Im(z):=P_\Im\left(\frac{z+z^{-1}}{2}\right)$ and $q(z):=Q\left(\frac{z+z^{-1}}{2}\right)$. Then the normalization condition gives
\begin{equation*}
    \begin{split}
        \mf{F}(z)&=p_\Im(z)^2-q(z)^2\left(\frac{z-z^{-1}}{2}\right)^2\\
        &=\left(p_\Im(z)+\frac{z-z^{-1}}{2}q(z)\right)\left(p_\Im(z)-\frac{z-z^{-1}}{2}q(z)\right).
    \end{split}
\end{equation*}
Define $H(z):=p_\Im(z)+\frac{z-z^{-1}}{2}q(z)$. Since $p_\Im(z)=p_\Im(z^{-1})$ and $q(z)=q(z^{-1})$, one gets 
\begin{equation*}
    \mf{F}(z)=H(z)H(z^{-1}).
\end{equation*}
Notice that $z^d H(z)$ is a real polynomial of degree no more than  $2d$. Thus $z^d H(z)$ can be represented as $c \prod_{r\in\mathcal{D}}(z-r)$, where $c\in\CC$ and the multiset $\mathcal{D}$ contains all the roots of $z^d H(z)$ with multiplicity. Since $H(z)$ is real, $\mathcal{D}$ is closed under complex conjugation and $c\in \RR$.

With $H(z)=c z^{-d}\prod_{r\in\mathcal{D}}(z-r)$, one gets
\begin{equation*}
    \mf{F}(z)=H(z)H(z^{-1})=c^2 \prod_{r\in\mathcal{D}}(z-r)(z^{-1}-r).
\end{equation*}
By \cref{lma:unique_decom}, the cardinality of $\mathcal{D}$ is $2d$ and thus $z^{d}H(z)$ is of degree $2d$. Moreover, $r\in \CC\backslash\{0\}$, $\forall r\in \mathcal{D}$.

Besides, we can check that the parity of $z^dH(z)$ is even by examining the parity of $p_\Im(z)$ and $q(z)$. Thus $\mathcal{D}$ is closed under additive inverse and \cref{eq:PQformula} holds as well.

$\Longleftarrow$:
Denote $p_\Im(z):=\sqrt{\alpha} \frac{e(z)+e(z^{-1})}{2}$ and $q(z):=\sqrt{\alpha} \frac{e(z)-e(z^{-1})}{z-z^{-1}}$. Then
\begin{equation*}
\begin{split}
    \mf{F}(z)&= \alpha  \prod_{i=1}^{2d} (z-r_i)(z^{-1}-r_i)=\alpha e(z) e(z^{-1})\\
    &=\left(\sqrt{\alpha} \frac{e(z)+e(z^{-1})}{2}\right)^2-\left(\sqrt{\alpha} \frac{e(z)-e(z^{-1})}{2}\right)^2\\
    &=\left(\sqrt{\alpha} \frac{e(z)+e(z^{-1})}{2}\right)^2+\left(1-\left(\frac{z+z^{-1}}{2}\right)^2\right) \left(\sqrt{\alpha} \frac{e(z)-e(z^{-1})}{z-z^{-1}}\right)^2\\
    &=p_\Im(z)^2+\left(1-\left(\frac{z+z^{-1}}{2}\right)^2\right)q(z)^2.
\end{split}
\end{equation*}

Since $\mathcal{D}$ is closed under complex conjugation, $\prod_{i=1}^{2d} (z-r_i)$ is a real polynomial and $e(z)=z^d +\sum_{i=1-d}^{d-1}e_i z^i+\left(\prod_{i=1}^{2d} r_i\right)z^{-d}$ for some $e_i\in \RR$. Notice that for $k>0$
\begin{equation*}
    \begin{split}
        &\left(x+\I \sqrt{1-x^2}\right)^k+\left(x+\I \sqrt{1-x^2}\right)^{-k}= 2 T_k(x),\\
        &\left(x+\I \sqrt{1-x^2}\right)^k-\left(x+\I \sqrt{1-x^2}\right)^{-k}= 2\I U_{k-1}(x)\sqrt{1-x^2}.\\
    \end{split}
\end{equation*}
We can verify that $p_\Im(x+\I \sqrt{1-x^2})$ and $q(x+\I \sqrt{1-x^2})$ are real polynomials of $x$. Hence 
\begin{equation*}
    P_\Im (x):= p_\Im(x+\I \sqrt{1-x^2}),\quad\quad Q(x):=q(x+\I \sqrt{1-x^2}).
\end{equation*}
are well defined and $\left(P(x)=f(x)+\I P_\Im(x), Q(x)\right)$ satisfies the normalization condition,
\begin{equation*}
    1-f(x)^2=P_\Im(x)^2+(1-x^2)Q(x)^2.
\end{equation*}

Since $\mathcal{D}$ is closed under additive inverse,
\begin{equation*}
    e(-z)=(-z)^{-d} \prod_{i=1}^{2d} (-z-r_i)=(-1)^dz^{-d} \prod_{i=1}^{2d} (z-r_i)=(-1)^d e(z).
\end{equation*}
It leads to that $P_\Im(x)$ has parity $(d \bmod 2)$ and $Q(x)$ has parity $(d-1 \bmod 2)$.

Another important observation is that  $\prod_{r\in \mathcal{D}}r <0$. It comes from
\begin{equation}\label{eq:alpha_and_f}
    \mf{F}(z)= \alpha  \prod_{i=1}^{2d} (z-r_i)(z^{-1}-r_i)= \alpha \left( \prod_{i=1}^{2d}r_i\right) z^{-2d} \prod_{i=1}^{2d} (z-r_i)(z-r_i^{-1})
\end{equation}
and that the coefficient of $z^{2d}$ in $\mf{F}(z)$ is negative. According to \cref{eq:alpha_and_f}, one has that the coefficient of $z^{2d}$ in $\mf{F}(z)$ is $\alpha\left( \prod_{i=1}^{2d}r_i\right)$. Therefore,
\begin{equation}\label{eq:sign_K}
    \alpha\left( \prod_{i=1}^{2d}r_i\right)<0\Longrightarrow \left( \prod_{i=1}^{2d}r_i\right)<0.
\end{equation}
Since
\begin{equation*}
    \begin{split}
        \frac{e(z)-e(z^{-1})}{z-z^{-1}}&= \frac{1}{z-z^{-1}}\left[\left(1-\prod_{i=1}^{2d} r_i\right)(z^d-z^{-d}) +\sum_{i=1-d}^{d-1}e_i (z^i-z^{-i})\right]\\
        &= \left(1-\prod_{i=1}^{2d} r_i\right)\left( z^{d-1}+z^{-d+1}\right)+l.o..
    \end{split}
\end{equation*}
Here, we use $l.o.$ to denote any linear combination of $z^j$ with coefficients in $\CC$ and $\abs{j}<d-1$. By comparing with the coefficient of $x^{d-1}$ in the monomial expansion for $Q(x)$, $\fq_{d-1}$, one gets
\begin{equation}\label{eq:Q_lead_coef}
   \sqrt{\alpha}\left(1-\prod_{i=1}^{2d} r_i\right)= 2^{1-d}\fq_{d-1}.
\end{equation}
In this way, $\text{deg} (Q)=d-1$, $\text{deg}(P)=d$ and the leading coefficient of $Q(x)$ is positive.
\end{proof}

\begin{remark}[Constructing complementary polynomials]\label{re:construct_method}
\cref{thm:characterize_PQ} indicates that there exists a bijection between the admissible decomposition of $\mf{F}(z)$ and the admissible pair $(P,Q)$ satisfying \cref{con:PQ}. Furthermore, it provides a constructive way to find all admissible complementary polynomials:
\begin{enumerate}
    \item Choose a multiset $\mathcal{D}$ such that 
    \begin{enumerate}
      \item the joint union of $\mathcal{D}$ and its reciprocal contains all roots of $\mf{F}(z)$ with multiplicity.
        \item $\mathcal{D}$ is closed under complex conjugation and additive inverse.
    \end{enumerate}
    \item Construct $P_\Im $ and $Q$ according to \cref{eq:PQformula}.
\end{enumerate}
\end{remark}

\subsection{Maximal solution}\label{sec:maximal_sol}
Among all possible choices of $\mathcal{D}$, we are interested in choosing $\mathcal{D}$ to contain all roots of $\mf{F}(z)$ within the unit disc and outside the unit disc. 
As will be shown below, the associated complementary polynomials satisfy many desirapble properties. Let $\left(P_\Im,Q\right)$, $\left(\bar{P}_\Im, \bar{Q}\right)$ denote the corresponding pair of complementary polynomials respectively. Moreover, throughout the paper, unless otherwise noted, we use  $\left(\bar{P}_\Im, \bar{Q}\right)$ to denote the complementary polynomials constructed according to \cref{eq:PQformula}, where $\mathcal{D}$ contains all roots of $\mf{F}(z)$ outside the unit disc.

\begin{lemma}[Maximal magnitude of leading coefficients]
\label{lem:maximal_leading}
The leading coefficients of the expansion of $Q$ and $\bar{Q}$ in the monomial basis have the largest magnitude among all possible complementary polynomials.
\end{lemma}
\begin{proof}
Given any multiset $\mathcal{D}$ satisfying the condition mentioned in \cref{re:construct_method}, note that $\mf{F}(z)$ can be represented as
\begin{equation*}
    \mf{F}(z)= \beta z^{-2d}\prod_{r\in \mathcal{D}}(z-r)(z-r^{-1})
\end{equation*}
for some $\beta\in \CC$. By checking the leading coefficient, $\beta = -\frac{\ff^2_d}{4^{d}}$ is a quantity fully determined by $f$.

Recall the proof of \cref{thm:characterize_PQ}, we know that $\alpha \prod_{r\in \mathcal{D}} r = \beta $ by \cref{eq:alpha_and_f}. Together with \cref{eq:Q_lead_coef}, we obtain 
\begin{equation}\label{eq:q_d-1}
    \fq_{d-1}^2 = 4^{d-1} \abs{\beta}\left(K+\frac{1}{K}+2\right),
\end{equation}
where $K:=-\prod_{r\in \mathcal{D}}r>0$ by \cref{eq:sign_K}.

When $\mathcal{D}$ contains all the roots within the unit disc or outside the unit disc, $\frac{1}{K}+K$ is the largest among all the possibilities, that is, the leading coefficient of $Q(x)$ and $\bar{Q}(x)$ has the largest magnitude. 
\end{proof}
The following lemma shows that $\left(P_\Im,Q\right)$ and $\left(\bar{P}_\Im,\bar{Q}\right)$ are directly related to each other.
\begin{lemma}\label{lma:relation_maximal_sol}
The polynomials $\left(P_\Im,Q\right)$ and $\left(\bar{P}_\Im,\bar{Q}\right)$ are related to each other:
\begin{equation}
    P_\Im(x)=-\bar{P}_\Im (x),\quad\quad Q(x)=\bar{Q}(x).
\end{equation}
\end{lemma}
\begin{proof}
Denote $\mathcal{D}:=\left\{r \text{ is a root of } \mf{F}(z): \abs{r}<1\right\}$, and 
\begin{equation*}
    e(z):=z^{-d}\prod_{r\in \mathcal{D}}(z-r),\quad \bar{e}(z):=z^{-d}\prod_{r\in \mathcal{D}}(z-r^{-1}).
\end{equation*}
Then for some $\alpha, \bar{\alpha}>0$,
\begin{equation}
    \mf{F}(z)= \alpha e(z) e(z^{-1}) =\bar{\alpha} \bar{e}(z) \bar{e}(z^{-1}).
\end{equation}
Notice that 
\begin{equation}
    \begin{split}
        \bar{e}(z)&=z^{-d}\prod_{r\in \mathcal{D}}(z-r^{-1})=z^{-d}\prod_{r\in \mathcal{D}}\left(r^{-1}(rz-1)\right)=z^{d}\prod_{r\in \mathcal{D}}\left(r^{-1}(r-z^{-1})\right)\\
        &=\left(\prod_{r\in \mathcal{D}} r^{-1}\right) z^{-d}\prod_{r\in \mathcal{D}}(z^{-1}-r)=-\frac{1}{K} e(z^{-1}).
    \end{split}
\end{equation}
where $K=-\prod_{r\in \mathcal{D}} r$. Similarly, $\bar{e}(z^{-1})=-\frac{1}{K} e(z)$. Thus, $\bar{\alpha}=K^2 \alpha$. By \cref{thm:characterize_PQ}, we have
\begin{equation}
    \begin{split}
    \bar{P}_\Im\left(x\right)&=\sqrt{\bar{\alpha}} \frac{\bar{e}\left(x+\I \sqrt{1-x^2}\right)+\bar{e}\left(x-\I\sqrt{1-x^2}\right)}{2},\\ 
    &=-\sqrt{\alpha K^2} \frac{1}{K}\frac{e\left(x-\I \sqrt{1-x^2}\right)+e\left(x+\I\sqrt{1-x^2}\right)}{2}= -P_\Im(x),
    \end{split}
\end{equation}
and 
\begin{equation}
    \begin{split}
        \bar{Q}\left(x\right)& =\sqrt{\bar{\alpha}} \frac{\bar{e}\left(x+\I \sqrt{1-x^2}\right)-\bar{e}\left(x-\I\sqrt{1-x^2}\right)}{2\I \sqrt{1-x^2}}\\
        &=-\sqrt{\alpha K^2} \frac{1}{K}\frac{e\left(x-\I \sqrt{1-x^2}\right)-e\left(x+\I\sqrt{1-x^2}\right)}{2\I \sqrt{1-x^2}}=Q(x).
    \end{split}
\end{equation}
This proves the lemma.
\end{proof}

Due to this relation, we may only focus on the admissible pair $\left(P(x),Q(x)\right)$ obtained by choosing $\mathcal{D}$ to contain all the roots of $\mf{F}(z)$ within the unit disc. 
Due to the property in \cref{lem:maximal_leading}, the corresponding phase factors in $D_d$ will be referred to as the maximal solution. 

\begin{remark}[Maximal solution at $f=0$]
\cref{lma:relation_maximal_sol} agrees with the observation that two particular admissible pairs associated with $f=0$ are $\left(\I T_d(x),U_{d-1}(x)\right)$ and $\left(-\I T_d(x),U_{d-1}(x)\right)$, which are obtained by 
\begin{equation}
    \Phi^0:=\left(\frac{\pi}{4}, 0,\cdots,0,\frac{\pi}{4}\right)\quad \bar{\Phi}^0:=\left(-\frac{\pi}{4}, 0,\cdots,0,-\frac{\pi}{4}\right)
\end{equation}
respectively. In the following discussion, we will show that the maximal solution  will converge to $\Phi^0$ as $\norm{f}_{\infty} \to 0$. The same argument is suitable for $\left(\bar{P}_\Im,\bar{Q}\right)$ with $\Phi^0$ replaced by $\bar{\Phi}^0$.
\end{remark}
\begin{remark}\label{re:positive_lead_coef_Pim}
There is a useful observation that the coefficient of $x^d$ for $P_\Im$ is positive. By \cref{eq:PQformula}, we have
\begin{equation}
\begin{split}\label{eq:pim_coef}
    P_{\Im}(x)&=\sqrt{\alpha}\frac{e\left(x+\I \sqrt{1-x^2}\right)+e\left(x-\I \sqrt{1-x^2}\right)}{2}\\
    &=\sqrt{\alpha}\frac{1+\prod_{r\in \mathcal{D}}r}{2}\left[\left(x+\I \sqrt{1-x^2}\right)^d+\left(x-\I \sqrt{1-x^2}\right)^d\right]+l.o.\\
    &=\sqrt{\alpha}\left(1+\prod_{r\in \mathcal{D}}r\right)T_d(x)+l.o..
\end{split}
\end{equation}
Here $l.o.$ represents any polynomial of degree $<d$. To achieve the second equality, we use the expansion of $e(z)=z^{-d}\prod_{r\in \mathcal{D}} (z-r)$ in the monomial basis and discover that only terms $z^d$ and $\left(\prod_{r\in \mathcal{D}}r \right)z^{-d}$ contribute to the first term on the second line, if we replace $z,z^{-1}$ with $x+\I \sqrt{1-x^2}, x-\I \sqrt{1-x^2}$ respectively. 

The specific choice of $\mathcal{D}$ yields $\left(1+\prod_{r\in \mathcal{D}}r\right)>0$. Hence the leading coefficient of $P_\Im$ is positive.
\end{remark}

\subsection{Relation to other constructive methods}
There are two other methods for constructing complementary polynomials, which will be referred to as the GSLW method~\cite{GilyenSuLowEtAl2019} and the Haah method~\cite{Haah2019}, respectively. 
The two methods are equivalent.  The GSLW method solves the problem without using Laurent polynomials. The Haah method first points out that the decomposition can be viewed more naturally in terms of Laurent polynomials.

The Haah method considers 
\begin{equation}\label{eq:Haah_fun}
    1-P_\Re\left(\frac{z+z^{-1}}{2}\right)^2+\left(\frac{z-z^{-1}}{2}\right)^2Q_\Im\left(\frac{z+z^{-1}}{2}\right)^2
\end{equation}
 and the corresponding decomposition:
\begin{equation}
    1-P_\Re\left(\frac{z+z^{-1}}{2}\right)^2+\left(\frac{z-z^{-1}}{2}\right)^2Q_\Im\left(\frac{z+z^{-1}}{2}\right)^2= \alpha  \prod_{r\in \mathcal{D}} (z-r)(z^{-1}-r)
\end{equation}
where $\mathcal{D}$ contains all the roots within the unit disc. 

In our construction, we may omit the term $Q_\Im$ due to the symmetry constraint. However, one can check that our method is also suitable for the case that $Q_\Im\ne 0$ by simply replacing $\mf{F}(z)$ by \cref{eq:Haah_fun}. By the same method as that in \cref{thm:characterize_PQ}, it can shown that any admissible complementary polynomial $(P_\Re,Q_\Re)$ must take the form of \cref{eq:PQformula}, where $Q$ is replaced by $Q_{\Re}$. Hence our method can be generalized to find all admissible complementary polynomials, by going through all possible $e(z)=z^{-d} \prod_{r\in \mathcal{D}}(z-r)$. In particular,  the Haah method considers the special case that $\mathcal{D}$ contains all the root of $\mf{F}(z)$ within the unit disc, and therefore does not generate all admissible complementary polynomials. 

We would like to clarify that Ref \cite{Haah2019} states that we can construct $P_\Im$ and $Q$ as long as the joint union  of a conjugate-closed  multiset $\mathcal{D}$ and its reciprocal $\{\frac{1}{r}:r\in \mathcal{D}\}$ is the list of all the roots of \cref{eq:Haah_fun}, if there is no parity requirement for $P_\Im$ and $Q$. After imposing the parity condition, according to \cref{thm:characterize_PQ}, another necessary  condition for suitable complementary polynomials is that $\mathcal{D}$ is closed under additive inverse. Without it, the parity condition may not hold. This can be verified by the proof of \cref{thm:characterize_PQ}.

\section{Local convergence}\label{sec:local_convergence}
We prove \cref{thm:estimate_maximal} in this section. Throughout this section, polynomials are expanded in the Chebyshev basis without otherwise noted. This is more convenient than the expansion in the monomial basis.

According to the definition of the Chebyshev norms in \cref{sec:chebyshev} and the convention of the coefficients in \cref{sec:notation}, we specifically list the following relations which exploit the weighted orthogonality and will be frequently used in the proof,
\begin{equation}
\begin{split}
    \norm{f}_T^2=f^2_0+\frac{1}{2}\sum_{k=1}^d f^2_k,\quad\quad\norm{P_\Im}_T^2=s_0^2+\frac{1}{2}\sum_{k=1}^d s_k^2,\quad\quad\norm{Q}_U^2=\frac{1}{2}\sum_{k=0}^{d-1} {q_k}^2.
\end{split}
\end{equation}
By integrating  the point-wise normalization condition $f(x)^2+P_\Im(x)^2+(1-x^2)Q(x)^2=1$ with the weight function $\frac{1}{\sqrt{1-x^2}}$, we obtain the following Parseval's equality:
\begin{equation}\label{eqn:parseval}
f^2_0+s_0^2+\frac{1}{2}\sum_{k=1}^{d} \left(f^2_k+s_k^2\right)+\frac{1}{2}\sum_{k=0}^{d-1} {q_k}^2 =1.
\end{equation}

\subsection{Sketch of the proof}
We first give a sketch of the proof as a road map. For simplicity, we assume $d$ is odd in the sketch, while the proof works for any $d$.

\begin{enumerate}
    \item We first show that each element of $\wt{\Phi}^*-\wt{\Phi}^0$ is bounded by $\frac{\pi}{4}$ by properly restricting $\norm{f}_{\infty}$. 
    \item We further upper bound the distance $\norm{\wt{\Phi}^*-\wt{\Phi}^0}_2$ by $\abs{\sin\left(2\phi_0-\frac{\pi}{2}\right)}^2 + \sum_{\ell=1}^{\wt{d}-1} \abs{\sin(2\phi_{\ell})}^2$. It follows the inequality $\abs{x}\leq \frac{\pi}{2}\abs{\sin(x)}$ which holds for any $\abs{x}\leq \frac{\pi}{2}$. The bound in the first step ensures that the inequality can be invoked.
    \item We upper bound the phase factors by the leading Chebyshev coefficients in the recursion sequence of polynomials. The upper bound turns out to be
    \begin{equation*} 
    \begin{split}
        &\abs{\sin\left(2\phi_0-\frac{\pi}{2}\right)}^2\leq \frac{8}{3}\norm{f}_{\infty}^2, \\
        &\abs{ \sin(2 \phi_{\ell})}^2\leq \frac{8}{3}\left(\left({q^{(\ell)}_{d-1-2\ell}}\right)^2-\left({q^{(\ell-1)}_{d+1-2\ell}}\right)^2\right),\quad 1\leq \ell\leq \wt{d}-1.
    \end{split}
    \end{equation*}
    \item Combining the previous results, we finally get $\norm{\wt{\Phi}^*-\wt{\Phi}^0}_2^2 \le \frac{\pi^2}{6}\left(\norm{f}_{\infty}^2+1-q_{d-1}^2\right)$. The upper bound can be further simplified to be $\frac{\pi^2}{3} \norm{f}_\infty^2$ by exploiting the estimate to $q_{d-1}$ (see \cref{lma:coef_q_bound}). Then the statement in \cref{thm:estimate_maximal} follows.
\end{enumerate}

\subsection{Recursion relation in terms of Chebyshev coefficients}
In this subsection, we will convert the recursion relation in terms of monomial coefficients to that in terms of Chebyshev coefficients. Furthermore, we will derive some useful properties which will be used as intermediate steps in the proof of \cref{thm:estimate_maximal}.

Let $(P, Q) \in \CC_d[x] \times \RR_{d-1}[x]$ be a pair of polynomials satisfying the conditions in \cref{thm:existandunique}. The symmetric reduction procedures yield a sequence of polynomials $(P^{(\ell)}, Q^{(\ell)}) \in \CC_{d-2\ell}[x] \times \RR_{d-2\ell-1}[x]$ satisfying the conditions and an associated sequence of phase factors $\phi_\ell$ so that $e^{2\I \phi_\ell}=\frac{\fp^{(\ell)}_{d-2\ell}}{\fq^{(\ell)}_{d-1-2\ell}}$. In particular, $(P^{(0)}, Q^{(0)})=(P, Q)$. An important recursion relation generating the coefficients of new polynomials in the monomial basis is given in \cref{eq:monomial_PQ_leading}. For convenience, we restate it as follows: if $d\geq 3$, then 
\begin{equation*}
\begin{split}
    &\fp_{d-2}^{(1)}=\frac{\fq_{d-1}}{4}-\frac{1}{\fq_{d-1}}\abs{\Im[\fp_{d-2}e^{-2\I\phi_0}]}^2+ \I \Im[\fp_{d-2}e^{-2\I\phi_0}],\\
    &\fq_{d-3}^{(1)}=\frac{\fq_{d-1}}{4}+\frac{1}{\fq_{d-1}}\abs{\Im[\fp_{d-2}e^{-2\I\phi_0}]}^2.
\end{split}
\end{equation*}
Here, recall that the boldface symbols refer to the coefficients in the monomial basis. We will first rewrite this relation using the Chebyshev basis. As a remark, it suffices to consider the case when $\ell = 0$, and the general case can be recovered by substituting $(P, Q, P^{(1)}, Q^{(1)}, \phi_0) \leftarrow (P^{(\ell)}, Q^{(\ell)}, P^{(\ell+1)}, Q^{(\ell+1)}, \phi_\ell)$ in the arguments as long as $\deg\left(P^{(\ell)}\right)\geq 3$.

For simplicity, let us denote the Chebyshev polynomial in the monomial basis as $T_d(x) = 2^{d-1} x^d + t_d x^{d-2} + l.o.$ where we use $l.o.$ to denote any polynomial of degree no more than $d-3$ and $t_d$ is a real number whose exact value is irrelevant in the proof.

Expanding the polynomials in mononial basis and Chebyshev basis, we have
\begin{equation*}
\begin{split}
    & P(x)e^{-2\I\phi_0} =\sum_{k=0}^d \fp_{k}e^{-2\I\phi_0} x^k = e^{-2\I\phi_0}\fp_{d} x^d + e^{-2\I\phi_0}\fp_{d-2} x^{d-2} + l.o.\\
    &=\sum_{k=0}^d p_k e^{-2\I \phi_0} T_k(x) = e^{-2\I\phi_0} 2^{d-1} p_d x^d + \left(e^{-2\I\phi_0}p_dt_d + e^{-2\I\phi_0}p_{d-2}2^{d-3}\right) x^{d-2} + l.o..
\end{split}
\end{equation*}
Then, by comparing the coefficients of terms $x^d$ and $x^{d-2}$, we have $\fp_{d} = 2^{d-1} p_d $ and
\begin{equation*}
\begin{split}
    e^{-2\I\phi_0} \fp_{d-2} = e^{-2\I\phi_0}p_dt_d + e^{-2\I\phi_0}p_{d-2}2^{d-3} = q_{d-1} t_d + e^{-2\I\phi_0}p_{d-2}2^{d-3}.
\end{split}
\end{equation*}
In the last equality, we use
\begin{equation*}
    e^{2\I\phi_0} = \frac{\fp_d}{\fq_{d-1}} = \frac{2^{d-1} p_d}{2^{d-1} q_{d-1}} = \frac{p_d}{q_{d-1}} \Rightarrow e^{-2\I\phi_0} p_d = q_{d-1} \in \RR.
\end{equation*}
An important observation is
\begin{equation}
    \Im[\fp_{d-2}e^{-2\I \phi_0}]= 2^{d-3} \Im[p_{d-2} e^{-2\I \phi_0}],
\end{equation}
which follows $q_{d-1} t_d \in \RR$.

Plugging the previous equations in terms of Chebyshev coefficients in \cref{eq:monomial_PQ_leading}, we have the following recursion relation
\begin{equation}\label{eq:chebyshev_PQ_leading}
\begin{split}
    &p_{d-2}^{(1)}=q_{d-1}-\frac{1}{4q_{d-1}}\abs{\Im[p_{d-2}e^{-2\I\phi_0}]}^2+ \I \Im[p_{d-2}e^{-2\I\phi_0}],\\
    &q_{d-3}^{(1)}=q_{d-1}+\frac{1}{4q_{d-1}}\abs{\Im[p_{d-2}e^{-2\I\phi_0}]}^2.
\end{split}
\end{equation}
The second equality implies the following sign preserving property and the monotonicity
\begin{equation*}
    \text{sgn}(q_{d-1}) = \text{sgn}(q_{d-3}^{(1)}) \text{, and } \abs{q_{d-1}} \le \abs{q_{d-3}^{(1)}}.
\end{equation*}

Regarding the next phase factor $\phi_1$, its defined equation is converted similarly to Chebyshev basis
\begin{equation*}
e^{2\I \phi_{1}} =\frac{\fp^{(1)}_{d-2}}{\fq^{(1)}_{d-3}}=\frac{p^{(1)}_{d-2}}{q^{(1)}_{d-3}}.
\end{equation*}
Furthermore, we can use the derived recursion relation in \cref{eq:chebyshev_PQ_leading} to obtain
\begin{equation*}
    \sin\left(2\phi_1\right)=\Im\left[\frac{p^{(1)}_{d-2}}{q^{(1)}_{d-3}}\right]= \frac{\Im[p_{d-2} e^{-2\I \phi_{0}}]}{q_{d-3}^{(1)}}. 
\end{equation*}

The derived results can be generalized to $\left(P^{(\ell)},Q^{(\ell)}\right)$ where $1\leq \ell \leq \hat{d}$ ($\hat{d}$ is defined by \cref{eq:definition_hat_d}). We summarize these results as the following theorem.
\begin{theorem}\label{thm:cheby_version}
Given $\left(P(x), Q(x)\right)$ satisfying the assumption in \cref{thm:existandunique} and $\Phi=(\phi_0,\phi_1,\cdots,\phi_1,\phi_0)$ being a set of symmetric phase factors in $\mathcal{D}_d$ such that \cref{eq:UPQ} is satisfied, the following statements hold.
\begin{enumerate}
    \item For any $1\leq \ell\leq \hat{d}$, the leading coefficients of $P^{(\ell)}$ and $Q^{(\ell)}$ satisfy the recursion relation
\begin{equation}\label{eq:newPQ_coef_cheby}
\begin{split}
    &p^{(\ell)}_{d-2\ell}=q^{(\ell-1)}_{d+1-2\ell}-\frac{1}{4q^{(\ell-1)}_{d+1-2\ell}}\abs{\Im[p^{(\ell-1)}_{d-2\ell}e^{-2\I\phi_{\ell-1}}]}^2 +\I \Im[p^{(\ell-1)}_{d-2\ell}e^{-2\I\phi_{\ell-1}}],\\
    &q^{(\ell)}_{d-1-2\ell}=q^{(\ell-1)}_{d+1-2\ell}+\frac{1}{4q^{(\ell-1)}_{d+1-2\ell}}\abs{\Im[p^{(\ell-1)}_{d-2\ell}e^{-2\I\phi_{\ell-1}}]}^2.
\end{split}
\end{equation}
    \item For any $0\leq \ell\leq \hat{d}$, 
    \begin{equation}\label{eqn:phase_phil}
        e^{2\I \phi_{\ell}} =\frac{p^{(\ell)}_{d-2\ell}}{q^{(\ell)}_{d-1-2\ell}},
    \end{equation}
    and if $1\leq \ell \leq \hat{d}$,
    \begin{equation}\label{eqn:sin-2-phi-l-Im-over-q}
        \sin\left(2\phi_\ell\right)= \frac{\Im[p^{(\ell-1)}_{d-2\ell} e^{-2\I \phi_{\ell-1}}]}{q_{d-1-2\ell}^{(\ell)}}.
    \end{equation}
    \item If $d$ is even, $e^{\I \phi_{\wt{d}-1}}=p_0^{(\wt{d}-1)}$ and 
    \begin{equation}\label{eq:last_phi}
        \cos(\phi_{\wt{d}-1})=q_1^{(\wt{d}-2)},\ \sin\left( \phi_{\wt{d}-1}\right)= -\I p^{(\wt{d}-2)}_{0} e^{-2\I \phi_{\wt{d}-2}}.
    \end{equation}
    \item (The monotonicity of the leading coefficient of $Q^{(\ell)}$ in the Chebyshev basis) 
    
     If $q^{(0)}_{d-1}>0$, then 
     \begin{equation}\label{eq:chebyshev_Q_monotone}
     0< q^{(0)}_{d-1}\leq q^{(1)}_{d-3}\leq \cdots\leq q^{(\ell)}_{d-1-2\ell}\leq \cdots\leq q^{(\hat{d})}_{d-1-2\hat{d}}\leq 1.
     \end{equation}
\end{enumerate}

\end{theorem}
\begin{proof}
The statements 1, 2 and 4 follow the previous discussion. 
Hence, we only need to verify \cref{eq:last_phi}.

Note that when $\ell=\wt{d}-2$, we have $\deg\left(P^{(\wt{d}-2)}\right) = 2$ and 
\begin{equation*}
\begin{split}
    &\braket{0|e^{\I\phi_{\wt{d}-2} Z} W(x) e^{\I\phi_{\wt{d}-1}} W(x) e^{\I \phi_{\wt{d}-2} Z}|0} = e^{2\I\phi_{\wt{d}-2}} \left(2 \cos\left(\phi_{\wt{d}-1}\right) x^2 - e^{-\I\phi_{\wt{d}-1}}\right)\\
    &= P^{(\wt{d}-2)}(x) = p_2^{(\wt{d}-2)} T_2(x) + p_0^{(\wt{d}-2)} T_0(x) = 2 p_2^{(\wt{d}-2)} x^2 - p_2^{(\wt{d}-2)} + p_0^{(\wt{d}-2)} .
\end{split}
\end{equation*}
The first equality is obtained by direct computation. Comparing the constant  terms, we have
\begin{equation*}
\begin{split}
    e^{\I\phi_{\wt{d}-1}}&=\left(p_2^{(\wt{d}-2)}e^{-2\I \phi_{\wt{d}-2}}\right)^*-\left(p_0^{(\wt{d}-2)}e^{-2\I \phi_{\wt{d}-2}}\right)^*\\
    &=q_1^{(\wt{d}-2)}-\left(p_0^{(\wt{d}-2)}e^{-2\I \phi_{\wt{d}-2}}\right)^*.
\end{split}
\end{equation*}
Taking the modulus on both sides, we get
\begin{equation}\label{eqn:modulus-1-comparision-proof}
\begin{split}
    1&=\left(q_1^{(\wt{d}-2)}-\Re[p_0^{(\wt{d}-2)}e^{-2\I\phi_{\wt{d}-2}}]\right)^2+ \left(\Im[p^{(\wt{d}-2)}_{0} e^{-2\I \phi_{\wt{d}-2}}]\right)^2\\
    &=\left(q_1^{(\wt{d}-2)}\right)^2+\abs{p_0^{(\wt{d}-2)}}^2-2q_1^{(\wt{d}-2)}\Re[p_0^{(\wt{d}-2)}e^{-2\I\phi_{\wt{d}-2}}].
\end{split}
\end{equation}
Applying  Parseval's equality derived from $\abs{P^{(\wt{d}-2)}(x)}^2 + (1-x^2) \abs{Q^{(\wt{d}-2)}(x)}^2 = 1$, we get
\begin{equation*}
    \abs{p_0^{(\wt{d}-2)}}^2+\frac{1}{2}\abs{p_2^{(\wt{d}-2)}}^2+\frac{1}{2}\left(q_1^{(\wt{d}-2)}\right)^2=1.
\end{equation*}
Note that $\abs{p_2^{(\wt{d}-2)}} = \abs{p_2^{(\wt{d}-2)} e^{-2\I\phi_{\wt{d}-2}}} = \abs{q_1^{(\wt{d}-2)}}$, we have
\begin{equation}
    \abs{p_0^{(\wt{d}-2)}}^2 + \left(q_1^{(\wt{d}-2)}\right)^2 = 1.
\end{equation}
Plugging the last equality into \cref{eqn:modulus-1-comparision-proof}, we obtain $\Re[p_0^{(\wt{d}-2)}e^{-2\I\phi_{\wt{d}-2}}]=0$ since  $q_1^{(\wt{d}-2)}\ne 0$. Therefore, $p^{(\wt{d}-2)}_{0} e^{-2\I \phi_{\wt{d}-2}}$ is pure imaginary which yields  \cref{eq:last_phi}.
\end{proof}

\subsection{Bounding Chebyshev coefficients of the maximal solution}
In this subsection, we assume that $(P,Q)$ is the admissible pair corresponding to the maximal solution. We will derive bounds on the leading coefficients of $Q$ and $P_\Im(x)$ in the Chebyshev basis, \ie, $q_{d-1}$ and $s_{d}$ respectively.
\begin{lemma}\label{lma:coef_q_bound}
The leading coefficient of $Q$ in the Chebyshev basis satisfies the following estimate: $1-\norm{f}^2_{\infty}\leq q_{d-1}^2\leq 1-\norm{f}^2_{T}+\frac{1}{2}f_d^2$.
\end{lemma}
\begin{proof}
    Expanding the polynomials in the monomial basis, the normalization condition reduces to
    \begin{equation*}
        f(x)^2 + P_\Im(x)^2 + (1-x^2) Q(x)^2 = \left(\ff_d^2 + \fs_d^2 - \fq_{d-1}^2\right)x^{2d} + l.o. = 1,
    \end{equation*}
    where $l.o.$ denotes any polynomial of degree less than $2d$.
    It implies that $\ff_d^2+{\fs_d}^2=\fq_{d-1}^2$ which is further $f_d^2+{s_d}^2=q_{d-1}^2$. Hence, we have
    \begin{equation*}
    \begin{split}
        q_{d-1}^2&= \frac{1}{2}\left(f^2_d+s_d^2+q_{d-1}^2\right)\\
        &=1-\left(f_0^2+s_0^2+\frac{1}{2}\sum_{k=1}^{d-1} \left(f^2_k+s_k^2\right)+\frac{1}{2}\sum_{k=0}^{d-2} q_k^2\right) \\
        &\le 1-\left(f_0^2+\frac{1}{2}\sum_{k=1}^{d-1} f^2_k\right) = 1 -\norm{f}^2_{T}+\frac{1}{2}f_d^2.\\
    \end{split}
    \end{equation*}
    In the second line, we exploit \cref{eqn:parseval}, the Parseval's equality derived from the normalization condition.
    
To lower bound $q_{d-1}$, we first remark that \cref{eq:q_d-1} can be rewritten as
\begin{equation}
    q_{d-1}= \abs{\beta}\left(K+\frac{1}{K}+2\right),
\end{equation}
where $K:=-\prod_{r\in \mathcal{D}} r>0$, and $\mathcal{D}$ contains all roots of $\mf{F}(z)$ within the unit disc for the maximal solution.

Given a polynomial fully factored into $g(z)=\gamma \prod_{i=1}^n (z-a_i) \in \CC[z]$, its \emph{Mahler measure} \cite{Mahler1962} is defined as
\begin{equation*}
    M(g) := \abs{\gamma}\prod_{i=1}^n \max\{\abs{a_i},1\}.
\end{equation*}
This quantity is equal to the geometric mean of $\abs{g(z)}$ on the unit circle by using Jensen's formula in \cite{Jensen1900}
\begin{equation}
     M(g) =\exp\left(\frac{1}{2\pi}\int_{0}^{2\pi} \ln{\abs{g(e^{\I \theta})}} \rd \theta\right).
\end{equation}
Remarkably, note that $\abs{\beta}\frac{1}{K}$ is the Mahler measure of a polynomial $h(z)$ defined as
\begin{equation}
    h(z): =\beta \prod_{r\in \mathcal{D}} (z-r)(z-r^{-1})=z^{2d}\mf{F}(z) .
\end{equation}
Therefore, one has
\begin{equation}
\begin{split}
    \abs{\beta}\frac{1}{K} &= M(h) = \exp\left( \frac{1}{2\pi}\int_{0}^{2\pi} \ln{\abs{h(e^{\I \theta})}} \rd \theta\right)\\
    &=\exp\left( \frac{1}{2\pi}\int_{0}^{2\pi} \ln{\left(1-f(\cos{\theta})^2\right)} \rd \theta\right)\\
    &\geq \exp\left( \frac{1}{2\pi}\int_{0}^{2\pi} \ln{\left(1-\norm{f}_{\infty}^2\right)}\rd \theta\right)\\
    &=  1-\norm{f}_{\infty}^2.
\end{split}
\end{equation}
To conclude,
\begin{equation}
    q_{d-1}^2 = \abs{\beta}\left(K+\frac{1}{K}+2\right)\geq \abs{\beta}\frac{1}{K}\geq  1-\norm{f}_{\infty}^2,
\end{equation}
which completes the proof.
\end{proof}
\begin{remark}[Convergence of the maximal solution towards $\Phi^0$]
The derived bounds on $q_{d-1}$ consequentially yields the bounds on $s_d$. Using the equality $f_d^2+s_d^2=q_{d-1}^2$, we get the following estimate
\begin{equation}\label{eqn:bounds-on-sd}
    1-\norm{f}^2_{\infty}-f_d^2\leq s_d^2\leq 1-\norm{f}^2_{T}-\frac{1}{2}f_d^2.
\end{equation}
Combining the derived results, the remaining Chebyshev coefficients of $P_\Im$ and $Q$ can be upper bounded as follows
\begin{equation}
\begin{split}
    & s_0^2+\frac{1}{2}\sum_{k=1}^{d-1} s_k^2+\frac{1}{2}\sum_{k=0}^{d-2} q_k^2\\
    (\text{Parseval's equality})\ & = 1 - f_0^2 - \frac{1}{2} \sum_{k=1}^d f_k^2 - \frac{1}{2}s_d^2 - \frac{1}{2} q_{d-1}^2\\
    (f_d^2 + s_d^2 = q_{d-1}^2)\ & = 1 - q_{d-1}^2 - f_0^2  - \frac{1}{2} \sum_{k=1}^{d-1} f_k^2 \\
    (\text{\cref{lma:coef_q_bound}})\ & \le \norm{f}_\infty^2 - f_0^2 - \frac{1}{2} \sum_{k=1}^{d-1} f_k^2 \leq \norm{f}_{\infty}^2.
\end{split}
\end{equation}
As a consequence of the upper bound, we have $s_i, q_j \to 0\ \forall i \in [d], j\in [d-1]$ as $\norm{f}_\infty \to 0$. \cref{re:positive_lead_coef_Pim} implies that $s_d>0$. Together with $f_d^2 + s_d^2 = q_{d-1}^2$, it follows that $s_d = q_{d-1}$ in that limit since $f_d\to 0$, and the normalization condition yields $s_d = q_{d-1} = 1$ when the leading coefficient of $Q$ is fixed to be positive. 

To conclude, as $\norm{f}_\infty \to 0$, the maximal solution exhibits a limit $(P,Q)\to (\I T_d, U_{d-1})$ which agrees with the polynomials generated from the set of phase factors $\Phi^0 = (\frac{\pi}{4}, 0 ,\cdots, 0, \frac{\pi}{4})$. 
\end{remark}

\subsection{Proof of \cref{thm:estimate_maximal}}
In this subsection, we will prove \cref{thm:estimate_maximal} based on the previous results. We start from a technical lemma strengthening the results in \cref{thm:cheby_version} under additional assumptions.
\begin{lemma}\label{lma:supplement}
When $\norm{f}_{\infty}\leq \frac{1}{2}$, it holds that
\begin{enumerate}
    \item $\frac{\sqrt{3}}{2}\leq\sqrt{1-\norm{f}_{\infty}^2}\leq q^{(0)}_{d-1}\leq q^{(1)}_{d-3}\leq \cdots\leq q^{(\ell)}_{d-1-2\ell}\leq \cdots\leq q^{(\hat{d})}_{d-1-2\hat{d}}\leq 1$.
    \item  $\abs{\phi_0-\frac{\pi}{4}}< \frac{\pi}{4}$ and $\abs{\phi_\ell}< \frac{\pi}{4},  \quad\forall 1 \leq  \ell \leq \wt{d}-1$. 
\end{enumerate}
\end{lemma}
\begin{proof}
The first statement directly follows \cref{lma:coef_q_bound}.

To show the second statement, we first notice that according to \cref{thm:cheby_version},
\begin{equation*}
    \cos\left(2\phi_0-\frac{\pi}{2}\right)=\sin(2\phi_0)=\frac{\Im[p_d^{(0)}]}{q_{d-1}^{(0)}}.
\end{equation*}
According to the \cref{eq:pim_coef} in \cref{re:positive_lead_coef_Pim}, we have
\begin{equation}
    \Im[p_d^{(0)}]=\sqrt{\alpha}\left(1+\prod_{r\in \mathcal{D}} r\right)>0,
\end{equation}
where $\mathcal{D}$ contains the roots of $\mf{F}(z)$ in the unit disc. Therefore, $\abs{\phi_0-\frac{\pi}{4}}< \frac{\pi}{4}$ follows $\cos\left(2\phi_0-\frac{\pi}{2}\right)>0$.

When $1\leq \ell \leq \hat{d}$, according to \cref{eq:newPQ_coef_cheby,eqn:phase_phil}, we have
\begin{equation}\label{eq:1}
    \cos\left(2\phi_\ell\right) =\frac{1}{q^{(\ell)}_{d-1-2\ell}}\left(q^{(\ell-1)}_{d+1-2\ell}-\frac{1}{4q^{(\ell-1)}_{d+1-2\ell}}\abs{\Im[p^{(\ell-1)}_{d-2\ell}e^{-2\I\phi_{\ell-1}}]}^2\right).
\end{equation}
Using Parseval's equality, the normalization condition $\abs{P^{(\ell-1)}}^2+(1-x^2)\abs{Q^{(\ell-1)}}^2 =1$ yields
\begin{equation*}
    \begin{split}
        1&=\abs{p^{(\ell-1)}_0}^2+\frac{1}{2}\sum_{j=1}^{d-2\ell+2} \abs{p_j^{(\ell-1)}}^2+\frac{1}{2}\sum_{j=0}^{d-2\ell+1} \left(q_j^{(\ell-1)} \right)^2\\
        &\geq \frac{1}{2}\abs{p^{(\ell-1)}_{d-2\ell}}^2+\frac{1}{2}\abs{p^{(\ell-1)}_{d-2\ell+2}}^2+\frac{1}{2}\left(q_{d-2\ell+1}^{(\ell-1)}\right)^2.
    \end{split}
\end{equation*}
Furthermore, the cancellation of the leading order in the normalization condition implies $\abs{p^{(\ell-1)}_{d-2\ell+2}}^2 = \left(q_{d-2\ell+1}^{(\ell-1)}\right)^2$. Thus
\begin{equation*}
    \frac{1}{2}\abs{p^{(\ell-1)}_{d-2\ell}}^2+\left(q_{d-2\ell+1}^{(\ell-1)}\right)^2 \le 1.
\end{equation*}
Using the first statement frequently, we have
\begin{equation}\label{eq:2}
\begin{split}
    &\abs{\Im[p^{(\ell-1)}_{d-2\ell}e^{-2\I\phi_{\ell-1}}]}^2 \leq \abs{p^{(\ell-1)}_{d-2\ell}}^2\leq 2\left(1-\left(q_{d-2\ell+1}^{(\ell-1)}\right)^2\right)\\
    &\leq 2\left(1-\left(q_{d-1}^{(0)}\right)^2\right) \le 2 \norm{f}_\infty^2 \leq 4\norm{f}^2_{\infty} < 4 \left(1 - \norm{f}_\infty^2\right) \le 4\left(q^{(\ell-1)}_{d+1-2\ell}\right)^2.
\end{split}
\end{equation}
Plugging \cref{eq:2} into \cref{eq:1}, we have $\cos\left(2\phi_{\ell}\right)>0$. Hence, $\abs{\phi_{\ell}}< \frac{\pi}{4}$. 

If $d$ is even, following \cref{thm:cheby_version} statement 3, we have
\begin{equation*}
    \cos(\phi_{\wt{d}-1})=q_1^{(\wt{d}-2)} = q_1^{(\hat{d})} \geq \frac{\sqrt{3}}{2}.
\end{equation*}
Here, $\hat{d} = \wt{d} - 2$ when $d$ is even and the first statement is invoked. It implies that $\abs{\phi_{\wt{d}-1}}<\frac{\pi}{4}$.
\end{proof}
The second statement of the previous lemma ensures that the inequality $\frac{\pi}{2} \abs{\sin(x)}\geq \abs{x}$ for any $\abs{x}\leq \frac{\pi}{2}$ can be invoked. Now we are ready to prove \cref{thm:estimate_maximal}. 
\begin{proof}[Proof of \cref{thm:estimate_maximal}]\label{thm:estimate_maximal_proof}
We first estimate $\abs{\sin(2\phi_0-\frac{\pi}{2})}^2$ as follows
\begin{equation}
    \abs{\sin(2\phi_0-\frac{\pi}{2})}^2 =\abs{\frac{\Re[p_d]}{q_{d-1}}}^2=\frac{f_d^2}{q_{d-1}^2}\leq \frac{4}{3}f_d^2\leq \frac{8}{3}\norm{f}_{\infty}^2,
\end{equation}
where we use the inequality $f_d^2 \leq 2\norm{f}_{T}^2\leq 2\norm{f}_{\infty}^2$.

For any $1\leq \ell\leq\hat{d}$, by \cref{eqn:sin-2-phi-l-Im-over-q},  we have 
\begin{equation*}
\begin{split}
    \abs{ \sin(2 \phi_{\ell})}^2&=\abs{\frac{\Im[p^{(\ell-1)}_{d-2\ell} e^{-2\I \phi_{\ell-1}}]}{q_{d-1-2\ell}^{(\ell)}}}^2\\
    \text{(\cref{lma:supplement} 1.})&\leq \frac{4}{3}\abs{\Im[p^{(\ell-1)}_{d-2\ell} e^{-2\I \phi_{\ell-1}}]}^2\\
    \text{(\cref{thm:cheby_version} 1.}) &= \frac{16}{3}\left(q^{(\ell)}_{d-1-2\ell}-q^{(\ell-1)}_{d+1-2\ell}\right)q^{(\ell-1)}_{d+1-2\ell}\\
    \text{(\cref{thm:cheby_version} 4.})&\leq \frac{8}{3}\left(q^{(\ell)}_{d-1-2\ell}-q^{(\ell-1)}_{d+1-2\ell}\right)\left(q^{(\ell)}_{d-1-2\ell}+q^{(\ell-1)}_{d+1-2\ell}\right)\\
    &= \frac{8}{3}\left(\left({q^{(\ell)}_{d-1-2\ell}}\right)^2-\left({q^{(\ell-1)}_{d+1-2\ell}}\right)^2\right).\\
\end{split}
\end{equation*}
Here, the monotonicity in \cref{thm:cheby_version} implies that $q^{(\ell)}_{d-1-2\ell}-q^{(\ell-1)}_{d+1-2\ell} \ge 0$.

If $d$ is odd, one can derive an upper bound by telescoping and \cref{lma:supplement}
\begin{equation*}
\begin{split}
    \sum_{\ell=1}^{\hat{d}}\abs{ \sin(2\phi_{\ell})}^2 &\leq \frac{8}{3}\sum_{\ell=1}^{\hat{d}} \left(\left(q^{(\ell)}_{d-1-2\ell}\right)^2-\left(q^{(\ell-1)}_{d+1-2\ell}\right)^2\right)\\
    &\leq \frac{8}{3}\left(1 -\left({q_{d-1}^{(0)}}\right)^2\right)\leq \frac{8}{3}\norm{f}_{\infty}^2.
\end{split}
\end{equation*}
Therefore, we have
\begin{equation*}
    \begin{split}
        \norm{\wt{\Phi}^*-\wt{\Phi}^0}_2^2&=\left(\phi_0-\frac{\pi}{4}\right)^2+\sum_{\ell=1}^{\hat{d}} \phi_{\ell}^2\\
        &\leq \frac{1}{4}\frac{\pi^2}{4} \abs{\sin(2\phi_0-\frac{\pi}{2})}^2+ \frac{1}{4}\frac{\pi^2}{4}\sum_{\ell=1}^{\hat{d}}\abs{ \sin(2\phi_{\ell})}^2\\
        &\leq \frac{1}{4}\frac{\pi^2}{4} \left(\frac{8}{3}\norm{f}_{\infty}^2+\frac{8}{3}\norm{f}_{\infty}^2\right)=\frac{1}{3} \pi^2 \norm{f}_{\infty}^2.
    \end{split}
\end{equation*}

If $d$ is even, we have to treat $\phi_{\wt{d}-1}$ separately. Using \cref{thm:cheby_version} statement 3, we have 
\begin{equation*}
    \sin^2\left(\phi_{\wt{d}-1}\right) = 1 - \cos^2\left(\phi_{\wt{d}-1}\right) = 1- \left(q^{(\wt{d}-2)}_{1}\right)^2.
\end{equation*}
Recall that $\abs{\phi_{\wt{d}-1}}<\frac{\pi}{4}$, we apply a finer estimate $\abs{x}\leq \frac{\pi}{2\sqrt{2}} \abs{\sin(x)}$ for any $\abs{x} \le \frac{\pi}{4}$  and get
\begin{equation}
    \abs{\phi_{\wt{d}-1}}^2\leq \frac{\pi^2}{8}\sin^2\left(\phi_{\wt{d}-1}\right).
\end{equation}
Similarly,
\begin{equation*}
    \sum_{\ell=1}^{\hat{d}}\abs{ \sin(2\phi_{\ell})}^2 \leq \frac{8}{3}\sum_{\ell=1}^{\hat{d}} \left(\left(q^{(\ell)}_{d-1-2\ell}\right)^2-\left(q^{(\ell-1)}_{d+1-2\ell}\right)^2\right) =\frac{8}{3}\left(q^{(\wt{d}-2)}_1\right)^2-\frac{8}{3} \left(q_{d-1}^{(0)}\right)^2.
\end{equation*}
Furthermore, combining the derived results, one has
\begin{equation*}
    \begin{split}
        \norm{\wt{\Phi}^*-\wt{\Phi}^0}_2^2&=\left(\phi_0-\frac{\pi}{4}\right)^2+\sum_{\ell=1}^{\hat{d}} \phi_{\ell}^2+\phi_{\wt{d}-1}^2\\
        &\leq \frac{1}{4}\frac{\pi^2}{4} \abs{\sin(2\phi_0-\frac{\pi}{2})}^2+ \frac{1}{4}\frac{\pi^2}{4}\sum_{\ell=1}^{\hat{d}}\abs{ \sin(2\phi_{\ell})}^2+\frac{\pi^2}{8}\abs{\sin\left(\phi_{\wt{d}-1}\right)}^2\\
        &\leq \frac{\pi^2}{4}\left(\frac{2}{3}\norm{f}_{\infty}^2 + \frac{2}{3}\left(q^{(\wt{d}-2)}_1\right)^2-\frac{2}{3}\left(q_{d-1}^{(0)}\right)^2+\frac{1}{2}-\frac{1}{2} \left(q^{(\wt{d}-2)}_{1}\right)^2\right)\\
        (\text{\cref{thm:cheby_version} 4})&\leq \frac{\pi^2}{4}\left(\frac{2}{3}\norm{f}_{\infty}^2 + \frac{2}{3}-\frac{2}{3}\left(q_{d-1}^{(0)}\right)^2\right)\\
        (\text{\cref{lma:supplement} 1})&\leq \frac{\pi^2}{4}\left(\frac{2}{3}\norm{f}_{\infty}^2+\frac{2}{3}\norm{f}_{\infty}^2\right)= \frac{1}{3}\pi^2 \norm{f}_{\infty}^2.
    \end{split}
\end{equation*}

In conclusion, we have
\begin{equation*}
    \norm{\wt{\Phi}^*-\wt{\Phi}^0}_2\leq \frac{\pi}{\sqrt{3}}\norm{f}_{\infty}.
\end{equation*}
\end{proof}

\section{Local strong convexity}\label{sec:local_strong_convexity}
\subsection{Hessian matrix at the maximal solution is positive definite}
We first show that the Hessian matrix at the initial point $\wt{\Phi}^0$ is positive definite.  
\begin{theorem}\label{thm:strong-convex-init}
        At the initial point $\wt{\Phi}^0 := \left(\frac{\pi}{4}, 0, \cdots, 0\right) \in \RR^{\wt{d}}$, the Hessian matrix of the optimization problem \cref{eqn:optprob-intro} is
        \begin{equation}\label{eqn:hess-J}
                \mathrm{Hess}(\wt{\Phi}^0) = \frac{2}{\wt{d}}A(\wt{\Phi}^0)^\top A(\wt{\Phi}^0)=\left\{ \begin{array}{ll}
                        4 I & , \text{if }d\text{ is odd},\\
                        \mathrm{diag}\{ 4, \cdots, 4, 2\} & , \text{if }d\text{ is even}.
                \end{array} \right.
        \end{equation}
\end{theorem}
\begin{proof}
        According to \cite[Lemma 3]{DongMengWhaleyEtAl2021}, 
        $$\Re\left[ \langle 0 | U(x, \wt{\Phi}) | 0 \rangle \right] = - \Im\left[ \langle 0 | e^{-\I \frac{\pi}{4} Z} U(x, \wt{\Phi}) e^{-\I \frac{\pi}{4} Z} | 0 \rangle \right],$$
        which suggests us to extract $\frac{\pi}{4}$ from the phase factors at the edges and consider $e^{-\I \frac{\pi}{4} Z} U(x, \wt{\Phi}^0)\\ e^{-\I \frac{\pi}{4} Z} = W(x)^d$. The integer power of the matrix $W(x) := e^{\I \arccos(x) X}$ exactly generates the Chebyshev polynomial of the first kind as its upper-left matrix element. Furthermore, its upper-left is real. The claim follows the property of the exponential map of Pauli matrices $W(x)^d = e^{\I d \arccos(x) X} = \cos(d\arccos(x)) + \I \sin(d\arccos(x)) X$. Then, noting that $\braket{0|X|0} = 0$, it yields $\braket{0 | W(x)^d | 0} = \cos(d\arccos(x)) = T_d(x)$.
        
        To proceed, let us consider the monomial $M^{d_1,d_2,d_3}(x) := W(x)^{d_1}(\I Z) W(x)^{d_2} (\I Z) W(x)^{d_3}$ for any $d_1, d_2, d_3 \ge 0$. We will show that the upper-left element $\braket{0 | M^{d_1,d_2,d_3}(x) | 0}$ is real. Following the anticommutation relation $ZXZ = -X$, a useful equality can be derived $ZW(x)Z = e^{\I \arccos(x) ZXZ} = e^{-\I \arccos(x) X} = W(x)^{-1}$. Consequentially, by moving $Z$ gates to the same site for cancellation, it can be shown that $\braket{0 | M^{d_1,d_2,d_3}(x) | 0} = - T_{\abs{d_1+d_3-d_2}}(x)$ is real. 
        
        Note that taking second-order derivative on the unitary defined in \Cref{eqn:optprob-intro} at $\Phi^0$ results in an integer-coefficient superposition of monomials of the defined form. To illustrate, given an odd integer $d = 2\wt{d} - 1$ and $0 \le i < j < \wt{d}$, one can verify that
        \begin{equation*}
            \frac{\partial^2}{\partial \phi_i \partial \phi_j} U(x, \wt{\Phi}) \bigg|_{\wt{\Phi}_0} = M^{i,j-i,d-j}(x) + M^{i,d-j-i,j}(x) + M^{d-j,j-i,i}(x) + M^{j,d-j-i,i}(x).
        \end{equation*}
        Then, by taking the imaginary component, the second-order derivative $g_{ij}$ is vanishing at $\wt{\Phi}^0$. Thus, the second term in the right-hand side of \Cref{eq:hess-proof} vanishes since $g_{ij}(x, \wt{\Phi}^0) = 0$. 
        
        When $d = 2(\wt{d} - 1)$ is even, $g_i(x, \wt{\Phi}^0) = - 2 \Im\left[\braket{0| \I W(x)^i Z W(x)^{d-i}|0} \right] = -2 T_{d-2i}(x)$ for $i = 0, \cdots, \wt{d} - 2$ and $g_{\wt{d}-1}(x, \wt{\Phi}^0) = - 1$. According to the discrete orthogonality of Chebyshev nodes, we have
        $$\sum_{k=1}^{\wt{d}} g_i(x_k, \wt{\Phi}^0) g_j(x_k, \wt{\Phi}^0) = \wt{d} \left( 2 \delta_{ij} - \delta_{i, \wt{d}-1} \delta_{j, \wt{d}-1} \right).$$
        
        When $d = 2\wt{d}-1$ is odd, we have similarly $\sum_{k=1}^{\wt{d}} g_i(x_k, \wt{\Phi}^0) g_j(x_k, \wt{\Phi}^0) = 2 \wt{d} \delta_{ij}$.
        
        Following \cref{eq:hess-proof}, the Hessian matrix at $\wt{\Phi}^0$ takes the form in the theorem which completes the proof.
\end{proof}

To show that the Hessian matrix is bounded from below, we need the following lemma regarding the perturbation analysis to the singular values of a matrix.
\begin{lemma}[{\cite[Corollary 8.6.2]{GolubVanLoan96}}]\label{lma:svdperturb}
Given A, A+E $ \in \RR^{m\times n}$ with $m\geq n$, we have
\begin{equation*}
    \abs{\sigma_k (A+E)-\sigma_k(A)}\leq \sigma_0(E)=\norm{E}_2,\quad \forall k\in [n].
\end{equation*}
\end{lemma}

Then we bound the perturbation to the Jacobian matrix via that of the phase factors.
\begin{lemma}\label{lma:perturb-Jacobian}
    Given any symmetric phase factors $\Phi$ of length $d+1$, let $E:=A(\wt{\Phi}+\boldsymbol{\epsilon})-A(\wt{\Phi})$. Then $\norm{E}_2\leq 4\wt{d}^{\frac{3}{2}}\norm{\boldsymbol{\epsilon}}_2$.
\end{lemma}
\begin{proof}
Recall that in the proof of \cref{thm:strong-convex-init}, we have shown that taking the derivative \wrt $\phi_i$ yields terms by adding an $\I Z$ multiplication left to $e^{\I \phi_i Z}$. Note that $\I Z = e^{\I \frac{\pi}{2} Z}$, it shifts the $i$-th site by $\frac{\pi}{2}$ which essentially breaks the symmetry constraint. That means we have to use asymmetric set of phase factors to characterize the derivative. To capture this, we define the evaluation map of asymmetric phase factors as follows. Given an asymmetric set of phase factors $\Phi_\text{as} := (\phi_0, \phi_1, \cdots, \phi_{d})$, we define
\begin{equation*}
    U^{\text{(as)}}(x, \Phi_\text{as}) := e^{\I\phi_0 Z} e^{\I\arccos(x)X} e^{\I\phi_1 Z} e^{\I\arccos(x) X} \cdots e^{\I\phi_{d-1} Z} e^{\I\arccos(x) X} e^{\I\phi_d Z}.
\end{equation*}
We use a superscript $\text{(as)}$ to distinguish it with $U(x,\wt{\Phi})$ which takes the reduced phase factors $\wt{\Phi}$ of a symmetric set as argument.

When $d$ is odd, we define an asymmetric set of phase factors as
$$\Phi_\text{as}^{(i)} := \left(\phi_0, \cdots, \phi_{i-1}, \phi_i + \frac{\pi}{2}, \phi_{i+1}, \cdots, \phi_{\wt{d}-1}, \phi_{\wt{d}-1}, \cdots, \phi_{i+1}, \phi_i, \phi_{i-1}, \cdots, \phi_0\right)$$
which unrolls the full set of phase factors and add a $\frac{\pi}{2}$ shift on the $i$-th site. As a remark, the shifted site $i$ can be greater than $\wt{d}$ which adds a phase shift on the right half. 

Taking derivative yields that 
\begin{equation*}
\begin{split}
    g_i(x,\wt{\Phi}) &= \Re\left[ \braket{0| \frac{\partial}{\partial \phi_i} U(x, \wt{\Phi}) |0} \right] = \Re[\braket{0|U^{\text{(as)}}(x,\Phi_\text{as}^{(i)})|0}] + \Re[\braket{0|U^{\text{(as)}}(x,\Phi_\text{as}^{(d-i)})|0}]\\
    &= \Re[\braket{0|U^{\text{(as)}}(x,\Phi_\text{as}^{(i)})|0}] + \Re[\braket{0|U^{\text{(as)}}(x,\Phi_\text{as}^{(i)})^\top|0}] = 2\Re[\braket{0|U^{\text{(as)}}(x,\Phi_\text{as}^{(i)})|0}].
\end{split}
\end{equation*}
Due to the unitarity, we have
\begin{equation*}
    \abs{g_i(x, \wt{\Phi})} \le 2, \quad \forall i \in [\wt{d}].
\end{equation*}

To further take the second derivative, let us first introduce a double-shifted asymmetric set of phase factors. When $0 \le i < j < \wt{d}$,
\begin{equation*}
    \Phi_\text{as}^{(i,j)} := \left(\phi_0, \cdots, \phi_{i-1}, \phi_i + \frac{\pi}{2}, \phi_{i+1}, \cdots, \phi_{j-1}, \phi_j + \frac{\pi}{2}, \phi_{j+1}, \cdots, \phi_{\wt{d}-1}, \phi_{\wt{d}-1}, \cdots, \phi_0\right).
\end{equation*}
When $0 \le i = j < \wt{d}$,
\begin{equation*}
    \Phi_\text{as}^{(i,i)} := \left(\phi_0, \cdots, \phi_{i-1}, \phi_i + \pi, \phi_{i+1}, \cdots, \phi_{\wt{d}-1}, \phi_{\wt{d}-1}, \cdots, \phi_{i+1}, \phi_i, \phi_{i-1}, \cdots, \phi_0\right).
\end{equation*}
Similarly, when $i \text{ or } j \ge \wt{d}$, the phase shift is added on the right half accordingly.

With the defined notation, we are able to compute the second derivative as \begin{equation*}
    \begin{split}
        g_{ij}(x, \wt{\Phi}) &:= \frac{\partial}{\partial \phi_j} g_i(x, \wt{\Phi}) = 2\Re[\braket{0|\frac{\partial}{\partial \phi_j} U^{\text{(as)}}(x,\Phi_\text{as}^{(i)})|0}]\\
        &= 2 \Re[\braket{0|U^{\text{(as)}}(x,\Phi_\text{as}^{(i, j)})|0}] + 2 \Re[\braket{0|U^{\text{(as)}}(x,\Phi_\text{as}^{(i, d-j)})|0}].
    \end{split}
\end{equation*}

Due to the unitarity, we have
\begin{equation*}
    \abs{g_{ij}(x, \wt{\Phi})} \le 4, \quad \forall i, j \in [\wt{d}].
\end{equation*}
Therefore, for any given odd integer $d$ and symmetric $\Phi \in \RR^{d+1}$, we have
\begin{equation}\label{eq:estimate_gradient_odd}
    \norm{\nabla g_i(x,\wt{\Phi})}_2 = \sqrt{\sum_{j \in [\wt{d}]} \abs{g_{ij}(x, \wt{\Phi})}^2} \leq 4 \sqrt{\wt{d}}, \quad \forall i \in [\wt{d}].
\end{equation}

When $d$ is even, the full set of phase factors is $(\phi_0,...\phi_{\wt{d}-2},\phi_{\wt{d}-1},\phi_{\wt{d}-2},...,\phi_{0})$. The only difference in the unrolled structure comparing to that in the odd case is that there is a single unpaired phase factor $\phi_{\wt{d}-1}$ at the middle. Therefore, when the derivative only involves $\phi_0, \cdots, \phi_{\wt{d}-2}$, the analysis is exactly the same as that when $d$ is odd.

When $i=\wt{d}-1$, taking derivative \wrt $\phi_{\wt{d}-1}$ is equivalently shifting this site by $\frac{\pi}{2}$, namely,
\begin{equation*}
    g_{\wt{d}-1}(x,\wt{\Phi})= \Re[\braket{0|U(x,\wt{\Psi})|0}]=g(x,\wt{\Psi}),
\end{equation*}
where 
\begin{equation*}
    \Psi := \Phi + \frac{\pi}{2} \ve_{\wt{d}-1} = (\phi_0, \cdots, \phi_{\wt{d}-2}, \phi_{\wt{d}-1}+\frac{\pi}{2}, \phi_{\wt{d}-2}, \cdots, \phi_0).
\end{equation*}
Note that $\Psi$ is a set of symmetric phase factors (and therefore we use a different notation). Then, taking the second derivative on $g(x, \wt{\Phi})$ is equivalent to taking the first derivative on $g(x, \wt{\Psi})$ whose analysis is the same as that in the previous case. Therefore,
\begin{equation*}
    g_{i, \wt{d}-1}(x,\wt{\Phi}) = g_i(x, \wt{\Psi}) = 2 \Re\left[ \braket{0| U^{\text{(as)}}(x, \Psi_\text{as}^{(i)})|0} \right],\ \forall i \in [\wt{d}-1],
\end{equation*}
and
\begin{equation*}
    g_{\wt{d}-1,\wt{d}-1}(x,\wt{\Phi}) = g_{\wt{d}-1}(x, \wt{\Psi}) = - g(x,\wt{\Phi}).
\end{equation*}
These results lead to the inequality
\begin{equation*}
    \abs{g_{i, \wt{d}-1}(x,\wt{\Phi})} \le 2, \quad \forall i \in [\wt{d}].
\end{equation*}
Thus, the inequality in \cref{eq:estimate_gradient_odd} also holds when $d$ is even.

To conclude, for any given integer $d$ and symmetric $\Phi \in \RR^{d+1}$, we obtain a uniform upper bound of the gradient of $g_j(x,\wt{\Phi})$
\begin{equation}\label{eq:estimate_gradient}
    \norm{\nabla g_j(x,\wt{\Phi})}_2 \leq 4 \sqrt{\wt{d}},\  \forall j \in [\wt{d}], \ \forall x\in[-1,1].
\end{equation}
Following the mean value theorem, there exists $t\in (0,1)$ so that for any $j \in [\wt{d}]$ and $x \in [-1, 1]$
\begin{equation*}
    \abs{g_j(x,\wt{\Phi}+\boldsymbol{\epsilon})-g_j(x,\wt{\Phi})}\leq \norm{\nabla g_j(x,\wt{\Phi}+t\boldsymbol{\epsilon})}_2 \norm{\boldsymbol{\epsilon}}_2\leq 4\sqrt{\wt{d}}\norm{\boldsymbol{\epsilon}}_2.
\end{equation*}
Furthermore
\begin{displaymath}
    \norm{ A(\wt{\Phi}+\boldsymbol{\epsilon})-A(\wt{\Phi})}_F=\sqrt{\sum_{j=0}^{\wt{d}-1}\sum_{k=1}^{\wt{d}}\left(g_j(x_k,\wt{\Phi}+\boldsymbol{\epsilon})-g_j(x_k,\wt{\Phi})\right)^2}\leq 4\wt{d}^{\frac{3}{2}}\norm{\boldsymbol{\epsilon}}_2.
\end{displaymath}
Then, the theorem follows
\begin{displaymath}
    \norm{A(\wt{\Phi}+\boldsymbol{\epsilon})-A(\wt{\Phi})}_2\leq \norm{A(\wt{\Phi}+\boldsymbol{\epsilon})-A(\wt{\Phi})}_F\leq 4\wt{d}^{\frac{3}{2}}\norm{\boldsymbol{\epsilon}}_2.
\end{displaymath}
\end{proof}
Then the perturbation analysis to the singular value of the Jacobian matrix is obtained by directly applying \cref{lma:svdperturb}.
\begin{corollary}\label{cor:phi_around_initial}
If $\norm{\boldsymbol{\epsilon}}_2\leq \frac{1}{20\wt{d}}$, then 
$$\frac{4}{5}\sqrt{\wt{d}}\leq\sigma_{\min}(A(\wt{\Phi}^0+\boldsymbol{\epsilon}))\leq \sigma_{\max}(A(\wt{\Phi}^0+\boldsymbol{\epsilon}))\leq \left(\frac{1}{5}+\sqrt{2}\right)\sqrt{\wt{d}}.$$
\end{corollary}
\begin{proof}
Applying \cref{lma:svdperturb}, we have
\begin{equation}
    \abs{\sigma_k (A(\wt{\Phi}^0+\boldsymbol{\epsilon}))-\sigma_k(A(\wt{\Phi}^0))}\leq 4\wt{d}^{\frac{3}{2}}\norm{\boldsymbol{\epsilon}}_2.
\end{equation}
As a consequence of \cref{thm:strong-convex-init}, $\sigma_{\max}(A(\wt{\Phi}^0)) =\sqrt{2\wt{d}}$ and
\begin{equation*}
\sigma_\text{min}(A(\wt{\Phi}^0) = \left\{
    \begin{array}{ll}
        \sqrt{\wt{d}} &, \text{ when } d \text{ is even},\\
        \sqrt{2 \wt{d}} &, \text{ when } d \text{ is odd}. 
    \end{array}\right.
\end{equation*}
When $\norm{\boldsymbol{\epsilon}}_2\leq \frac{1}{20\wt{d}}$, we have
\begin{equation}
    \sigma_{\min}(A(\wt{\Phi}^0+\boldsymbol{\epsilon}))\geq \sigma_{\min}(A(\wt{\Phi}^0))- \norm{E}_2\geq \sqrt{\wt{d}}-4\wt{d}^{\frac{3}{2}}\norm{\boldsymbol{\epsilon}}_2 \geq\frac{4}{5}\sqrt{\wt{d}}.
\end{equation}
and 
\begin{equation}
    \sigma_{\max}(A(\wt{\Phi}^0+\boldsymbol{\epsilon}))\leq \sigma_{\max}(A(\wt{\Phi}^0))+ \norm{E}_2\leq \sqrt{2\wt{d}}+4\wt{d}^{\frac{3}{2}}\norm{\boldsymbol{\epsilon}}_2 \leq \left(\frac{1}{5}+\sqrt{2}\right)\sqrt{\wt{d}}.
\end{equation}
\end{proof}

Recall that the Hessian matrix at the maximal solution satisfies $\text{Hess}(\wt{\Phi}^*)=\frac{2}{\wt{d}} A(\wt{\Phi}^*)^\top A(\wt{\Phi}^*)$.  The positive definiteness of this matrix follows by combining the previous results with \cref{thm:estimate_maximal}.
\begin{theorem}
        Let $\Phi^*$ be the optimal phase factors corresponding to the maximal solution. If $\norm{f}_\infty \leq \frac{\sqrt{3}}{20\pi\wt{d}}$, then $\mathrm{Hess}(\wt{\Phi}^*)$ is positive definite and 
        \begin{equation}
            \lambda_{\min} \left(\mathrm{Hess}(\wt{\Phi}^*)\right)\geq \frac{32}{25}.
        \end{equation}
\end{theorem}
\begin{proof}
\cref{thm:estimate_maximal} implies that $\norm{\wt{\Phi}^* -\wt{\Phi}^0}_2 \leq \frac{1}{20\wt{d}}$. As a consequence of \cref{cor:phi_around_initial}, we have $\sigma_\text{min}(A(\wt{\Phi}^*)) \ge \frac{4}{5}\sqrt{\wt{d}}$. Therefore
\begin{equation*}
    \lambda_\text{min}(\text{Hess}(\wt{\Phi}^*)) = \lambda_\text{min}\left(\frac{2}{\wt{d}} A(\wt{\Phi}^*)^\top A(\wt{\Phi}^*)\right) = \frac{2}{\wt{d}} \sigma_\text{min}^2(A(\wt{\Phi}^*)) \ge \frac{32}{25},
\end{equation*}
which completes the proof.
\end{proof}

\subsection{Lower bounding the distance to the initial point}
In the last subsection, we show that the Hessian matrix is positive definite at the maximal solution. To show the positive definiteness of the Hessian matrix in a neighborhood of the optima, we need the following lower bound as an intermediate step.
\begin{theorem}\label{thm:estimate_Cheby_norm}
Given any symmetric phase factor $\wt{\Phi}$ satisfying $\norm{\wt{\Phi}-\wt{\Phi}^0}_2\leq \frac{1}{\sqrt{2}}$, the following inequality holds
\begin{equation}
   \frac{1}{\sqrt{3}} \norm{g(x,\wt{\Phi})}_{T}\leq \norm{\wt{\Phi}-\wt{\Phi}^0}_2.
\end{equation}
\end{theorem}
\begin{proof}
We denote the expansion of the polynomial $g(x, \wt{\Phi})$ in the Chebyshev basis as
\begin{equation*}
    g(x, \wt{\Phi}) = \sum_{k=0}^d g_k(\wt{\Phi}) T_k(x).
\end{equation*}
For simplicity, we drop the $\wt{\Phi}$-dependence in the Chebyshev coefficient $g_k$  (in particular, $g_k$ is not the derivative of $g(x,\wt{\Phi})$ with respect to $\phi_k$ as in the previous section). 

Then we apply \cref{eq:chebycoefPQ} to estimate $q_{d-1}$, the coefficient of $Q(x)$. Notice that $\cos(\phi_{\ell})$ is positive for any $1 \le \ell \le \wt{d}-1$, which follows $\abs{\phi_{\ell}}\leq \norm{\wt{\Phi}-\wt{\Phi}^0}_2\leq \frac{1}{\sqrt{2}}<\frac{\pi}{4}$. We use \cref{eq:chebycoefPQ} and the inequality $\cos(\phi_{\wt{d}-1}) \ge \cos(\phi_{\wt{d}-1})^2$ when $d$ is even, and obtain
\begin{equation}\label{eqn:q0-lower-bound}
    q_{d-1}^{(0)} \ge \prod_{\ell=1}^{\wt{d}-1} \cos^2(\phi_{\ell}) = \prod_{\ell=1}^{\wt{d}-1} \left(1-\sin^2(\phi_{\ell})\right) \geq 1-\sum_{\ell=1}^{\wt{d}-1} \phi_{\ell}^2\geq 1-\norm{\wt{\Phi}-\wt{\Phi}^0}_2^2\geq \frac{1}{2}.
\end{equation}
Hence $\deg(Q)=d-1$ and the leading coefficient of $Q$ is positive. Now we are able to apply \cref{thm:cheby_version,lma:coef_q_bound} and have
\begin{equation*}
    \sqrt{1-\norm{g(x,\wt{\Phi})}_{\infty}^2}\leq q^{(0)}_{d-1}\leq q^{(1)}_{d-3}\leq \cdots\leq q^{(\ell)}_{d-1-2\ell}\leq \cdots\leq q^{(\hat{d})}_{d-1-2\hat{d}}\leq 1.
\end{equation*}
According to the reduction procedure, the leading Chebyshev coefficients and the phase factor are connected as follows
\begin{equation}
    \abs{\sin(2\phi_0-\frac{\pi}{2})}^2 =\abs{\frac{\Re[p_d]}{q_{d-1}}}^2=\frac{g_d^2}{q_{d-1}^2}\geq g_d^2,
\end{equation}
where the last inequality follows $\abs{q_{d-1}} \le 1$.

Following \cref{eq:value_phi}, for any $ 1\leq \ell\leq\hat{d}$, we have 
\begin{equation*}
\begin{split}
    \abs{ \sin(2 \phi_{\ell})}^2&=\abs{\frac{\Im[p^{(\ell-1)}_{d-2\ell} e^{-2\I \phi_{\ell-1}}]}{q_{d-1-2\ell}^{(\ell)}}}^2 \geq \abs{\Im[p^{(\ell-1)}_{d-2\ell} e^{-2\I \phi_{\ell-1}}]}^2\\
    &= 4q^{(\ell-1)}_{d+1-2\ell}\left(q^{(\ell)}_{d-1-2\ell}-q^{(\ell-1)}_{d+1-2\ell}\right)\geq 4q_{d-1}^{(0)}\left(q^{(\ell)}_{d-1-2\ell}-q^{(\ell-1)}_{d+1-2\ell}\right),\\
\end{split}
\end{equation*}
where the last inequality follows the monotonicity in \cref{thm:cheby_version}.

Now we are going to show 
\begin{equation}
    \norm{\wt{\Phi}-\wt{\Phi}^0}_2^2\geq \frac{1}{4} g_d^2+q_{d-1}^{(0)}\left(1 -{q_{d-1}^{(0)}}\right),
\end{equation}
where $g_d$ denotes the leading Chebyshev coefficient of $g(x,\wt{\Phi})$.
 
When $d$ is odd, we have 
\begin{equation*}
    \sum_{\ell=1}^{\hat{d}}\abs{ \sin(2\phi_{\ell})}^2 \geq 4 q_{d-1}^{(0)}\sum_{\ell=1}^{\hat{d}} \left(q^{(\ell)}_{d-1-2\ell}-q^{(\ell-1)}_{d+1-2\ell}\right) = 4q_{d-1}^{(0)} \left(q_0^{(\hat{d})} - q_{d-1}^{(0)}\right) = 4 q_{d-1}^{(0)}\left(1 -{q_{d-1}^{(0)}}\right).
\end{equation*}
Here, we have used the fact that $q_0^{(\hat{d})}=1$ obtained from the proof of \cref{lma:phi_expression}. 
Hence 
\begin{equation*}
    \begin{split}
        \norm{\wt{\Phi}-\wt{\Phi}^0}_2^2&=\left(\phi_0-\frac{\pi}{4}\right)^2+\sum_{\ell=1}^{\hat{d}} \phi_{\ell}^2\\
        &\geq \frac{1}{4} \abs{\sin(2\phi_0-\frac{\pi}{2})}^2+ \frac{1}{4}\sum_{\ell=1}^{\hat{d}}\abs{ \sin(2\phi_{\ell})}^2\\
        &\geq \frac{1}{4} g_d^2+q_{d-1}^{(0)}\left(1 -{q_{d-1}^{(0)}}\right).
    \end{split}
\end{equation*}

When $d$ is even, the dangling unpaired phase factor $\phi_{\wt{d}-1}$ has to be analyzed individually
\begin{equation*}
    \abs{\sin\left(\phi_{\wt{d}-1}\right)}^2= 1- \left({q^{(\wt{d}-2)}_{1}}\right)^2.
\end{equation*}
Furthermore, we have
\begin{equation*}
    \sum_{\ell=1}^{\hat{d}}\abs{ \sin(2\phi_{\ell})}^2 \geq 4 q_{d-1}^{(0)} \sum_{\ell=1}^{\hat{d}} \left(q^{(\ell)}_{d-1-2\ell}-q^{(\ell-1)}_{d+1-2\ell}\right) =4 q_{d-1}^{(0)} \left(q_1^{(\wt{d}-2)}-q_{d-1}^{(0)}\right).
\end{equation*}
Unlike the odd case, the coefficient $q_1^{(\wt{d}-2)}$ is not necessary equal to $1$ when $d$ is even. Then
\begin{equation*}
    \begin{split}
        \norm{\wt{\Phi}-\wt{\Phi}^0}_2^2&=\left(\phi_0-\frac{\pi}{4}\right)^2+\sum_{\ell=1}^{\hat{d}} \phi_{\ell}^2+\phi_{\wt{d}-1}^2\\
        &\geq \frac{1}{4} \abs{\sin(2\phi_0-\frac{\pi}{2})}^2+ \frac{1}{4}\sum_{\ell=1}^{\hat{d}}\abs{ \sin(2\phi_{\ell})}^2+ \abs{\sin\left(\phi_{\wt{d}-1}\right)}^2\\
        &\geq \frac{1}{4}g_d^2 + q_{d-1}^{(0)} \left(q_1^{(\wt{d}-2)}-q_{d-1}^{(0)}\right) +1- \left(q^{(\wt{d}-2)}_{1}\right)^2\\
        &= \frac{1}{4} g_d^2+q_{d-1}^{(0)}\left(1 -{q_{d-1}^{(0)}}\right)+\left(1 -{q_{d-1}^{(0)}}\right)\left(1 -{q_{1}^{(\wt{d}-2)}}\right)+q_{1}^{(\wt{d}-2)}\left(1 -{q_{1}^{(\wt{d}-2)}}\right)\\
        &\geq \frac{1}{4} g_d^2+q_{d-1}^{(0)}\left(1 -{q_{d-1}^{(0)}}\right).
    \end{split}
\end{equation*}
According to \cref{eqn:q0-lower-bound,lma:coef_q_bound}, we have $\frac{1}{2}\leq q_{d-1}^{(0)}\leq \sqrt{1-\norm{g(x, \wt{\Phi})}_T^2+\frac{1}{2}g_d^2}$. It implies that 
\begin{equation*}
\begin{split}
    \norm{\wt{\Phi}-\wt{\Phi}^0}_2^2&\geq \frac{1}{4} g_d^2+\sqrt{1-\norm{g(x,\wt{\Phi})}_T^2+\frac{1}{2}g_d^2}-1+\norm{g(x,\wt{\Phi})}_T^2-\frac{1}{2}g_d^2\\
    &\geq \norm{g(x,\wt{\Phi})}_T^2-\frac{1}{4}g_d^2-\frac{\norm{g(x,\wt{\Phi})}_T^2-\frac{1}{2}g_d^2}{\sqrt{1-\norm{g(x,\wt{\Phi})}_T^2+\frac{1}{2}g_d^2}+1}\\
    &\geq \norm{g(x,\wt{\Phi})}_T^2-\frac{1}{4}g_d^2-\frac{2\norm{g(x,\wt{\Phi})}_T^2-g_d^2}{3}\\
    &\geq \frac{1}{3}\norm{g(x,\wt{\Phi})}_T^2.
\end{split}
\end{equation*}
\end{proof}

\subsection{Local strong convexity of the optimization problem}
Now we are going to present the proof of \cref{thm:Hess_PD}, where we use \cref{cor:phi_around_initial} and the following lemma regarding the perturbative analysis to the eigenvalues of a matrix.

\begin{lemma}[{\cite[Theorem 8.1.4]{GolubVanLoan96}} \textbf{Wielandt--Hoffman}]\label{lma:perturb_eigen}
If A, A+E $ \in \RR^{n\times n}$ are symmetric matrices, then 
\begin{equation*}
    \sum_{k\in[n]}\left(\lambda_k (A+E)-\lambda_k(A)\right)^2\leq \norm{E}_F^2.
\end{equation*}
\end{lemma}
Although Weyl's inequality may give a sharper estimate about the perturbation of eigenvalues, we adopt \cref{lma:perturb_eigen} in our analysis. That is because the former exploits the operator norm of the perturbation of matrix, which is hard to estimate in our problem.

\begin{proof}[Proof of \cref{thm:Hess_PD}]
Recall that 
\begin{equation}
    F(\wt{\Phi})=\frac{2}{\wt{d}}\left(A^\top(\wt{\Phi}) A(\wt{\Phi})\right)\REV{+}R(\wt{\Phi}),
\end{equation}
where $R_{ij}:=\frac{2}{\wt{d}}\sum_{k=1}^{\wt{d}} \left(g(x_k,\wt{\Phi})-f(x_k)\right)g_{ij}(x_k,\wt{\Phi})$. The residual term is estimated as follows
\begin{equation}
    \begin{split}
        \norm{R}_F^2 &=\sum_{i=0}^{\wt{d}-1}\sum_{j=0}^{\wt{d}-1}\left(\frac{2}{\wt{d}}\sum_{k=1}^{\wt{d}} \left(g(x_k,\wt{\Phi})-f(x_k)\right)g_{ij}(x_k,\wt{\Phi})\right)^2\\
        &\leq \frac{4}{\wt{d}^2} \sum_{i=0}^{\wt{d}-1}\sum_{j=0}^{\wt{d}-1} \left(\sum_{k=1}^{\wt{d}} \left(g(x_k,\wt{\Phi})-f(x_k)\right)^2\right)\left(\sum_{k=1}^{\wt{d}} g_{ij}(x_k,\wt{\Phi})^2\right)\\
        &\leq 8\left(\sum_{k=1}^{\wt{d}} g(x_k,\wt{\Phi})^2+\sum_{k=1}^{\wt{d}}f(x_k)^2\right)\max_{i,j}\left(\sum_{k=1}^{\wt{d}} g_{ij}(x_k,\wt{\Phi})^2\right)\\
        &= 8 \wt{d}^2\left(\norm{g(x,\wt{\Phi})}_T^2+\norm{f(x)}_T^2\right) \max_{i,j} \norm{g_{ij}(x,\wt{\Phi})}_T^2.
    \end{split}
\end{equation}
Here, the last line uses the discrete orthogonality of Chebyshev polynomial.

Then, we have the following inequality
\begin{equation}
    \norm{R}_F^2 \le 128\wt{d}^2 \left(3\norm{\wt{\Phi}-\wt{\Phi}^0}_2^2+\norm{f}_{\infty}^2\right)\leq \frac{24}{25}\left(1+\frac{1}{\pi^2}\right),
\end{equation}
where the first inequality uses \cref{thm:estimate_Cheby_norm}, the fact $\norm{f}_T \le \norm{f}_\infty$, and its consequence $\norm{g_{ij}(x,\wt{\Phi})}_T \le  \norm{g_{ij}(x,\wt{\Phi})}_\infty \le 4$ derived in the proof of \cref{lma:perturb-Jacobian}.

Then, applying \cref{lma:perturb_eigen}, one has
\begin{equation}
    \begin{split}
        &\lambda_{\min} (\mathrm{Hess}(\wt{\Phi}))\geq \lambda_{\min}(\frac{2}{\wt{d}} A(\wt{\Phi})^\top A(\wt{\Phi}))-\norm{R}_F\geq\frac{32}{25}- \sqrt{\frac{24}{25}\left(1+\frac{1}{\pi^2}\right)}\geq \frac{1}{4},\\
        &\lambda_{\max} (\mathrm{Hess}(\wt{\Phi}))\leq \lambda_{\max}(\frac{2}{\wt{d}} A(\wt{\Phi})^\top A(\wt{\Phi}))+\norm{R}_F\leq 2\left(\frac{1}{5}+\sqrt{2}\right)^2+\sqrt{\frac{24}{25}\left(1+\frac{1}{\pi^2}\right)}\leq \frac{25}{4}.
    \end{split}
\end{equation}
\end{proof}

\begin{remark}
According to \cref{thm:Hess_PD}, we know that if $\norm{f}_\infty \leq \frac{\sqrt{3}}{20\pi\wt{d}}$, the optimization problem \cref{eqn:optprob-intro} is strongly convex in the following region 
\begin{equation}
    \mathcal{S}:=\left\{\wt{\Phi}\in \RR^{\wt{d}}: \norm{\wt{\Phi}-\wt{\Phi}^0}_2\leq\frac{1}{20\wt{d}}\right\}.
\end{equation}
Furthermore, the maximal solution is included in that region $\wt{\Phi}^*\in \mathcal{S}$. For convenience, the lower bound of the Hessian matrix is referred to as $\sigma= \frac{1}{4}$, while the upper bound is referred to as $L = \frac{25}{4} $.
\end{remark}

\subsection{Projected gradient descend algorithm for symmetric phase-factor evaluation}

\begin{algorithm}[htbp]
\caption{Projected gradient descend algorithm for optimizing symmetric phase factors  (with performance guarantee)}
\label{alg:proj-gradient-descent}
\begin{algorithmic}
\STATE{\textbf{Input:} A target real polynomial $f$ of degree $d$ so that it has definite parity and $\norm{f}_{\infty}\leq \frac{\sqrt{3}}{20\pi\wt{d}}$, an initial guess $\wt{\Phi}^0=\left(\frac{\pi}{4},0,\cdots,0\right)\in \mathcal{S}$ and the error tolerance $\epsilon$. }
\STATE{\textbf{Output:} A set of reduced phase factors $\wt{\Phi}$ so that $F(\wt{\Phi}) \le \epsilon$ implying $\norm{g(x, \wt{\Phi}) - f(x)}_\infty \le d \sqrt{\epsilon}$ \cite[Theorem 4]{DongMengWhaleyEtAl2021}.}
\STATE{}
\STATE{Evaluate the objective function $F(\wt{\Phi})$ at $\wt{\Phi}^0$ according to \cref{eqn:optprob-intro}.}
\WHILE{$F(\wt{\Phi})> \epsilon$}
\STATE{Evaluate the gradient $\nabla F(\wt{\Phi})$ of the objective.}
\STATE{Update $\wt{\Phi}_i$ by gradient descent $\wt{\Phi} \leftarrow \wt{\Phi} -\frac{1}{L} \nabla F(\wt{\Phi})$.}
 \IF{$\norm{\wt{\Phi}-\wt{\Phi}^0}_2> \frac{1}{20\wt{d}}$}
  \STATE{Project the point onto the ball $\wt{\Phi} \leftarrow \wt{\Phi}^0-\frac{\wt{\Phi}-\wt{\Phi}^0}{20\wt{d}\norm{\wt{\Phi}-\wt{\Phi}^0}_2}$.}
  \ENDIF
\ENDWHILE
\RETURN $\wt{\Phi}$
\end{algorithmic}
\end{algorithm}

We list a practical algorithm for symmetric phase-factor evaluation in \cref{alg:proj-gradient-descent} which is based on the projected gradient descent method \cite{Bubeck2015}. In \cref{cor:conv_pg}, we give the convergence analysis of this algorithm \REV{and here is the proof.}
\REV{\begin{proof}[Proof of \cref{cor:conv_pg}]
According to \cref{thm:Hess_PD}, the cost function $F(\wt{\Phi})$ is $\sigma$-strongly convex and $L$-smooth on ball $\mathcal{S}:=\{\wt{\Phi}:\norm{\wt{\Phi}-\wt{\Phi}^0}\leq \frac{1}{20\wt{d}}\}$, where $\sigma=\frac{25}{4}$ and $L=\frac{1}{4}$. Using \cref{thm:estimate_maximal}, we know that $\wt{\Phi}^*\in \mathcal{S}$, since $\norm{\wt{\Phi}^*-\wt{\Phi}^0}\leq \frac{\pi}{\sqrt{3}}\norm{f}_{\infty}\leq \frac{1}{20\wt{d}}$. Applying \cite[Theorem 3.10]{Bubeck2015}, the projected gradient method with step size $\frac{1}{L}$ converges exponentially, that is,
\begin{equation*}
    \norm{\wt{\Phi}^\ell-\wt{\Phi}^*}_2\leq e^{-\frac{\sigma}{L}\ell}\norm{\wt{\Phi}^0-\wt{\Phi}^*}_2,
\end{equation*}
where $\wt{\Phi}^\ell$ is the set of reduced phase factors in the $\ell$-th iteration step. In particular, $\wt{\Phi}^\ell \in \mc{S}$ for all $\ell$ follows the construction of the algorithm. 
\end{proof}}

\REV{Let $\wt{\Phi}^\ell$ be the set of reduced phase factors in the $\ell$-th iteration step.} Note that $F(\wt{\Phi}^*) = 0,\ \nabla F(\wt{\Phi}^*) = 0$ and $\wt{\Phi}^* \in \mc{S}$. Following  Taylor's theorem, there exists $s \in (0, 1)$ so that
\begin{equation*}
\begin{split}
    F(\wt{\Phi}^\ell) &= F(\wt{\Phi}^*) + \nabla F(\wt{\Phi}^*)^\top \left(\wt{\Phi}^\ell - \wt{\Phi}^*\right) + \frac{1}{2} \left(\wt{\Phi}^\ell-\wt{\Phi}^*\right)^\top \text{Hess}\left(s \wt{\Phi}^\ell + (1-s) \wt{\Phi}^*\right) \left(\wt{\Phi}^\ell-\wt{\Phi}^*\right)\\
    &= \frac{1}{2} \left(\wt{\Phi}^\ell-\wt{\Phi}^*\right)^\top \text{Hess}\left(s \wt{\Phi}^\ell + (1-s) \wt{\Phi}^*\right) \left(\wt{\Phi}^\ell-\wt{\Phi}^*\right).
\end{split}
\end{equation*}
Because $\wt{\Phi}^\ell, \wt{\Phi}^* \in \mc{S}$ and $\mc{S}$ is convex, the intermediate point also lies in the convex set $s \wt{\Phi}^\ell + (1-s) \wt{\Phi}^* \in \mc{S}$. Then, the Hessian is upper bounded at the intermediate point by \cref{thm:Hess_PD}. Thus
\begin{equation*}
    F(\wt{\Phi}^\ell) \le \frac{1}{2} \lambda_\text{max}\left(\text{Hess}\left(s \wt{\Phi}^\ell + (1-s) \wt{\Phi}^*\right)\right) \norm{\wt{\Phi}^\ell - \wt{\Phi}^*}_2^2 \le \frac{L}{2} \norm{\wt{\Phi}^\ell - \wt{\Phi}^*}_2^2.
\end{equation*}
Furthermore, using \cref{cor:conv_pg}, we have
\begin{equation*}
    F(\wt{\Phi}^\ell) \le \frac{L}{2} \norm{\wt{\Phi}^\ell - \wt{\Phi}^*}_2^2 \le \frac{L}{2} e^{-\frac{\sigma}{L}\ell} \norm{\wt{\Phi}^0 - \wt{\Phi}^*}_2^2.
\end{equation*}
To obtain $F(\wt{\Phi}^\ell) \le \epsilon$, it suffices to take $\ell = \Or\left(\log\left(\frac{1}{\epsilon}\right)\right)$ iterations. 

Remarkably, \cref{alg:proj-gradient-descent} does not require a large number of bits. To see the claim, we consider the model of finite precision arithmetic (see standard axioms in \cite{Higham}). The parameter $u$ is referred to the unit roundoff. In each iteration step, the rounding error when computing the function $F$ and its gradient is $\Or(d^2u)$. Furthermore, the projection is needed to be executed when the proposed new $\Phi$ lies outside $\mathcal{S}$. Under this circumstance, the possible error amplification by taking square roots and division does not contribute significantly and the total rounding error accumulated in each iteration step remains $\mathcal{O}(d^2 u)$. To conclude, the numerical error of the algorithm is upper bounded by invoking triangle inequality. Then, there exists some constant $C$ so that
\begin{equation*}
    C d^2 u \log\frac{1}{\epsilon} \le \epsilon \ \Rightarrow \ \frac{1}{u} := 2^{n_\text{bits}} \ge \frac{1}{C} \frac{d^2 \log(1/\epsilon)}{\epsilon} \ \Rightarrow \ n_\text{bits} = \Or\left(\log\left(\frac{d \log(1/\epsilon)}{\epsilon}\right)\right).
\end{equation*}
Therefore, the number of bits required by \cref{alg:proj-gradient-descent} is $\mathcal{O}\left(\log\left(\frac{d\log(1/\epsilon)}{\epsilon}\right)\right)$.
Therefore the algorithm is numerically stable, and it is sufficient to use the standard double precision arithmetic operations.

\section{Additional numerical results}\label{sec:numer}

In this section, we present some additional numerical results demonstrating some of the features of the energy landscape of the optimization problem for symmetric phase factors in \cref{eqn:optprob-intro}.

\subsection{Hessian matrix without symmetry constraint is singular at global optima}
 When the symmetry constraint is not imposed, the Hessian matrix can be derived from \cref{eq:hess-proof}. Note that the Jacobian matrix $A \in \RR^{\wt{d}\times d}$ does not have full column rank because $d > \wt{d}$. Then, following that $\ker(A^\top A) = \ker(A)$, we have $\text{rank}(A^\top A) = \text{rank}(A) \le \wt{d} < d$. Note that at the global optima, the Hessian matrix is precisely $\text{Hess}(\Phi^*) = (A^*)^\top A^*$ due to the vanishing residual term $R = 0$. Here, $A^*$ refers to the Jacobian matrix at the optima. Thus, the Hessian matrix at optima is singular. Because the residual term $R$ is small near the optima by continuity, the Hessian matrix is almost singular in a neighborhood of the optima. This behavior is demonstrated by the numerical result displayed in \cref{fig:hess-singular-near-optmum} (a). 
On the other hand, when the symmetry constraint is imposed, \cref{fig:hess-singular-near-optmum} (b) shows that the Hessian matrix is positive definite in a neighborhood of the maximal solution and it agrees our theoretical results about the local strong convexity of the Hessian matrix. 

\begin{figure}[htbp]
    \centering
    \includegraphics[width=\textwidth]{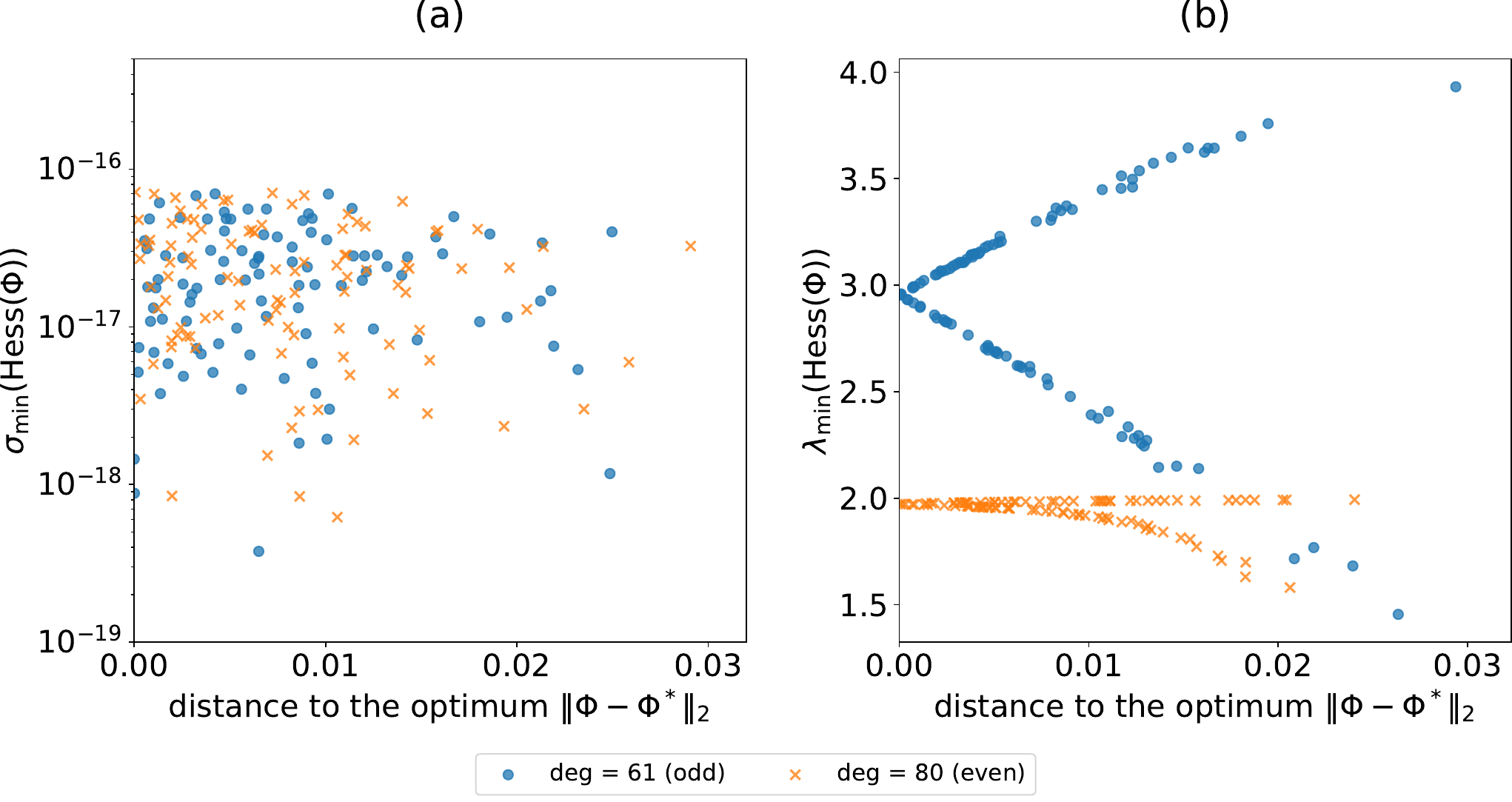}
    \caption{The smallest singular value ($\sigma_\text{min}$) or eigenvalue ($\lambda_\text{min}$) of the Hessian matrix evaluated at 100 randomly sampled $\Phi$ near the optimizer $\Phi^*$ of a given polynomial. Blue dots correspond to an odd polynomial of degree $61$ and orange crosses correspond to an even polynomial of degree $80$. (a) The symmetry constraint is not imposed. The plotted data concentrates around the machine precision $\sim 10^{-17}$, which implies a large basin near the optimizer on which the Hessian matrix is singular if the symmetry constraint is not imposed. (b) The symmetry constraint is imposed. The numerical results show that the Hessian matrix is well conditioned near the optimizer, and it is positive definite which agrees our result of local strong convexity.}
    \label{fig:hess-singular-near-optmum}
\end{figure}

\subsection{Existence of local minima}
To visualize the full energy landscape of the optimization problem, we first consider the simplest scenario so that only two phase factors suffice to solve the optimization problem exactly.
We choose two target polynomials $f(x) = x^2-\frac{1}{2}$ and $f(x) = \frac{1}{\sqrt{3}}x^3 - \frac{2}{\sqrt{3}}x$. 
The energy landscape on the irreducible domain is displayed in \cref{fig:landscape_global}. The global minima derived from the proposed method are annotated in the domain. In this special case, there are \REVN{no local minima}. 
However, this is a rare exception rather than the norm. Existence of local minima can be observed by slightly increasing the degree of the target polynomial $f(x) =0.2103 T_4(x)+0.1231T_2(x)+0.1666$. The local minimizer is numerically searched by randomly initiating the stochastic gradient descent algorithm to solve \cref{eqn:optprob-intro}. The value of the objective function at $\Phi^\text{loc}$ is $0.0218$. The Hessian matrix has eigenvalues $0.1359, 4.5815, 7.7510$, which indicates that $\Phi^\text{loc}$ is indeed a local minimum. To visualize the landscape, we plot a two-dimensional section of the optimization landscape spanned by the affine plane spanned by the eigenvectors of the least two eigenvalues of the Hessian matrix (\cref{fig:landscape_global}).

\REVN{Next we present an example to demonstrate the complexity of the full landscape of the optimization problem. We choose $f(x)= \frac{1}{\alpha}\left(T_4(x)+2T_2(x)+T_0(x)\right)$ as the target polynomial. The scale factor $\alpha=440$ is set to be large enough and satisfies the requirements of \cref{thm:Hess_PD}. We first found a global minimum around $\wt{\Phi}^0 = (\frac{\pi}{4}, 0,0)$, by taking $\wt{\Phi}^0$ as the initial guess for \cref{alg:proj-gradient-descent}.
We also numerically searched two local minimizers by randomly initiating the stochastic gradient descent algorithm to solve \cref{eqn:optprob-intro}.  We use $\wt{\Phi}^*$, $\wt{\Phi}^1$ and $\wt{\Phi}^2$ to denote the global minimum and two local minima respectively. Their values are as follows:
\begin{equation*}
    \begin{split}
        \wt{\Phi}^*&= (0.7843, -0.0023, -0.0023),\\
        \wt{\Phi}^1&= ( 1.2317, 1.5685, 2.2467),\\
        \wt{\Phi}^2&= ( 0.4370, 1.5731, 0.6991).
    \end{split}
\end{equation*}The values of the objective function at $\wt{\Phi}^1$ and $\wt{\Phi^2}$ are both around 2.5769e-06. The Hessian matrix at both points are positive semidefinite, which indicates that they are indeed local minima. To visualize the landscape, we plot a two-dimensional section of the optimization landscape spanned by the affine plane spanned by $\wt{\Phi}^1-\wt{\Phi}^*$ and $\wt{\Phi}^2-\wt{\Phi}^*$ (\cref{fig:localmin_multi}). It agrees with our conclusion that the energy landscape of the optimization problem remains complex and has many local minima as well as global minima. 
This confirms that even when $\norm{f}_{\infty}$ is sufficiently small, the global energy landscape is still very complex, and the local energy landscape described by \cref{thm:Hess_PD} is nontrivial.
}

\begin{figure}[htbp]
    \centering
    \includegraphics[width=\textwidth]{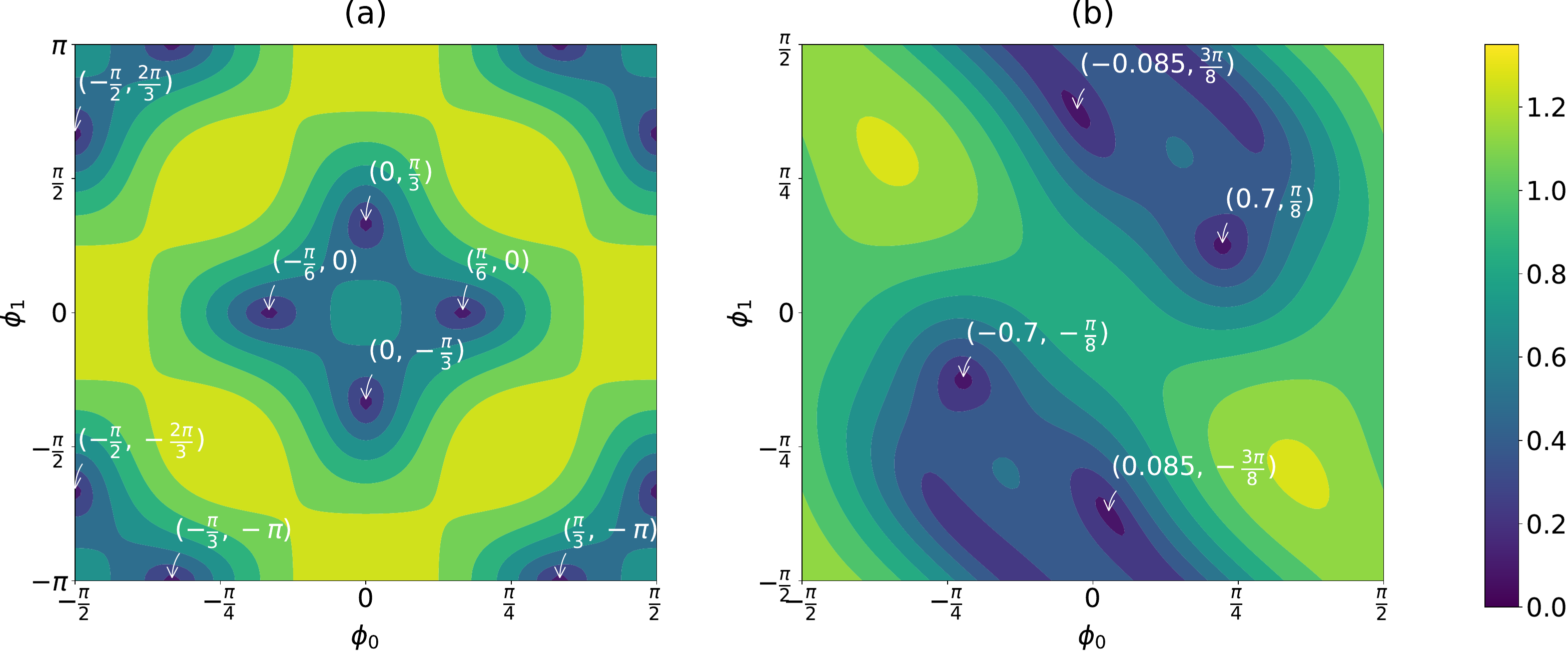}
    \caption{The optimization landscape of the modified cost function $F(\wt{\Phi})^{1/3}$ on the irreducible domain $D_d$. Here, the cube root is taken to signify the structure of the landscape near optima. The annotated values are all optima exactly computed from the proposed method. (a) The target function is set to $f(x)=x^2-\frac{1}{2}$ which yields 8 inequivalent optima. (b) The target function is set to $f(x)=\frac{1}{\sqrt{3}} x^3 - \frac{2}{\sqrt{3}} x$ which yields 4 inequivalent optima.}
    \label{fig:landscape_global}
\end{figure}

\begin{figure}[htbp]
    \centering
    \includegraphics[width=.5\textwidth]{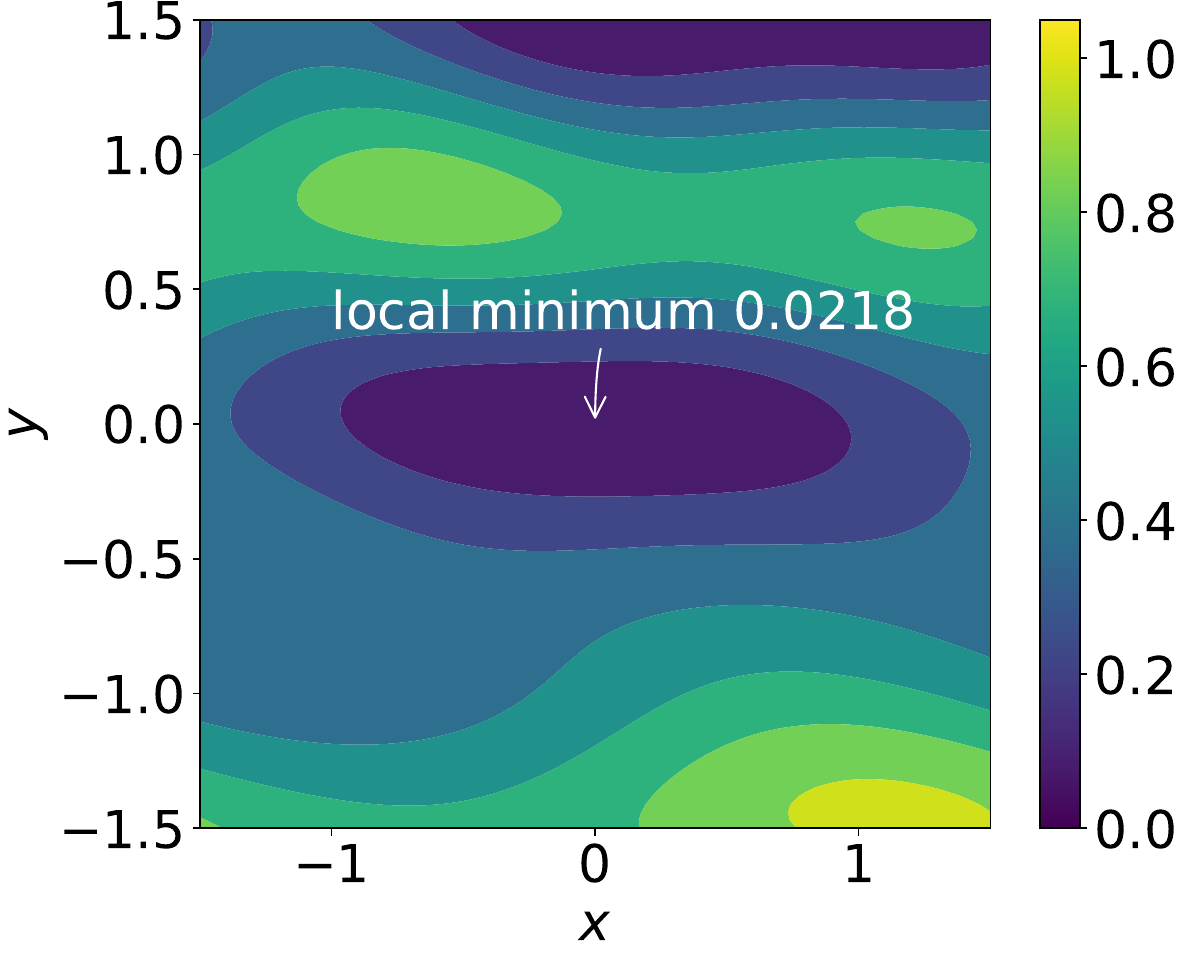}
    \caption{A two-dimensional section of the optimization landscape $F(\wt{\Phi}^{\mathrm{loc}}+xv_1+yv_2)$, where $\Phi^{\mathrm{loc}}$ is a local minimum of a given degree-$5$ target polynomial, and $v_1, v_2$ are unit eigenvectors corresponding to the least two eigenvalues of the Hessian matrix evaluated at $\Phi^{\mathrm{loc}}$. The objective value of the local minimum is annotated on the figure.}
    \label{fig:localmin}
\end{figure}

\begin{figure}[htbp]
    \centering
    \includegraphics[width=.5\textwidth]{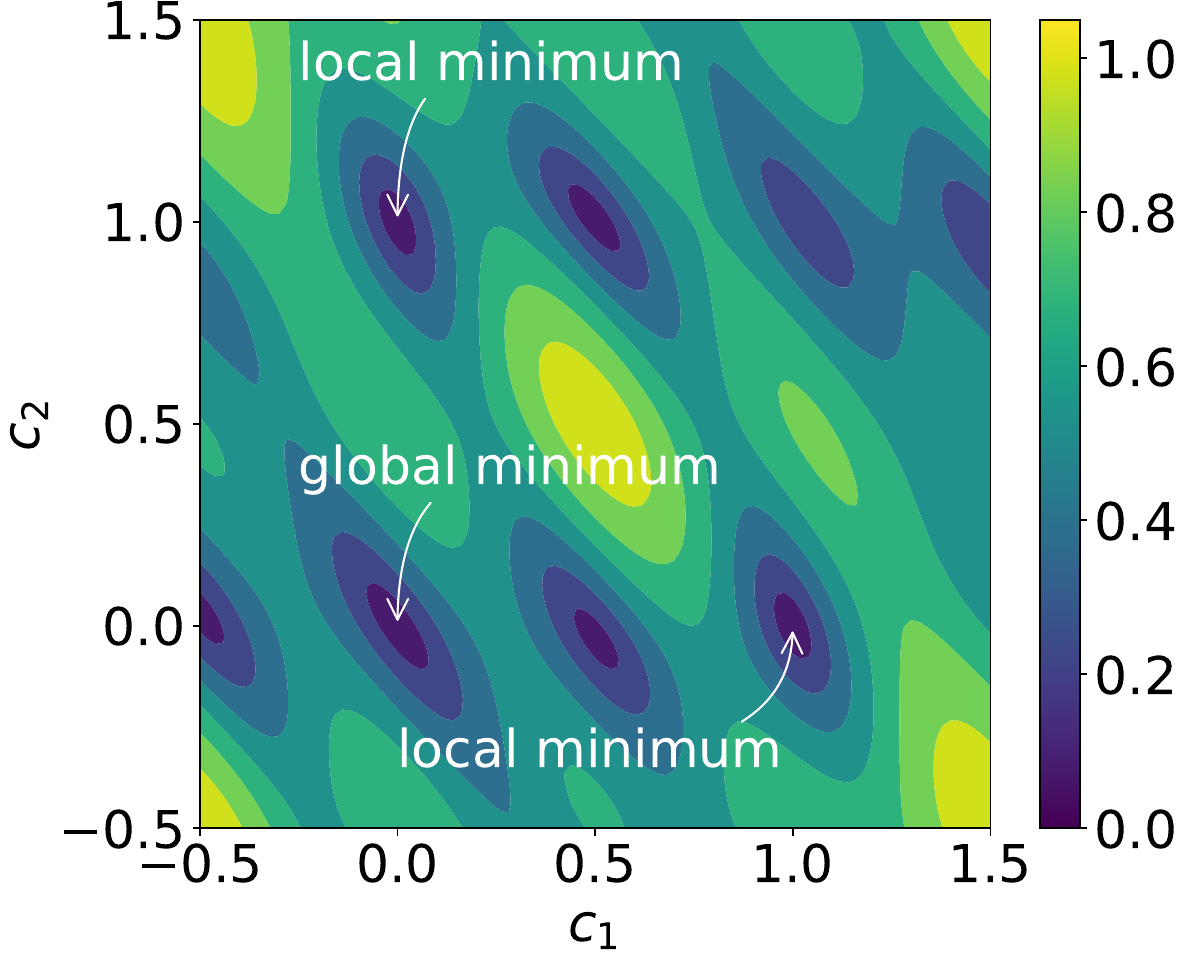}
    \caption{A two-dimensional section of the optimization landscape $F(\wt{\Phi}^*+c_1 v_1+c_2v_2)$. Here $\Phi^*$ is a global minimum of the target polynomial $f(x)=\frac{1}{\alpha}\left(T_4(x)+2T_2(x)+T_0(x)\right)$ with $\alpha = 440$, and $v_1 = \wt{\Phi}^1-\wt{\Phi}^*, v_2=\wt{\Phi}^2-\wt{\Phi}^*$, where $\wt{\Phi}^1$ and $\wt{\Phi}^2$ are two local minima found by numerical schemes. The objective values of the global minimum and local minima are annotated on the figure.}
    \label{fig:localmin_multi}
\end{figure}

\subsection{Numerical behavior near the maximal solution}
Here we demonstrate that among all global minima, the maximal solution is particularly desirable from the perspective of numerical optimization. 
Given an arbitrarily chosen polynomial $f(x)=\frac{1}{4}T_6(x)+\frac{5}{4}T_4(x)+\frac{1}{8}T_2(x)-T_0(x)$, we consider a sequence of scaled polynomials $f^{(k)}(x) := 10^{-k} f(x)$, so that $\lim_{k\to \infty}f^{(k)}(x)=0$.
All exact solutions to the optimization problem associated to $f^{(k)}(x)$ can be constructed explicitly via \cref{re:construct_method}. This procedure yields several sequences of phase factors, which converge to different limiting points as $k \to \infty$. For the specific degree-$6$ polynomial above, there are four distinct limits which is referred to as $\Phi_0^\text{class}$,
\begin{itemize}
    \item class 1: $\Phi^1_0=\left(\frac{\pi}{4},0,0,0,0,0,\frac{\pi}{4}\right),$
    \item class 2: $\Phi^2_0=\left(\frac{\pi}{4},0,\frac{\pi}{4},-\frac{\pi}{2},\frac{\pi}{4},0,\frac{\pi}{4}\right),$
    \item class 3: $\Phi^3_0=\left(\frac{\pi}{4},0,-\frac{\pi}{4},\frac{\pi}{2},-\frac{\pi}{4},0,\frac{\pi}{4}\right),$
    \item class 4: $\Phi^4_0=\left(\frac{\pi}{4},\frac{\pi}{4},0,-\frac{\pi}{2},0,\frac{\pi}{4},\frac{\pi}{4}\right).$
\end{itemize}
Specifically, the first class is associated with the maximal solution. The sequence of the set of phase factors corresponding to the scaled polynomial $f^{(k)}(x)$ and a given class of construction is referred to as $\Phi^\text{class}_*(k)$. \cref{fig:convergence_rate} shows that the class associated with the maximal solution distinguishes from other classes, in the sense that the convergence rate of $\Phi^\text{1}_*(k)$ towards $\Phi^\text{1}_0$ is much faster. 
This suggests that the convergence basin near $\Phi^\text{1}_0$ is much flatter, which justifies the choice of using $\Phi^\text{1}_0$ as the initial guess of the optimization.

To further test the performance of numerical optimization associated with different initial guesses, we use the limit point $\Phi_0^\text{class}$ as the initial guess and find the optimal phase factors $\Phi_*^\text{class}(k)$ generating the scaled polynomial \REVN{$\frac{1}{2}f(x)$} by running \textsf{QSPPACK} \cite{DongMengWhaleyEtAl2021}. The trajectory of the optimization process is displayed in \cref{fig:convergence_qsppack}. It shows that the optimization starting from $\Phi^\text{1}_0$ also has the fastest convergence rate.

\begin{figure}[htbp]
    \centering
    \includegraphics[width=.5\textwidth]{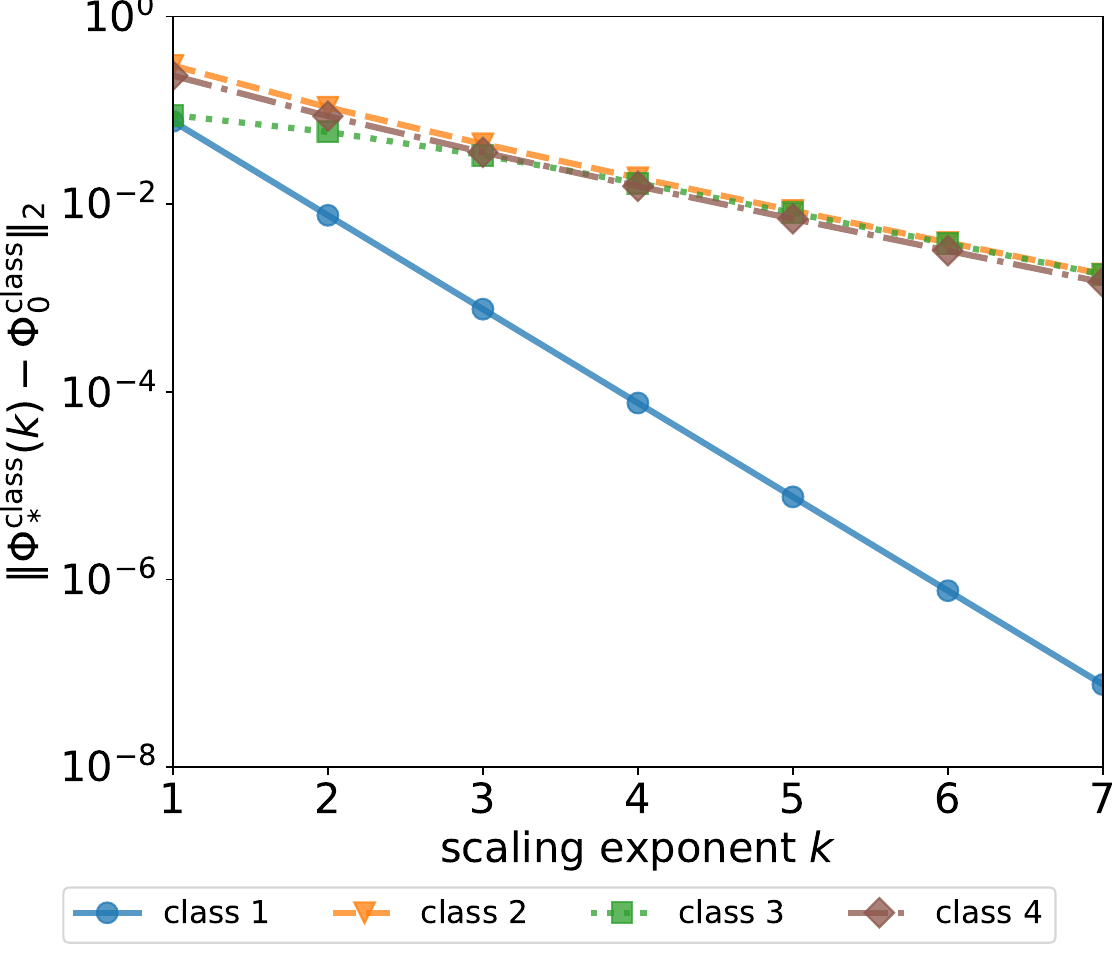}
    \caption{The convergence to the limit point of different class as the target polynomial being scaled down. Given a polynomial $f(x)$, the optimal set of phase factors $\Phi_*^\text{class}(k)$ parameterizes the scaled polynomial $f^{(k)}(x) := 10^{-k} f(x)$. Here, $\Phi^{\mathrm{class}}_0$ is the limit point of the corresponding class.}
    \label{fig:convergence_rate}
\end{figure}
\begin{figure}[htbp]
    \centering
    \includegraphics[width=.5\textwidth]{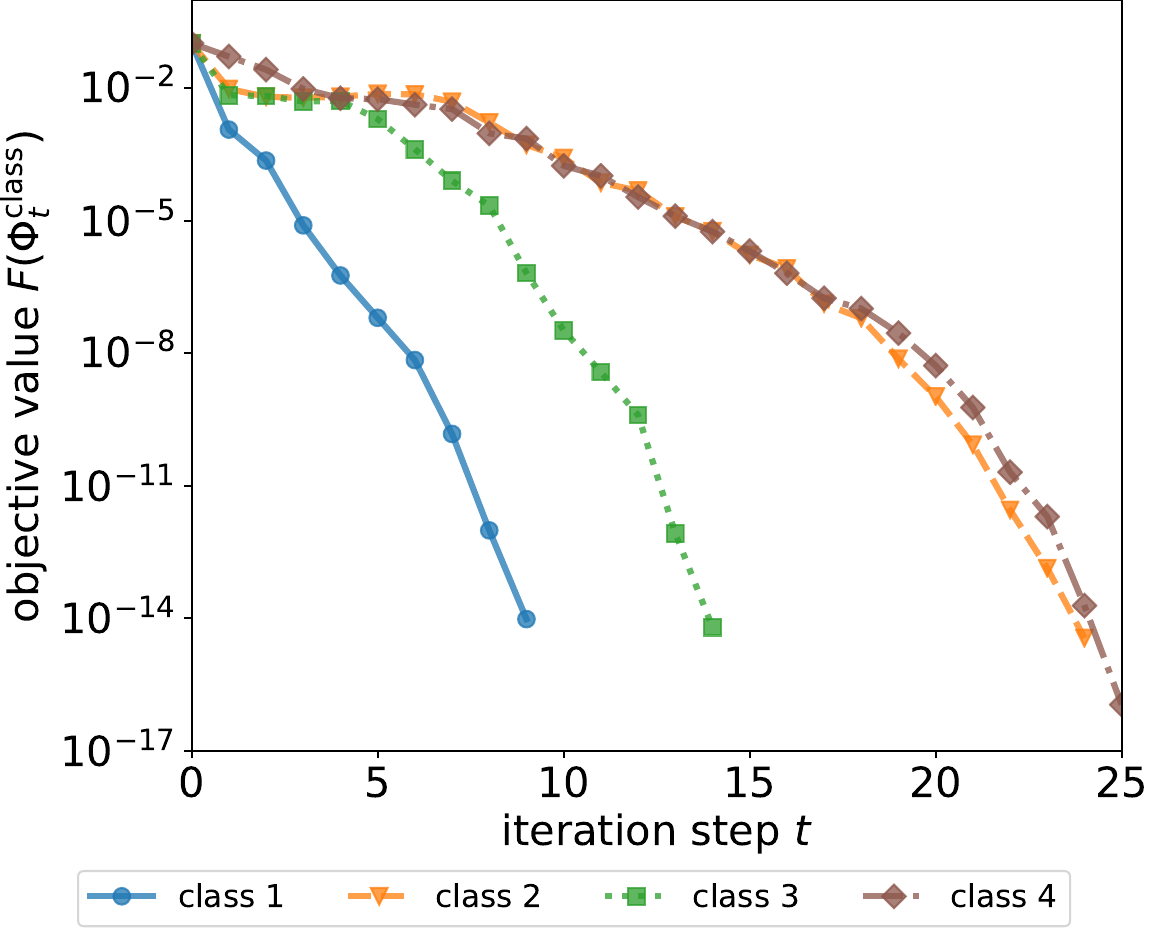}
    \caption{The objective value in each iteration step by running \textsf{QSPPACK} \REVN{ for target function $\frac{1}{2}f(x)$}. Initiating the optimization from the limit point of a class $\Phi_0^\text{class}$, $\Phi_{t}^\text{class}$ is the set of phase factors in the $t$-th optimization step. Different symbols correspond to the optimization trajectories starting from distinct classes of initial guess.}
    \label{fig:convergence_qsppack}
\end{figure}

\newpage
\appendix

\section{Proof of \cref{lma:leading_coef}}\label{sec:proof_leading_coef}
\begin{proof}
    One can check that
    \begin{equation}
    \begin{split}
        & P(x) = e^{\I\left(\phi_0+\phi_d\right)} \bra{0} W(x)\prod_{j=1}^{d-1} \left[e^{\I\phi_j Z} W(x)\right] \ket{0},\\
        & \I \sqrt{1-x^2} Q(x) = e^{\I\left(\phi_0-\phi_d\right) } \bra{0} W(x)\prod_{j=1}^{d-1} \left[e^{\I\phi_j Z} W(x)\right] \ket{1}.
    \end{split}
    \end{equation}
    Then, we expand the product of matrices from left to right by using $e^{\I \phi Z} = \cos(\phi) + \I \sin(\phi) Z$ and $W(x) Z = Z W(x)^{-1}$. We want to show that the term containing $\sin(\phi_j)$ does not contribute to $p_d$ and $q_{d-1}$. We use $l.o.$ to denote any polynomial or polynomial-valued matrix of degree $<d$. Here, we relax the definition of polynomial and classify the entry of the matrix of our interest as polynomial by considering $\sqrt{1-x^2}$ as another variable. Note that
    \begin{equation}
    \begin{split}
        &W(x)\prod_{j=1}^{d-1} \left[e^{\I\phi_j Z} W(x)\right] \\&= \cos\left(\phi_1\right) W(x)^2 \prod_{j=2}^{d-1} \left[e^{\I\phi_j Z} W(x)\right] + \I \sin\left(\phi_1\right) W(x) Z W(x) \prod_{j=2}^{d-1} \left[e^{\I\phi_j Z} W(x)\right]\\
        &= \cos\left(\phi_1\right) W(x)^2 \prod_{j=2}^{d-1} \left[e^{\I\phi_j Z} W(x)\right] + \I \sin\left(\phi_1\right) Z \prod_{j=2}^{d-1} \left[e^{\I\phi_j Z} W(x)\right]\\
        &= \cos\left(\phi_1\right) W(x)^2 \prod_{j=2}^{d-1} \left[e^{\I\phi_j Z} W(x)\right] + l.o.,
    \end{split}
    \end{equation}
    where the product of the second term is indeed a matrix whose elements are polynomials of degree $\leq d-2$.
    
    Using the fact that 
    \begin{equation}
        [W(x)]^d =\begin{pmatrix}
        T_d(x)  & \I \sqrt{1-x^2} U_{d-1}(x)\\
        \I \sqrt{1-x^2}U_{d-1}(x) & T_d(x)
        \end{pmatrix},
    \end{equation}
    we may prove inductively,
    \begin{equation}
        \begin{split}
            & P(x) = e^{\I\left(\phi_0+\phi_d\right)} \prod_{j=1}^{d-1} \cos\left(\phi_j\right) T_d(x) + l.o.\\
            & \I \sqrt{1-x^2} Q(x) = \I \sqrt{1-x^2} e^{\I\left(\phi_0-\phi_d\right)} \prod_{j=1}^{d-1} \cos\left(\phi_j\right) U_{d-1}(x) + l.o.,
        \end{split}
    \end{equation}
    which proves the lemma.
\end{proof}

\section{Proof of \cref{lma:leading_coef_sym}}\label{sec:proof_leading_coef_sym}
\begin{proof}
Now that $d$ is odd and $\Phi$ is symmetric, \cref{lma:leading_coef} indicates
\begin{equation*}
    q_{d-1} = \prod_{j=1}^{\wt{d}-1} \cos^2\left(\phi_j\right)\geq 0.
\end{equation*}
If $\phi_j\ne \frac{\pi}{2}+k\pi$ for all $1\leq j\leq \wt{d}-1$, where $k$ is some integer, then the leading Chebyshev coefficient of $Q$ is $q_{d-1}$, which is positive. Otherwise, there exists positive integer $j_0$ and integer $k$ such that $\phi_{j_0}= \phi_{d-j_0}= \frac{\pi}{2}+k\pi$. And either of the following equalities holds
\begin{equation*}
    e^{\I \phi_{j_0} Z}=e^{\I \phi_{d-j_0}Z}=e^{\I \frac{\pi}{2}Z},\quad e^{\I \phi_{j_0} Z}=e^{\I \phi_{d-j_0}Z}=e^{-\I \frac{\pi}{2}Z}.
\end{equation*}
Without losing generality, we only consider the cases that $\phi_{j_0}= \phi_{d-j_0}= \pm \frac{\pi}{2}$.

If $j_0<\wt{d}-1$, using the identity 
\begin{equation*}
    W(x) e^{\pm\I \frac{\pi}{2}Z} W(x) = e^{\pm\I \frac{\pi}{2}Z},
\end{equation*}
the product of matrices $$e^{\I \phi_{j_0-1}Z} W(x) e^{\I \phi_{j_0} Z} W(x) e^{\I \phi_{j_0+1}Z}$$
can be replaced by
$$e^{\I \left(\phi_{j_0-1} +\phi_{j_0} + \phi_{j_0+1} \right)Z}.$$ 
We can perform a similar reduction for for $$e^{\I \phi_{d-j_0-1}Z} W(x) e^{\I \phi_{d-j_0}Z} W(x) e^{\I \phi_{d-j_0+1}Z}.$$
Therefore, $U(x,\Phi)$ can be constructed by a new set of symmetric phase factors whose full length is $d-3$ and $\deg(Q)\leq d-5$. Apply \cref{lma:leading_coef} again and one has $q_{d-5}\geq 0$. 

If $j_0=\wt{d}-1$, the product of matrices $$e^{\I \phi_{\wt{d}-2}Z} W(x) e^{\I \phi_{\wt{d}-1} Z} W(x) e^{\I \phi_{\wt{d}}Z}W(x) e^{\I \phi_{\wt{d}+1}Z} $$
can be replaced by
$$- e^{\I \phi_{\wt{d}-2}Z} W(x) e^{\I \phi_{\wt{d}+1}Z}.$$ 
Therefore, $U(x,\Phi)$ can be constructed by a new set of symmetric phase factors whose full length is $d-1$ and $\deg(Q)\leq d-3$. Apply \cref{lma:leading_coef} again and one has $q_{d-3}\leq 0$.

Repeat the procedure above until the length of the new set of symmetric phase factors reaches $\text{deg} (Q)+2$, and it follows that the leading Chebyshev coefficient of $Q$ has the same sign as $\left(-1\right)^{\frac{d-1-\deg (Q)}{2}}$. 
\end{proof} 

\section{Proof of \cref{lma:newPQ_coef}}\label{sec:proof_newPQ_coef}
First, we consider the expansion of $P(x)$ and $Q(x)$ in the monomial basis,
\begin{equation*}
    P(x)=\sum_{i=0}^d \fp_i x^i, \quad\quad Q(x)=\sum_{j=0}^{d-1} \fq_j x^j.
\end{equation*}
and derive some useful equalities from the normalization condition, which comes from the unitarity of $U(x,\Phi)$.
\begin{lemma}\label{lma:normalization_property}
For any $P\in \CC[x],Q\in \RR[x]$ satisfying the condition (1)-(3) in \cref{thm:existandunique}, one can define $e^{2\I \phi_0}:=\frac{\fq_{d-1}}{\fp_{d}^*}$ and derive two useful equalities
\begin{equation}\label{eq:equality1}
\begin{split}
    & 2\Re[\fp_{d-2}e^{-2\I\phi_0}]+\fq_{d-1} -2\fq_{d-3}=0,\quad \forall d\geq 2,
\end{split}
\end{equation}
and 
\begin{equation}\label{eq:equality2}
    2\Re[\fp_{d-4}e^{-2\I\phi_0}]+\fq_{d-3}-2\fq_{d-5}=-\frac{\fq_{d-1}}{4}-\frac{1}{\fq_{d-1}}\abs{\Im[\fp_{d-2}e^{-2\I\phi_0}]}^2,\forall d\geq 3.
\end{equation}
Here we use the convention that $\fq_k=\fp_k=0$ if $k<0$.
\end{lemma}
\begin{proof}
The normalization condition implies that $\abs{\fp_d}=\abs{\fq_{d-1}}$, so the complex argument $\phi_0$ is well defined. Since $Q(x)$ is real, $\fp_d e^{-2i\phi_0}=\fp^*_d e^{2i\phi_0}= \fq_{d-1}\in \RR$ and the normalization condition becomes
\begin{equation}\label{eqn:normalization_real}
e^{-2i\phi_0} P(x)\cdot e^{2i\phi_0} P^* (x) + (1-x^2)Q(x)^2 =1.
\end{equation}
For $d\geq 2$, the coefficient of $x^{2d-2}$ in \cref{eqn:normalization_real} is 
\begin{equation*}
\begin{split}
0&=\fp_d e^{-2i\phi_0}\fp_{d-2}^* e^{2i\phi_0}+\fp_d^* e^{2i\phi_0}\fp_{d-2} e^{-2i\phi_0}+\fq_{d-1}^2- 2\fq_{d-1}\fq_{d-3}\\
&= \fq_{d-1}\left(\fp_{d-2}^* e^{2i\phi_0}+\fp_{d-2} e^{-2i\phi_0}+\fq_{d-1}-2\fq_{d-3}\right)\\
\Longrightarrow \quad & 2\Re[\fp_{d-2}e^{-2\I\phi_0}] + \fq_{d-1} -2\fq_{d-3}=0.
\end{split}
\end{equation*}
For $d\geq 3$, the coefficients of $x^{2d-4}$ in \cref{eqn:normalization_real} is
\begin{equation*}
\begin{split}
0&=\fp_d e^{-2i\phi_0}\fp_{d-4}^* e^{2i\phi_0}+\fp_d^* e^{2i\phi_0}\fp_{d-4} e^{-2i\phi_0}+\fp_{d-2}^* e^{2i\phi_0}\fp_{d-2} e^{-2i\phi_0}\\
&\quad+2\fq_{d-1}\fq_{d-3}-\fq_{d-3}^2- 2\fq_{d-1}\fq_{d-5}\\
&= \fq_{d-1}\left(2\Re[\fp_{d-4}e^{-2\I\phi_0}]+\fq_{d-3}-2\fq_{d-5}\right)-\fq_{d-3}^3+\fq_{d-1}\fq_{d-3}+\abs{\fp_{d-2}e^{-2\I\phi_0}}^2\\
&= \fq_{d-1}\left(2\Re[\fp_{d-4}e^{-2\I\phi_0}]+\fq_{d-3}-2\fq_{d-5}\right)-\fq_{d-3}^3+\fq_{d-1}\fq_{d-3}\\
&\quad +\abs{\Re[\fp_{d-2}e^{-2\I\phi_0}]}^2 +\abs{\Im[\fp_{d-2}e^{-2\I\phi_0}]}^2\\
&= \fq_{d-1}\left(2\Re[\fp_{d-2}e^{-2\I\phi_0}]+\fq_{d-3}-2\fq_{d-5}\right)-\fq_{d-3}^3+\fq_{d-1}\fq_{d-3}\\
&\quad +\left(\frac{2\fq_{d-3}-\fq_{d-1}}{2}\right)^2+\abs{\Im[\fp_{d-2}e^{-2\I\phi_0}]}^2\\
&= \fq_{d-1}\left(2\Re[\fp_{d-2}e^{-2\I\phi_0}]+\fq_{d-3}-2\fq_{d-5}\right)+\frac{\fq_{d-1}^2}{4}+\abs{\Im[\fp_{d-2}e^{-2\I\phi_0}]}^2\\
\Longrightarrow \quad & 2\Re[\fp_{d-2}e^{-2\I\phi_0}]+\fq_{d-3}-2\fq_{d-5}=-\frac{\fq_{d-1}}{4}-\frac{1}{\fq_{d-1}}\abs{\Im[\fp_{d-2}e^{-2\I\phi_0}]}^2,
\end{split}
\end{equation*}
where we exploit \cref{eq:equality1} in the fourth equality.
\end{proof}

\begin{proof}[Proof of \cref{lma:newPQ_coef}]
Direct computation shows that 
\begin{equation}\label{eq:tiledP}
P^{(1)} = 2(1-x^2)x Q-(1-x^2) P^* e^{2i \phi_0} + x^2 P e^{-2i \phi_0},
\end{equation}
and 
\begin{equation}\label{eq:tiledQ}
Q^{(1)} = (2x^2-1) Q - x P^* e^{2i \phi_0} - x P e^{-2i \phi_0}=(2x^2-1) Q - 2x\Re[ P^* e^{2i \phi_0}].
\end{equation}
Hence, $Q(x)\in \RR[x]$. Normalization condition is preserved due to unitarity. We can also verify that the parity of $P^{(1)}$ is the same as $P$ and the parity of $Q^{(1)}$ is the same as $Q$.

Now we  apply the results of \cref{lma:normalization_property} and first examine the coefficients of $P^{(1)}$. The coefficient of $x^{d+2}$ is
\begin{equation*}
\fp^{(1)}_{d+2}= -2 \fq_{d-1} + \fp_d^* e^{2i\phi_0} + \fp_d e^{-2i \phi_0} =0.
\end{equation*}
If $d\geq 2$, the coefficient of $x^{d}$ is 
\begin{equation*}
\begin{split}
\fp^{(1)}_d=&2\fq_{d-1} -2 \fq_{d-3} + \fp_{d-2}^* e^{2i\phi_0} + \fp_{d-2} e^{-2i\phi_0} -\fp_d^* e^{2i\phi_0}\\
=&\left(2\Re[\fp_{d-2}e^{-2\I\phi_0}] + \fq_{d-1} -2\fq_{d-3}\right) + \left(\fq_{d-1}-\fp_d^* e^{2i\phi_0}\right)=0.
\end{split}
\end{equation*}
If $d=2$, $P^{(1)}=-\fp_0^* e^{2\I \phi_0}=-\left(\frac{\fp_0}{\fp_2}\right)^*\fq_1$. And if $d\geq 3$, the coefficient of $x^{d-2}$ is 
\begin{equation*}
    \begin{split}
        \fp^{(1)}_{d-2}& =2\fq_{d-3}-2\fq_{d-5}-\fp_{d-2}^* e^{2\I \phi_0}+\fp^*_{d-4} e^{2\I \phi_0}+\fp_{d-4} e^{-2\I \phi_0}\\
        &= \fq_{d-3}-\fp_{d-2}^* e^{2\I \phi_0}-\frac{\fq_{d-1}}{4}-\frac{1}{\fq_{d-1}}\abs{\Im[\fp_{d-2}e^{-2\I\phi_0}]}^2\\
        & = \frac{\fq_{d-1}}{4}-\frac{1}{\fq_{d-1}}\abs{\Im[\fp_{d-2}e^{-2\I\phi_0}]}^2+ \I \Im [\fp_{d-2}e^{-2\I\phi_0}]\\
    \end{split}
\end{equation*}
where we use \cref{eq:equality1} in the second equality and \cref{eq:equality2} in the third equality.

Then we also apply the results of \cref{lma:normalization_property} to examine the coefficients of $Q^{(1)}$. The coefficient of $x^{d+1}$ is 
\begin{equation*}
\fq^{(1)}_{d+1}=2 \fq_{d-1} -\fp_d e^{-2i\phi_0} -\fp_d^* e^{2i \phi_0} =0.
\end{equation*} 
If $d\geq 2$, the coefficient of $x^{d-1}$ is 
\begin{equation*}
\fq^{(1)}_{d-1}=2 \fq_{d-3} -\fq_{d-1} -\fp_{d-2} e^{-2i\phi_0} -\fp_{d-2}^* e^{2i \phi_0}=0.
\end{equation*}
If $d=2$, $Q^{(1)}=0$. And if $d\geq 3$, the coefficient of $x^{d-3}$ is 
\begin{equation*}
\begin{split}
    \fq^{(1)}_{d-3}&=2 \fq_{d-5} -\fq_{d-3} -\fp_{d-4}^* e^{2i\phi_0}-\fp_{d-4} e^{-2i\phi_0}\\
    &=\frac{\fq_{d-1}}{4}+\frac{1}{\fq_{d-1}}\abs{\Im[\fp_{d-2}e^{-2\I\phi_0}]}^2,
\end{split}
\end{equation*}
where we exploit \cref{eq:equality2}. Furthermore, we get that $\fq^{(1)}_{d-3}$ is positive if $d$ is odd.

\end{proof}

\end{document}